\documentclass[a4paper,12pt,twoside]{report}

\usepackage{appendix} 
\usepackage{graphicx} 
\usepackage{etoolbox} 
\usepackage{chngcntr} 
\usepackage{amssymb} 
\usepackage{amsmath} 
\usepackage{mathtools} 
\usepackage{upgreek} 
\usepackage{amsthm} 
\usepackage{setspace} 
\usepackage[OMLmathsfit]{isomath} 

\DeclareFontFamily{U}  {MnSymbolC}{}
\DeclareFontShape{U}{MnSymbolC}{m}{n}{
    <-6>  MnSymbolC5
   <6-7>  MnSymbolC6
   <7-8>  MnSymbolC7
   <8-9>  MnSymbolC8
   <9-10> MnSymbolC9
  <10-12> MnSymbolC10
  <12->   MnSymbolC12}{}
\DeclareFontShape{U}{MnSymbolC}{b}{n}{
    <-6>  MnSymbolC-Bold5
   <6-7>  MnSymbolC-Bold6
   <7-8>  MnSymbolC-Bold7
   <8-9>  MnSymbolC-Bold8
   <9-10> MnSymbolC-Bold9
  <10-12> MnSymbolC-Bold10
  <12->   MnSymbolC-Bold12}{}
\DeclareSymbolFont{MnSyC}         {U}  {MnSymbolC}{m}{n}
\DeclareMathSymbol{\mnmedcircle}{\mathbin}{MnSyC}{90}
\DeclareMathSymbol{\mncircledcirc}{\mathbin}{MnSyC}{99}
\DeclareMathSymbol{\mnmeddiamond}{\mathbin}{MnSyC}{110}
\DeclareMathSymbol{\mndiamonddiamond}{\mathbin}{MnSyC}{127}

\usepackage[margin=1in,includefoot,includehead,bindingoffset=1cm]{geometry}

\allowdisplaybreaks

\usepackage[symbol,perpage]{footmisc} 
\usepackage{fnbreak} 
\DefineFNsymbolsTM*{mysymbols}{ 
  \textasteriskcentered *
  \textsection          \mathsection
  \textparagraph        \mathparagraph
  \textbardbl           \|%
  {\textasteriskcentered\textasteriskcentered}{**}
}
\setfnsymbol{mysymbols}

\makeatletter
\patchcmd{\@chapter}{\addtocontents{lof}{\protect\addvspace{10\p@}}}{}%
{}{\message{WARNING: Failed to patch \string\@chapter}}
\makeatother

\newcommand{\cleardoublepageemptyheadings}{\clearpage\pagestyle{empty}\cleardoublepage\pagestyle{headings}}

\usepackage[
  bookmarks=true,
  bookmarksnumbered=true,
  bookmarksopen=true,
  pdfborder={0 0 0},
]{hyperref}
\usepackage[all]{hypcap} 

\hypersetup{
  pdfauthor      = {Yannai A. Gonczarowski <yannai@gonch.name>},
  pdftitle       = {Timely Coordination in a Multi-Agent System},
  pdfsubject     = {Master's thesis. Advisors: Prof. Gil Kalai and Prof. Yoram Moses.},
}

\DeclareGraphicsExtensions{.pdf}
\graphicspath{{figures/}}
\newcommand{\fig}[3]{\begin{figure}[ht]
   \centering\includegraphics[width=14cm]{#1}
   \caption[#2]{#3}
   \label{fig:#1}
\end{figure}}
\newcommand{\figref}[1]{Figure~\ref{fig:#1}}
\counterwithout{figure}{chapter} 
\usepackage[font={small,it},margin=10pt]{caption} 

\newtheorem{thm}{Theorem}[chapter]
\newcommand{\thmref}[1]{Theorem~\ref{thm:#1}}
\newtheorem{defn}[thm]{Definition}
\newcommand{\defnref}[1]{Definition~\ref{defn:#1}}
\newtheorem{lemma}[thm]{Lemma}
\newcommand{\lemmaref}[1]{Lemma~\ref{lemma:#1}}
\newtheorem{cor}[thm]{Corollary}
\newcommand{\corref}[1]{Corollary~\ref{cor:#1}}
\newtheorem{ex}[thm]{Example}
\newcommand{\exref}[1]{Example~\ref{ex:#1}}
\newtheorem{claim}[thm]{Claim}
\newcommand{\claimref}[1]{Claim~\ref{claim:#1}}
\newtheorem{remark}[thm]{Remark}
\newcommand{\remarkref}[1]{Remark~\ref{remark:#1}}

\newcommand{\chapterref}[1]{Chapter~\ref{chapter:#1}}
\newcommand{\sectionref}[1]{Section~\ref{section:#1}}
\newcommand{\appref}[1]{Appendix~\ref{app:#1}}

\newcommand{\tuple}[1]{\bar{#1}}
\newcommand{\tuplefull}[1]{\tuple{#1}=(#1_m)_{m=1}^n}
\newcommand{\functiondefn}[5]{\begin{align*}
#1: #2 &\rightarrow #3 \\* #4 &\mapsto #5
\end{align*}}
\newcommand{\eqdef}{\triangleq}

\newcommand{\distinctpairs}[1]{#1^{\bar{2}}}
\newcommand{\apath}{\tuple{p}}
\newcommand{\apathfull}{\tuplefull{p}}
\newcommand{\paths}[1]{\mathcal{P}(#1)}
\newcommand{\gpaths}{\paths{G}}
\newcommand{\glength}{L_G}

\newcommand{\agents}{\mathbb{I}_{\gamma}}
\newcommand{\contextgraph}{\mathcal{G}_{\gamma}}
\newcommand{\externalinputs}{\tilde{E}_{\gamma}}
\newcommand{\contextbounds}{b_{\gamma}}
\newcommand{\contextneighbours}{N_{\gamma}}
\newcommand{\context}{(\contextgraph=(\agents,\contextneighbours,\contextbounds), (S_i)_{i \in \agents}, \externalinputs, (i_{\tilde{e}})_{\tilde{e} \in \externalinputs})}

\newcommand{\timeset}{\mathbb{T}}
\newcommand{\messages}{\mathcal{M}}
\newcommand{\protocols}{\mathbb{P}_{\gamma}}
\newcommand{\runs}{\mathcal{R}_{\gamma}}

\newcommand{\RP}{R_{\gamma}(P)}
\newcommand{\RPTAG}{R_{\gamma}(P')}
\newcommand{\RPZERO}{R_{\gamma}(P_0)}
\newcommand{\ND}{\mathit{ND}_{\gamma}}
\newcommand{\TR}[1]{\mathit{R}_{\gamma}^{\tilde{e}}(#1)}
\newcommand{\TRP}{\TR{P}}

\newcommand{\ER}[1]{\mathit{ER}_{\gamma}\langle #1 \rangle}
\newcommand{\implspec}{(I,\delta)}
\newcommand{\TCRspec}{(\gamma,\tilde{e},I,\delta)}
\newcommand{\TCR}[1]{\mathit{TCR}_{\gamma}\langle #1 \rangle}
\newcommand{\TCRTAG}[1]{\mathit{TCR'}_{\gamma}\langle #1 \rangle}

\newcommand{\OR}[1]{\mathit{OR}_{\gamma}\langle #1 \rangle}
\newcommand{\SR}[1]{\mathit{SR}_{\gamma}\langle #1 \rangle}
\newcommand{\OJR}[1]{\mathit{OJR}_{\gamma}\langle #1 \rangle}
\newcommand{\WTR}[1]{\mathit{WTR}_{\gamma}\langle #1 \rangle}
\newcommand{\TTR}[1]{\mathit{TTR}_{\gamma}\langle #1 \rangle}

\newcommand{\timpl}{\mathsans{t}}
\newcommand{\ttildeimpl}{\mathsans{\tilde{t}}}
\newcommand{\dpaths}{\paths{G_{\delta}}}
\newcommand{\dhatpaths}{\paths{G_{\hat{\delta}}}}
\newcommand{\dlength}{L_{G_{\delta}}}
\newcommand{\dtaglength}{L_{G_{\delta'}}}
\newcommand{\dtagtaglength}{L_{G_{\delta''}}}

\newcommand{\syncausal}[1]{\overset{\gamma}{\underset{\vphantom{#1}\smash{\raisebox{.3em}{$\scriptstyle #1$}}}{\rightsquigarrow}}}
\newcommand{\notsyncausal}[1]{\hspace{.08em}\not\hspace{-.08em}\syncausal{#1}}
\newcommand{\boundguarantee}{\overset{\gamma}{\dashrightarrow}}
\newcommand{\notboundguarantee}{\,\,\not\!\!\boundguarantee}

\newcommand{\PND}[1]{\mathit{PND}_{\gamma}^{#1}}
\newcommand{\RND}{\mathit{RND}_{\gamma}^P}

\newcommand{\tcap}{\overset{t}{\cap}}
\newcommand{\tprimecap}{\overset{t'}{\cap}}
\newcommand{\tsubseteq}{\overset{t}{\subseteq}}

\newcommand{\points}{\Omega_{R}}
\newcommand{\protocolpoints}{\Omega_{\RP}}
\newcommand{\pointsets}{\mathcal{F}_{R}}
\newcommand{\protocolpointsets}{\mathcal{F}_{\RP}}
\newcommand{\ensemble}{\mathbf{e}}
\newcommand{\rensemble}{\mathbf{r}}

\newcommand{\nolaterthan}[1]{\mncircledcirc^{\le{#1}}}
\newcommand{\sometime}{\mndiamonddiamond}
\newcommand{\atexactly}[1]{\mncircledcirc^{#1}}

\newcommand{\dcspec}{(\gamma,R,I,\delta)}
\newcommand{\dck}{C_I^{\delta}}
\newcommand{\dckg}{\mbox{\it \c{C}\,}_{\!I}^{\delta}}

\begin{document}

\pagestyle{empty}
\newgeometry{margin=1in,bindingoffset=1cm} 

\pagenumbering{roman}


\message{Title Page}
\pdfbookmark[1]{Title Page}{title}
\begin{titlepage}

\begin{center}

\textsc{\LARGE
The Hebrew University of Jerusalem
\\
Faculty of Science
\\
Einstein Institute of Mathematics\\[2.75cm]}

{\Huge\bf{
Timely Coordination in a Multi-Agent System}
\\[2cm]
}

{\Large A research thesis submitted in partial fulfillment of the requirements
for the degree of Master of Science\\[2cm]}

{\Large Author: \\}
{\LARGE \href{mailto:yannai@gonch.name}{Yannai A. Gonczarowski}\\[2cm]}

{\Large Thesis Advisors:\\[0.45cm]}
\begin{minipage}[t]{0.4\textwidth}
\begin{flushleft}
\LARGE Prof.\ Gil Kalai
\end{flushleft}
\end{minipage}
\begin{minipage}[t]{0.45\textwidth}
\begin{flushright}
{\LARGE Prof.\ Yoram Moses} \\
\begin{singlespace}{\small Department of Electrical Engineering\\
Technion, Israel Institute of Technology}
\end{singlespace}
\end{flushright}
\end{minipage}

\vfill
{\Large March 2012}

\end{center}
\end{titlepage}

\renewcommand{\titlepage}{} 

\cleardoublepage


\message{Dedication}
\pdfbookmark[1]{Dedication}{dedication}
\vspace*{4cm}
\begin{flushright}
\begin{Large}
To my parents, who taught me that knowledge is the one thing that no one can
ever take away from me,
\\[2cm]
to my brothers, who know everything,
\\[2cm]
and most of all, to Elee, the stable event of my life; the fixed point of my life.\\
\begin{footnotesize}
(This work would most probably never have been written were it not
for her Kyoto adventure.)
\end{footnotesize}
\\
\end{Large}
\end{flushright}

\cleardoublepage


\let\realabstractname=\abstractname
\renewcommand{\abstractname}{Acknowledgements}
\message{\abstractname}
\pdfbookmark[1]{\abstractname}{acknowledgements}
\begin{abstract}
I would like to thank my advisor, Prof.\ Gil Kalai, for his friendship
and support throughout my studies. I would like to thank my second
advisor, Prof.\ Yoram Moses, for his friendship, for introducing me to
the fascinating subject of knowledge, and for his patience while I bombard
him with emails during the small hours of the night.

In addition, I would like to thank my friends, my family, and last but not
least, my other half Elee, for their support, their encouragement and
their friendship, and simply for being there.
\end{abstract}
\let\abstractname=\realabstractname 

\cleardoublepage


\message{\abstractname}
\pdfbookmark[1]{Abstract}{\abstractname}

\begin{abstract}
In a distributed algorithm, multiple processes, or agents, work toward a common
goal. More often than not, the actions of some agents
are dependent on the previous execution (if not also on the outcome) of the
actions of other agents. The resulting interdependencies between the timings
of the actions of the various agents
give rise to the study of methods for timely coordination of these actions.

In this work, we formulate and mathematically analyze a novel multi-agent
coordination problem, which we call ``Timely-Coordinated Response'', and in
which the time difference between each pair of actions may be constrained
by upper and/or lower bounds.
This problem generalizes some classic coordination
problems formulated and studied by Halpern and Moses,
and some coordination problems recently
formulated and studied by Ben-Zvi and Moses.

We optimally solve (i.e.\ provide an optimal protocol for solving)
the timely-coordinated response problem in two ways: one
using a generalization of the fixed-point approach of Halpern and Moses, and
one using a generalization of the synchronous causality (``syncausality'')
approach of Ben-Zvi
and Moses. Furthermore, we constructively show the equivalence of the solutions
yielded by both approaches, despite the vast conceptual differences between
them.
By combining both approaches, we derive strengthened versions of
known results for some previously-defined special cases of this problem.

Our analysis is conducted under minimal assumptions: we work in a
continuous-time model with possibly infinitely many agents. The general results
we obtain for this model reduce to stronger results for discrete-time models
with only finitely many agents. In order to distill the properties of
such models that are significant to this reduction, we define several novel
classes of naturally-occurring models, all generalizing discrete-time models
with finitely many agents, which in a sense separate the different results.
We investigate the timely-coordinated response problem in these models, and
present both a more practical optimal solution for the problem, 
as well as a surprisingly simple condition for solvability thereof,
for these models.

To conclude this work, we show how our results for the timely-coordinated
response problem generalize the results known for
previously-studied special cases of this problem, and present some open
questions and further research directions.
\end{abstract}

\cleardoublepage
 
\restoregeometry
\pagestyle{headings}


\message{\contentsname}
\pdfbookmark[1]{\contentsname}{toc}
\setcounter{tocdepth}{3}
\tableofcontents
\clearpage

\message{\listfigurename}
\pdfbookmark[1]{\listfigurename}{lof}
\listoffigures

\cleardoublepageemptyheadings


\pagenumbering{arabic}

\chapter{Introduction}\label{chapter:introduction}

In a distributed algorithm, multiple processes, or agents, work toward a common
goal. More often than not, the actions of some agents
are dependent on the previous execution (if not on the outcome) of the actions
of other agents. This introduces interdependencies between the timing
of the actions of the various agents.

\section{An Informal Example}

We begin with a simple example that illustrates how the coordination
problem underlying this work arises as a natural, albeit nontrivial,
continuation of previously studied problems.

\begin{ex}\label{ex:acme}
Consider ACME, an IT company providing on-line storage services.\footnote{
Some readers may be acquainted with other products provided by \href{http://en.wikipedia.org/wiki/Acme_Corporation}{ACME}, such
as high-tech warfare gear used by \href{http://en.wikipedia.org/wiki/Wile_E._Coyote_and_Road_Runner}{Wile E.\ Coyote} in his endless quest to
capture the \href{http://en.wikipedia.org/wiki/Wile_E._Coyote_and_Road_Runner}{Road Runner}, or be acquainted with the ACME detective agency on
the hunt for V.I.L.E. ringleader and former ACME agent \href{http://en.wikipedia.org/wiki/Carmen_Sandiego}{Carmen Sandiego}.
As both these venues have become less lucrative in recent years, ACME
decided to follow the trend and go into IT.}
When ACME's on-line storage service is founded, its user base is relatively
small, and one server fulfills all of the requirements of this service.
To be on the safe side, though, ACME operates
a backup server. Being fairly confident in the stability of its main server,
ACME does not impose any freshness constraints on the backup server,
other than that the backup server must never be ahead of the main server,
in order to avoid
even the slightest potential of a data change being reflected solely in the
backup server. Thus, we may
concisely capture the only timing constraint ACME imposes on its servers:
If a user changes her data, then eventually the main server reflects this
change, and eventually, at some later time, the backup server reflects it
as well. A generalization of the problem underlying
such a scenario was studied by Lamport\cite{lamport-causality} in an
asynchronous model. Recently, Ben Zvi and Moses\cite{bzm1,bzm2,bzm3} extended
this study to synchronous models as well, dubbing the generalized problem
``ordered response''.

We return to ACME's story.
After a while, the user base of the company's storage service grows and
moreover, many
users use it more heavily than before, as they have grown both
accustomed to it, and confident of its abilities and stability.
Eventually, a single-server
solution becomes inadequate for this service, and ACME turns
its backup server into a second live server. Optimally, ACME would like
to impose the following timing constraint on its servers:
If a user changes her data, then eventually both servers reflect this change,
and they do so {\em simultaneously}. A generalization of the problem
underlying such a scenario
has been extensively studied\cite{firing-squad1,firing-squad2} under the name
``firing squad''. In particular,
it was studied both by Ben-Zvi and Moses\cite{bzm1,bzm2,bzm3}, who dub it 
``simultaneous response'', and by Halpern and Moses\cite{halpern-moses-1990,book},
who call it ``perfect coordination''.
Unfortunately for ACME, though, it is shown
in \cite{halpern-moses-1990} that this problem is unsolvable under realistic
conditions.

Having read \cite{halpern-moses-1990}, ACME decides to go for what
its engineers perceive as ``the next best thing'', replacing the requirement for
simultaneous reflection of a change in both servers to ``almost simultaneous''
reflection. Formally, they demand that if any server reflects a change at any
time, then the other server reflects this change no later than 100 milliseconds
thereafter. The problem underlying such a scenario no longer falls within the
scope of the study of Ben-Zvi and Moses\cite{bzm1,bzm2,bzm3,bzm4}, although
a generalization thereof was studied by Halpern and
Moses\cite{halpern-moses-1990,book} under the name ``$\upvarepsilon$\mbox{-}coordination''.

Naturally, as long as this ``near simultaneity'' constraint (or, in the
preceding scenarios, the relevant timing constraint introduced there)
is met, ACME wish
for both servers to reflect each user action as close to the time of its
occurrence as possible.\footnote{
We later phrase both the worst-case response time
in each of the above scenarios, as well as a necessary and sufficient condition
for the sheer solvability of this problem, as functions of the topology of
ACME's network, and of the worst-case communication lag times in it.
}

Shortly after the switch to two live servers, an email reaches ACME's
headquarters. ACME's on-line gaming subsidiary, which uses ACME's on-line
storage infrastructure to store the state of their on-line multi-player games,
complains that the slow response time of ACME's servers, coupled with a
100-millisecond lag between these servers, renders its multi-player games
unplayable. ACME engineers convene for an emergency meeting, and propose the
following plan: as most of ACME's gaming customers are located in
the same vicinity (ACME's gaming platform is very popular in Israel),
from which traffic to
ACME's server \#2 is particularly fast (server \#2 resides in Israel, while
server \#1 resides in the U.S.A.), they can simply route all the gaming
traffic to server \#2, effectively eliminating the 100-millisecond lag the
gaming subsidiary is complaining about. Unfortunately, this does not solve the
other cause of this complaint: the slow response time of ACME's servers.
(The performance price of ACME's new algorithm, which coordinates a maximum
freshness lag of 100 milliseconds between the servers, is a slower response
time of both servers than the one achieved in the old single-live-server
algorithm.)
One of the engineers raises the following question: perhaps they can achieve
a faster response time for server \#2 if the timing constraint is revised
to be asymmetric: server \#2 must update no later than 100 milliseconds after
server \#1, however server \#1 may update as late as 300 milliseconds after
server \#2. The problem underlying such a scenario falls out of the scope
both of the studies of Ben-Zvi and Moses\cite{bzm1,bzm2,bzm3,bzm4},
and the studies of Halpern and Moses\cite{halpern-moses-1990} and their
extension by Fagin et al.\cite[Section~11.6]{book}.
\end{ex}

The multi-agent coordination problem that we present and analyze in this thesis
generalizes, among others, the problems arising in the above
example. In particular, the problem that we present generalizes the last problem
arising from this example, by allowing to arbitrarily bound the time difference
between each pair of actions both from above and from below.
This generalizes the study of Halpern and Moses\cite{halpern-moses-1990}
(and of Fagin et al.\cite[Section~11.6]{book})
by allowing different bounds to be specified for different pairs of actions,
and generalizes the study of Ben-Zvi and Moses\cite{bzm1,bzm2,bzm3,bzm4} by
allowing the specification of an upper and a lower bound that do not coincide,
on the time difference between a pair of actions.

\section{Overview}

In this work, we present and mathematically analyze a novel multi-agent
coordination problem, which we call ``timely-coordinated response''.
In this coordination problem, which we define and analyze in a synchronous
model, a set of agents are to perform local actions,
and the time difference between each pair of actions
may be constrained by an upper and/or a lower bound (or neither), which are
parameters given as part of the problem description. Following the studies of
Halpern and Moses\cite{halpern-moses-1990}, Fagin et al.\cite{book},
and Ben-Zvi and Moses\cite{bzm1,bzm2,bzm3,bzm4}, which we generalize, most of
this work revolves around the interaction between time and coordination.

After presenting the timely-coordinated response problem in
\chapterref{tcr-exhibition}, we perform, in \chapterref{delta-analysis},
a graph-theoretical analysis of the set of constraining parameters (upper and
lower bounds on the time difference, for each pair of actions) that define
this problem. This analysis leads to a definition of a canonical representative
for each class of constraint-sets that define the same problem, and to a
characterisation for solvability of the timely-coordinated response problem
under what may be regarded as ideal conditions.

In the following two chapters,
we optimally solve the timely-coordinated response problem in two ways,
each generalizing one of the approaches previously used to analyze some
special cases thereof:
In \chapterref{syncausality-approach}, we survey, and the generalize,
the ``syncausality'' approach of Ben-Zvi and Moses, which may be viewed
as more of a concrete ``nuts and bolts'' approach. In this chapter, which
is combinatorial in character, we study the timely-coordinated response
problem in the possible presence of guarantees on message delivery times
between agents. (In \chapterref{practical}, we present a result showing the
impossibility of timely coordination using mutual constraints in the absence of
such bounds.)
In \chapterref{fixed-point-approach}, we survey, and then generalize,
the ``fixed-point'' approach of Halpern and Moses, which was later studied by
Fagin et al.\ as well.
This approach, whose origins are traceable to temporal logic, may be conversely
viewed as more of an abstract ``higher level'' approach.
The main result presented in each of these two chapters is a
description of an optimal protocol/algorithm for solving the timely-coordinated
response problem.
Following these analyses, we constructively show, in \chapterref{equivalence},
that, despite the significant conceptual and technical differences between
these two approaches, they both yield equivalent solutions for the
timely-coordinated response problem.

The above-surveyed analysis is conducted under minimal assumptions:
it applies to a continuous-time model, which may contain infinitely many
agents. The general results obtained using this analysis reduce to
stronger results when the model in question is a discrete-time one, and
contains only finitely many agents.
In \chapterref{practical}, we define several novel classes of
naturally-occurring models, all of which generalize discrete-time models that
contain finitely many agents. These classes of models, in a sense, separate
the specialized discrete-time finite-agent results from the generic
continuous-time infinite-agent results. We investigate the timely-coordinated
response problem in these
models, and derive, for these models, both a more practical description of the
optimal solution for this problem, as well as a surprisingly simple condition
for solvability thereof, in terms of the available network communication
channels and the worst-case delivery times therein.
We conclude this chapter by
combining both approaches to derive a strengthened version of
a known impossibility result for some previously-defined special cases of the
timely-coordinated response problem.

Following the above analysis, we show, in \chapterref{previous}, how the
results obtained in the previous chapters reduce to generalizations of the
known best results for previously-studied special cases of the
timely-coordinated response problem.

Finally, in \chapterref{open-questions}, we qualitatively discuss some of our
results, and present some open questions and
some future research directions in which the results of this work may prove
to be useful.

The main contributions of this work are:
\begin{enumerate}
\item Identifying, defining and analyzing the timely-coordinated response
problem.
\item Applying both above-described approaches to analyze this problem,
thereby unifying,
generalizing and strengthening results previously achieved using these
approaches.
\item Deriving generic results for continuous-time models, as well as
specialized results for discrete-time models, and defining ``intermediate''
model classes which, in a sense, separate the continuous-time results from
the discrete-time results.
\end{enumerate}

Underlying this work are three different currents. While these are
interconnected, each of these may stand alone in its own right, and may
be of interest to a different audience:
\begin{itemize}
\item Our graph-theoretic analysis from \chapterref{delta-analysis} may
be of most interest to combinatorists.
\item Our generalization of the fixed-point approach and its results may be
of most interest to logicians and game theorists.
\item Our generalization of the syncausality approach and its results may be
of most interest to computer scientists and engineers.
\end{itemize}

\chapter{Notation}\label{chapter:notation}

Throughout this work, we use the following notation:
\begin{itemize}
\item $\forall n \in \mathbb{N}: [n] \eqdef \{1,\ldots,n\}$.

\item We denote the non-negative reals by $\mathbb{R}_{\ge0} \eqdef \{t \in \mathbb{R} \mid t \ge 0\}$.

\item Given a set $I$, we denote the set of ordered pairs of distinct elements of $I$
by \[\distinctpairs{I} \eqdef \{(i,j) \in I^2 \mid i \ne j\}.\]

\item Given a set $A$, $n \in \mathbb{N}$ and an $n$-tuple
$\tuplefull{a} \in A^n$, we denote the $n$-tuple containing the elements
of $\tuple{a}$ in reverse order by $\tuple{a}^{\mathit{rev}} \eqdef
(a_{n-m+1})_{m=1}^n$.

\item
Given a directed graph $G=(V,E)$,
we denote the set of paths in $G$ by
\[
\gpaths \eqdef \{\apathfull \in V^n \mid
n \in \mathbb{N} \And \forall m \in [n-1]: (p_m,p_{m+1}) \in E\}.
\]

\item
Given a weighted directed graph $G=(V,E,w)$, we denote the length of
a path $\apathfull \in \gpaths$ by
\[
\glength(\apath) \eqdef \sum_{m=1}^{n-1} w(p_m,p_{m+1}).
\]
Furthermore, we denote the distance function
between vertices of $G$ by
\functiondefn{\delta_G}{V^2}{[-\infty,\infty]}{(i,j)}
{\inf\{\glength(\apath) \mid \apathfull \in \gpaths \And p_1=i \And p_n = j\}.}

\end{itemize}

\chapter{A Discrete-Time Model}\label{chapter:discrete}

We model a set of agents that communicate with each other solely via message
passing. Each agent follows a predetermined protocol, which is common knowledge
to all agents.
In the following chapters, we concern ourselves
with the task of devising such protocols with the goal of analyzing a
coordination problem that we define in \chapterref{tcr-exhibition}.

To avoid over-burdening the reader with cumbersome details, the model presented
in this chapter is a discrete-time model, conceptually based on \cite{book},
which may be assumed while reading
this work. It should be noted, though, that the results presented throughout
this work hold verbatim also for a more intricate continuous-time
model, which we present in
\appref{continuous}. In \chapterref{practical}, we consider several natural
properties of practical continuous-time models, some of which always hold for
the model presented in this chapter, and prove results for models with these
properties.

\section{Context Parameters}

Intuitively, a context describes the environment in which the agents operate.
In this discrete-time model, we formally denote a context by a tuple
$\gamma\eqdef\context$, where:

\begin{enumerate}
\item $\contextgraph$ is a weighted directed graph with positive, integral or infinite,
weights.
The vertices $\agents$ of $\contextgraph$ model the agents.
We say that an agent $j \in \agents$
neighbours an agent $i \in I \setminus \{j\}$
if $(i,j) \in \contextneighbours$. If this is a case, then $i$ may,
as will be defined in greater precision below,
send messages to $j$, which are guaranteed to arrive
no later than $\contextbounds(i,j)$ after being sent.

\item For each agent $i \in \agents$, $S_i$ is a set of legal states for $i$.
We assume that $S_i$ is of large enough cardinality to accommodate all our
needs.

\item $\externalinputs$ is a set of ``possible external inputs''.
We think of external inputs as non-deterministic events, the occurrence of
which may not be anticipated in advance by any agent.

\item Each external input $\tilde{e} \in \externalinputs$ is associated with a
single agent $i_{\tilde{e}} \in \agents$, which observes this input when it
occurs.
\end{enumerate}

Additionally, we define the set of times as $\timeset \eqdef \mathbb{N} \cup \{0\}$.
As noted above, in \appref{continuous} we give an alternative model
description, in which time is continuous.

\section{Events and the Environment}
At each possible time $t \in \timeset$, zero or more events may take place.
Intuitively, an event is an occurrence that is observed by a single
agent. An event may be of one of the following types:

\begin{enumerate}
\item An external input event $\tilde{e} \in \externalinputs$. (Observed by $i_{\tilde{e}}$.)

\item A message delivery event $(m,t',(i,j))$ of a message $m$, sent by $i$ at
time $t'$, to $j$, s.t.\ $(i,j) \in \contextneighbours$. (Observed by $j$.)

\end{enumerate}

We define the ``state of the environment'' at any given time as the set of
events that take place at that time.
We denote the set of all possible states of the environment by
$S_e \eqdef 2^{\externalinputs} \times 2^{\messages \times \timeset \times \contextneighbours}$,
where $\messages$ is a set of all possible messages.

\section{States, Actions and Protocols}
The problem that we define in the next chapter deals with the coordination
of the responses of different agents to an external input.
At any time $t \in \timeset$, each agent $i \in \agents$ performs the
following,
in a manner that is based on its state, as well as on any events observed by it
at $t$:
\begin{enumerate}
\item Sets a new state for itself (which may be identical to its old state).

\item Sends any number of messages, each with possibly different content, and
to a possibly different neighbouring agent $j \in \agents$
(i.e.\ $j \in \agents$ s.t.\ $(i,j) \in \contextneighbours).$

\item Possibly ``responds''. This is the action that we aim to
coordinate.\footnote{
For simplicity, we define only one type of response per agent.
Our models, and our results in this work, may be readily generalized to
allow a set of possible responses for each agent.
}
\end{enumerate}
We thus define the set of possible actions that may be taken by $i$ at
$t$ as $A_i \eqdef
S_i \times 2^{\messages \times \{j \in I \mid (i,j) \in \contextneighbours\}}
\times \{\mathrm{\text{false}},\mathrm{true}\}$. (Each element in $A_i$ consists
of a new state for $i$ at $t$, a set of messages to be sent by $i$ at $t$,
and a boolean value that indicates whether $i$ responds at $t$.)

A ``local protocol'' for an agent $i \in \agents$ consists of a set of
possible initial states for $i$,
together with an ``action'' function, receiving as input a state
of $i$ just before a certain time $t \in \timeset$ and any events observed by
$i$ at $t$,
and outputting the actions to be performed by
$i$ at $t$.\footnote{
We restrict ourselves to deterministic protocols solely for ease of exposition.
}
Formally, a local protocol for $i$ is a pair $(\tilde{S}_i,P_i)$,
s.t.\ $\tilde{S}_i \subseteq S_i$ and
$P_i:S_i \times 2^{\{\tilde{e} \mid i_{\tilde{e}} = i\}} \times
2^{\messages \times I} \rightarrow A_i$. In certain cases, we may wish to allow the
actions of $i$ at $t$ to also depend on $t$ as well.
In such cases, we say that the
model is a ``shared-clock model'', and the actions
function takes the form
$P_i:S_i \times 2^{\{\tilde{e} \mid i_{\tilde{e}} = i\}} \times
2^{\messages \times I} \times \timeset \rightarrow A_i$.

A ``joint protocol'' (``protocol'', for short) is a collection of local
protocols,
one for each $i \in \agents$. We denote the set of all protocols of $\gamma$
by $\protocols$.

An important set of protocols is the set of ``full-information''
protocols\cite{full-information-protocol}, in
which the state of each agent $i \in \agents$ at any time $t \in \timeset$
uniquely determines the full details of every event observed by $i$ up until,
and including, $t$. Furthermore, at every $t \in \timeset$, $i$
sends to every neighbouring agent a message including the full current state of
$i$ at $t$, and in a shared-clock model --- also the current time.

\section{Runs}
For the duration of this section, fix a protocol
$P=((\tilde{S}_i,P_i))_{i \in I} \in \protocols$.
A ``run'' of $P$ in $\gamma$ is, intuitively, a ``possible infinite history'' of
$P$ executed in $\gamma$, in which the behaviour of all agents is governed by $P$.
Formally, a run of $P$ is a function
$r:\timeset \rightarrow S_e \times \bigtimes_{i \in I} S_i$
assigning, for each time,
a state for each agent, and a state for the environment,
while satisfying
the following properties:\footnote{
While a run is customarily defined in an inductive fashion, we define it here
without induction in order to minimize the differences between the definitions
of our discrete- and continuous-time models.
}

\begin{itemize}
\item Agent state consistency with local protocol:
Let $i \in I$ and $t \in \timeset$, then $r_i(t)$ must equal the first part
of the output of $P_i$ when applied to $r_i(t-1)$ (or to some
$\tilde{s}_i \in \tilde{S}_i$,
if $t=0$) and to the events observed by $i$ at $t$ (which depend only
on $r_e(t)$).
(The other parts of this output of $P_i$ determine the actions of $i$ at $t$.)

\item
Environment state properties:
\begin{enumerate}
\item Each external input $\tilde{e} \in \externalinputs$ may occur no more than once during a run.

\item Bounded message delivery: If a message is sent at any
$t \in \timeset$ by $i \in \agents$ to $j \in \agents$ (where $(i,j) \in \contextneighbours$),
then it must be delivered exactly once,
at some time $t' \in \timeset$, s.t.\ $t < t' \le t+\contextbounds(i,j)$.
If $t' < t+\contextbounds(i,j)$, then we say that this message is delivered
{\em early}.
Only messages that are sent during a run may be delivered during it.
\end{enumerate}
\end{itemize}

We call an event ``non-deterministic'' (ND for short) if it is either an
external input event, or an early delivery. (Intuitively, we may think of
ND events as events that cannot be foreseen by any agent before they occur,
and thus occur, in a sense, at the whim of the environment.) It should be noted
that this definition is context-dependent, as it depends on $\contextbounds$.

We denote the set of all possible runs of $P$ in $\gamma$ by $\RP$.
We denote the set of all runs of all protocols by in $\gamma$ by
$\runs \eqdef \cup_{P \in \protocols} \RP$.

It should be noted that two full-information protocols (as defined in the
previous section) $P,P' \in \protocols$ may only
differ in their response logic, and therefore
there is a natural isomorphism between $\RP$ and $\RPTAG$, which preserves the
set of ND events along with their occurrence times.

As was essentially shown in \cite{full-information-protocol}, given
a protocol $P \in \protocols$, there exists a full-information protocol
$P' \in \protocols$, s.t.\ there is a
natural monomorphism from $\RP$ into $\RPTAG$, which preserves both the set of
ND events (and their occurrence times), and all responses (and response times).
This ability of a full-information protocol to
simulate any other protocol implies, for our purposes, that if there exists
some protocol that solves a certain coordination problem, then there also
exists a full-information protocol that solves this problem. This justifies
restricting to full-information protocols when analyzing solvability, which
will sometimes prove convenient. (See, e.g.\
\defnref{optimal-response-logic} and \corref{tcr-iff-sr}.)

\section{Notation}
Given a context $\gamma=\context$, we introduce the following notation:

\begin{itemize}
\item Given a run $r \in \runs$, we denote the set of all
events in $r$ by $E(r)$. As we wish to regard each event as unique, and as we
defined events above not to contain their occurrence time, we
technically define $E(r)$ as a set of event-time pairs $(e,t_e)$, where each
event is paired with its occurrence time.

\item Given a run $r \in \runs$, we denote the set of all the ND events in $r$
by $\ND(r) \subseteq E(r)$.

\item
As mentioned above, we analyze the coordination of the responses of
different agents to a given external input. Thus,
given a protocol $P \in \protocols$ and an external input
$\tilde{e} \in \externalinputs$, we define the set of ``$\tilde{e}$-triggered''
runs of
$P$ as
\[\TRP \eqdef \{r \in \RP \mid \tilde{e} \in \ND(r)\}.\]
where by slight
abuse of notation, we use $\tilde{e} \in \ND(r)$ as a shorthand for
$\exists\, t \in \timeset: (\tilde{e},t) \in \ND(r)$.
We follow this convention occasionally, when no confusion can arise.
\end{itemize}

\chapter{Timely-Coordinated Response}\label{chapter:tcr-exhibition}

In this chapter, we define the main coordinated response problem underlying
this work, which we call ``timely-coordinated response'',
and analyze some of its basic properties.

\section{Coordinated Response Problems}

Before defining the timely-coordinated response problem, we first define
a broader, much simpler, coordinated response problem. (In a sense, it is the
simplest coordinated response problem.) This simpler problem will serve as a
building block for the timely-coordinated response problem. In addition,
the short discussion of this simpler problem will provide an introduction to
a novel theory of coordinated response problems,\footnote{
While Ben-Zvi and Moses\cite{bzm1,bzm2,bzm3,bzm4} define quite a
few problems to which they refer as ``response problems'' (all of which we
survey in Chapters \ref{chapter:syncausality-approach} and
\ref{chapter:previous} and
from which we draw motivation), they do not give a formal definition for a
response problem, or for a coordinated response problem.
}
which we attempt to formalize throughout this work. During this discussion, we
introduce some definitions and concepts, which we later reuse while discussing
various coordinated response problems throughout this work.
While philosophically a problem may be thought of as a specification, or as a
collection of constraints, we formally associate a coordinated response problem
with the set of protocols solving it, effectively treating two problems
that share the same set of solutions as identical.

\begin{defn}[Eventual Response]\label{defn:eventual-response}
Given a context $\gamma$, an external input $\tilde{e} \in \externalinputs$
and a set of agents $I \subseteq \agents$, we define the
``eventual response'' problem $\ER{\tilde{e},I} \subseteq \protocols$ as the
set of all protocols $P$ satisfying:
\begin{itemize}
\item
In each $r \in \TRP$,
each $i \in I$ responds exactly once.
In this case, we denote, for each $i \in I$,
the response time of $i$ in $r$ by $\timpl_r(i)$.
(Hence, $\timpl_r$ is a function from $I$ to $\timeset$.)

\item
In each $r \in \RP \setminus \TRP$, neither of the agents in $I$ responds.
In this case, we define $\timpl_r\equiv\infty$.
\end{itemize}
Thus, for every $r \in \RP$, we have defined a function
$\timpl_r:I\rightarrow \timeset \cup \{\infty\}$.
\end{defn}

\begin{remark}\label{remark:response-not-before}
Let $P \in \ER{\tilde{e},I}$ and let $r \in \TRP$. While each $i \in I$
responds in $r$, none of them do so before $\tilde{e}$ occurs.
\end{remark}

Intuitively, \remarkref{response-not-before} holds because, as we noted
when describing our model(s), the non-deterministic nature of $\tilde{e}$
implies that no agent may possibly infer that a run
is triggered before $\tilde{e}$ occurs in that run.
(Indeed, as far as
any agent is concerned, as long as $\tilde{e}$ has not occurred yet, it may be
the case that $\tilde{e}$ will never occur during this run.\footnote{
In the continuous-time model presented in
\appref{continuous}, this stems from the ``no foresight'' property.
})
Formally, this may be readily proven
using machinery that we have not introduced yet.
(See, e.g.\ the beginning of the proof of \thmref{path-traversing-centipede}
for a proof of a stronger statement using the tools of
\chapterref{syncausality-approach}, and \corref{ensembles-and-nd}
for a conceptually-similar
proof using the tools of \chapterref{fixed-point-approach}.)

\begin{ex}\label{ex:got-talent}
In the popular TV talent show ``Got Talent'', a panel of judges (that we
denote by $I$) judge various amateur performances, with themes ranging from
music, through magic, to some very obscure themes that are better left
undescribed. Once a contestant starts
performing on stage, each judge may press an ``X'' button to signal her
desire for this performance to end. (We denote this action as a response by the
judge.)
The performance continues until it has run its course, or until all
judges have pressed their respective ``X'' buttons.

As the ``mean'' judges sometimes implicitly compete between themselves
regarding who presses his ``X'' button first, let us consider a hypothetical
``enhancement'' to the show, in which a bucket of water is poured over the
head of any judge who presses his ``X'' button before a performance starts.

Consider a hypothetical repetitive (and thus, potentially never-ending)
performance in this TV show. (This indeed sometimes seems to be the case,
especially for very bad performances, which are abundant in this show.)
In order to guarantee that this performance
does not continue forever, the show producers must make sure that each judge
presses her ``X'' button at some time during the performance (but not before
the performance starts, as she wishes to remain dry).
The order of the ``X'' presses of the different judges is insignificant, as long
as it is guaranteed that each judge eventually presses her ``X'' button.
As the beginning of the performance depends on the rambling of the comedian
serving as show host (which may possibly also never end),
it is impossible for the judges to predict before it
actually occurs, so we may regard it as a non-deterministic external input,
which we denote as $\tilde{e}$.
\end{ex}

\begin{remark}
In \defnref{eventual-response} and hereafter, we assume, for simplicity,
that each agent is associated with no more than one response, as we did
when defining our model. (While it may seem
a bit silly to assume otherwise in the eventual response problem, it may
make sense to do so for more intricate coordinated response problems defined
below.)
Nonetheless, all the results we derive throughout this work regarding any
response problem are easily adaptable to the case in which more than one
response type is available per agent, and in which $I$ is a set of
agent-response pairs,
rather than merely agents. (In this way, some agents may be associated with
more than one of the responses being coordinated.)
Indeed, all our results apply verbatim to this generalized
case as well, if we allow ourselves,
for each agent-response pair $i=(\tilde{\imath},a) \in I$, to slightly
abuse notation by writing $i$ to refer to $\tilde{\imath}$ as well.
\end{remark}

While studying a coordinated response problem, we are usually interested
in two questions:
\begin{itemize}
\item Solvability: Under which conditions is it solvable? and if so,
\item Optimality: What is the ``fastest'' way to solve it?
\end{itemize}

Before answering these questions regarding the eventual response problem,
we first define them (and define coordinated response problems) precisely.

\begin{defn}[Coordinated Response]
Let $\gamma$ be a context, let $\tilde{e} \in \externalinputs$ and
let $I \subseteq \agents$.
We call a problem a ``coordinated response'' problem to $\tilde{e}$ by $I$,
if the set of protocols
solving it is a subset of (the set of protocols solving) $\ER{\tilde{e},I}$.
As before, we \linebreak
formally identify a coordinated response problem with the set of
protocols solving it.
\end{defn}

We usually define coordinated response problems by restrictions on
the response times $\timpl_r$ in all triggered runs $r$
(e.g.\ ``$I$ must
respond together in each triggered run'',
``$I$ must respond in a given order in each triggered run'', etc.).
In order to ease the reading of the following two
definitions, one may consider, as an example, the case in which
$\mathit{CP}\eqdef\ER{\tilde{e},I}$ for some $\tilde{e} \in \externalinputs$
and $I \subseteq \agents$.

\begin{defn}[Solvability]
Let $\gamma$ be a context and let $\mathit{CP} \subseteq \protocols$ be (the set of protocols solving) a
coordinated response problem.
We say that $\mathit{CP}$ is ``solvable'' if $\mathit{CP} \ne \emptyset$.
Otherwise, we say that it is ``unsolvable''.
\end{defn}

\begin{defn}[Optimal Response Logic]\label{defn:optimal-response-logic}
Let $\gamma$ be a context and let $\mathit{CP} \subseteq \protocols$ be
a coordinated response problem.
Assume that $\mathit{CP}$ is solvable and let $P \in \mathit{CP}$.
Assume w.l.o.g.\ that $P$ is a full-information protocol.
Let $P'$ be the protocol obtained from $P$ by modifying the response logic
of each $i \in I$ in some way.
As $P$ and $P'$ only differ by their response logic, there is a natural
isomorphism between $\RP$ and $\RPTAG$ that preserves the set of ND events.

We say that the response logic of $P'$ is an ``optimal response logic'' for
solving $\mathit{CP}$ if the following are satisfied:

\begin{enumerate}
\item
$P' \in \mathit{CP}$.

\item
If $r \in \RP$ and $r' \in \RPTAG$ are two runs of the respective protocols
matched under the above isomorphism, then
$\timpl_{r'} \le \timpl_r$.
\end{enumerate}
\end{defn}

We may now answer the above questions of solvability and optimality, with
regard to the eventual response problem.

\begin{remark}\label{remark:eventual-response-properties}
Let $\gamma$ be a context, let $\tilde{e} \in \externalinputs$ and let
$I \subseteq \agents$.
The following may be easily verified:
\begin{itemize}
\item
$\ER{\tilde{e},I}$ is
solvable iff for every $i \in I$ there exists $\apathfull \in \gpaths$ s.t.\
$p_1=i_{\tilde{e}}$ and $p_n=i$.
\item
An optimal response logic for solving $\ER{\tilde{e},I}$ is, for every
$i \in I$, ``respond as soon as $i$ receives information guaranteeing that
$\tilde{e}$ occurred.''\footnote{
In some runs of certain contexts with infinitely many agents under the
continuous-time model presented in
\appref{continuous}, the set of times at which $i$
has information guaranteeing that $\tilde{e}$ has occurred does not attain
its infimum value. It is straightforward to show that for such pathological
cases, no optimal response logic exists. Similar observations hold for
all other optimal response logics presented in this work as well.
}
(This seemingly-vague condition has a very precise meaning in a
full-information protocol, as each message sent in such a protocol uniquely
determines a list of events that are guaranteed to have occurred.)
\end{itemize}
\end{remark}

It should be noted that for many coordinated response problems,
it is fairly straightforward to deduce a solvability criterion (phrased
as a requirement on the topology and weights of $\contextgraph$) from an
optimal response logic. In the above example, the optimal response logic
demands, for $i$ to respond in each triggered run, that in each such run
$i$ receives information guaranteeing that $\tilde{e}$ occurred.
Thus, a necessary and sufficient condition for
solvability is that in each triggered run, each agent is guaranteed to receive
information to this effect. This is exactly equivalent to the solvability
criterion above. For this reason, for many coordinated response problems,
we will be primarily interested in an optimal response logic.

Due to the first part of \remarkref{eventual-response-properties},
we assume hereafter, whenever discussing a coordinated response problem
of a set of agents $I \subseteq \agents$ 
to an external input $\tilde{e} \in \externalinputs$, that there exist
paths in $\contextgraph$ from $i_{\tilde{e}}$ to every agent $i \in I$.

\section{Defining Timely-Coordinated Response}

We define the timely-coordinated response problem as a response problem in which
maximum and minimum values for the differences between the response times of
each pair of agents are provided.
In a discrete-time model, such constraints may be defined by any integral, or
infinite, value. In a continuous-time model, such constraints may be defined
by any real,
or infinite, value. In order to formally present the timely-coordinated response
problem, we first define the set of such possible constraining values.

\begin{defn}
We define $\Delta = (\timeset - \timeset) \cup \{-\infty, \infty\}$.
\end{defn}

\begin{remark}
If $\timeset = \mathbb{N} \cup \{0\}$, then $\Delta = \mathbb{Z} \cup \{-\infty, \infty\}$.
If $\timeset = \mathbb{R}_{\ge0}$, then $\Delta = [-\infty,\infty]$.
\end{remark}

We now turn to define the constraints imposed on response times in triggered
runs in the timely-coordinated response problem. Recall that given a set $I$,
we denote the set of ordered pairs of distinct elements of $I$ by
$\distinctpairs{I}$.

\begin{defn}[Implementation]\label{defn:implementation}
We call a pair $\implspec$ an ``implementation-spec'', if $I$ is a set
and if $\delta$ is a function
$\delta:\distinctpairs{I} \rightarrow \Delta$.
Given an implementation-spec $\implspec$, we call a function
$\timpl : I \rightarrow \timeset$ an ``implementation''
of $\delta$, if $\timpl(j) \le \timpl(i) + \delta(i,j)$ for every
$(i,j) \in \distinctpairs{I}$.
We denote the set of all implementations of $\delta$ by $T(\delta)$.
If $T(\delta) \ne \emptyset$, we say that $\delta$ is ``implementable''.
Otherwise, we say that it is ``unimplementable''.

\end{defn}

\begin{remark}\label{remark:implementation-properties}
Let $\implspec$ be an implementation-spec.
By \defnref{implementation}:
\begin{itemize}
\item
Obviously, $\delta$ is unimplementable unless $\delta>-\infty$.
Nonetheless, we still allow $\delta$ to take the value of $-\infty$ for
some or all agent pairs, for
technical reasons that may become apparent when we define a canonical
form for $\delta$ in the next chapter.
\item
Every $\timpl \in T(\delta)$ satisfies
$-\delta(j,i) \le \timpl(j) - \timpl(i) \le \delta(i,j)$ for every
$(i,j) \in \distinctpairs{I}$.

\item
Let $\timpl:I \rightarrow \timeset$. If $\timpl \in T(\delta)$,
then $\timpl+c \in T(\delta)$ as well, for every $c \in \timeset$. We say that
two implementations of $\delta$ are ``similar'' if they differ by a translation.

\item
$T$ is monotone: Let $\delta':\distinctpairs{I} \rightarrow \Delta$.
If $\delta \le \delta'$, then $T(\delta) \subseteq T(\delta')$.
\end{itemize}
\end{remark}

At last, we are ready to define the timely-coordinated response problem.

\begin{defn}[Timely-Coordinated Response]\label{defn:tcr}
We call a quadruplet $\TCRspec$ a ``TCR-spec'',
if $\gamma$ is a context, $\tilde{e} \in \externalinputs$ is an external input,
and $\implspec$ is an implementation-spec s.t.\ $I \subseteq \agents$.
Given a TCR-spec $\TCRspec$,
we define the
``timely-coordinated response'' problem
$\TCR{\tilde{e},I,\delta} \subseteq \ER{\tilde{e},I}$ as the set of all
eventual-response protocols $P$ for which $\timpl_r \in T(\delta)$
for every triggered run $r \in \TRP$.
\end{defn}

\begin{remark}\label{remark:tcr-properties}
Let $\TCRspec$ be a TCR-spec.
By \defnref{tcr}:
\begin{itemize}
\item
$\forall J \subseteq I: \TCR{\tilde{e},I,\delta} \subseteq \TCR{\tilde{e},J,\delta|_{\distinctpairs{J}}}$.
\item
Let $\delta':\distinctpairs{I} \rightarrow \Delta$.
If $T(\delta) \subseteq T(\delta')$, then $\TCR{\tilde{e},I,\delta} \subseteq \TCR{\tilde{e},I,\delta'}$.
\end{itemize}
\end{remark}

\begin{ex}
Returning to the ``Got Talent'' show from \exref{got-talent}.
Assume that the panel of judges consists of two judges: Alice and Bob.
The producers wish to create, among the viewers, the general impression
that Bob is a ``meaner'' judge than Alice, but only by a subtle difference.
One way to achieve this may be to ensure that Bob never responds more than
5 seconds after Alice, and that Alice never responds more than 40 seconds
after Bob. Coordinating this may not be very simple if, for example,
the judges are seated in a way that prevents each judge from knowing
when the other judge presses the ``X'' button. (Thus, for example,
Alice may have to rely on information regarding Bob's taste, such
``Bob always presses his ``X'' button no more than 10 seconds after someone
falls on stage.'') This is a simple instance of the timely-coordinated response
problem.
\end{ex}

The rest of this work, as noted above, is dedicated to the analysis of the
timely-coordinated response problem.

\chapter{The Constraining Function}\label{chapter:delta-analysis}

Before we turn to analyze the coordination required in order to solve
the timely-coordinated
response problem, we first note that for an unimplementable $\delta$,
the situation is hopeless to begin with, as the timely-coordinated response 
problem is is unsolvable regardless of the context $\gamma$ in which it is
defined.
In this chapter, we embark on a graph-theoretic discussion
with the aim of phrasing a necessary and sufficient condition for
implementability of a constraining function $\delta$.
First, though, we make the above comment regarding ``hopelessness'' precise:

\begin{claim}\label{claim:solvable-iff-implementable}
Let $\TCRspec$ be a TCR-spec.
\begin{enumerate}
\item
$\TCR{\tilde{e},I,\delta}$ is unsolvable
if $\delta$ is unimplementable.
\item
If $\TCR{\tilde{e},I,0}$\footnote{
We use 0 here and hereafter as a shortcut for the constant zero
constraining function, i.e.\ $\delta_0:I\rightarrow\Delta$ s.t.\
$\delta_0 \equiv 0$.
} is solvable,
then the converse holds as well, i.e.\ $\TCR{\tilde{e},I,\delta}$ is solvable
if $\delta$ is implementable. Furthermore, for any implementation $\ttildeimpl$
of $\delta$, there exists a solving protocol $P \in \TCR{\tilde{e},I,\delta}$,
for which the map from agents to response
times in every one of its triggered runs is similar to $\ttildeimpl$.
\end{enumerate}
\end{claim}

\begin{remark}\label{remark:re-solvable-iff-implementable}
Regarding \claimref{solvable-iff-implementable}:
\begin{itemize}
\item
Solvability of $\TCR{\tilde{e},I,0}$ is equivalent to the ability to
coordinate a simultaneous response of all agents in $I$
in every $\tilde{e}$-triggered run. This classic problem, known as the
``Firing Squad'' problem\cite{firing-squad1,firing-squad2} will repeatedly
appear in this
work. (See, e.g.\ \thmref{broom},
\corref{tcr-iff-sr}, \thmref{infinite-broom-or-infinite-brooms} and
\thmref{concise-broom}.)
\item
\corref{tcr-iff-sr} shows that under certain conditions, the 0 function
in the second part of \claimref{solvable-iff-implementable} may be replaced
with a variety of other functions.
\end{itemize}
\end{remark}

\begin{proof}[Proof of \claimref{solvable-iff-implementable}]
For the first part, assume that
$\TCR{\tilde{e},I,\delta}$ is solvable. Thus,
there exists $P \in \TCR{\tilde{e},I,\delta}$. Let $r \in \TRP$. By
\defnref{tcr}, $\timpl_r$ is an implementation of $\delta$,
and hence $\delta$ is implementable.

For the second part, assume that $\delta$ is implementable and
let $\ttildeimpl \in T(\delta)$.
Let $P_0 \in \TCR{\tilde{e},I,0}$.
Let $P$ be the protocol obtained from $P_0$ by modifying the response logic of
each $i \in I$ to ``respond $\ttildeimpl(i)$ time units
after the time $i$ would have responded in $P_0$''.
(This may require adding some auxiliary variables, which consume only a finite
amount of memory, to the state of $i$.\footnote{
As the set of possible states of $i$ is predefined by $\gamma$, it is
not technically
possible to add variables to it. The technical operation that we denote as
``adding a variable'' to the state of $i$ in $P_0$ involves utilizing states of
$i$ that are not utilized by $P_0$: Let $S^P_i \subseteq S_i$ be the
set of states of $i$ that are utilized in $P_0$ (i.e.\ the union of its set of
initial states $\tilde{S}_i$, with the image of the first coordinate of
${P_0}_i$). We choose, as the set of states of $i$
utilized by $P$, a subset of $S_i$ that is in one-to-one correspondence with
$\tilde{S}_i \times V$, where $V$ is the set of possible values of the variable
we wish to ``add'' to the state of $i$. As noted in \chapterref{discrete},
we assume that $S_i$ is of large enough cardinality to allow for the existence
of such a subset thereof.
} For the continuous-time model presented in \appref{continuous}, this response
logic may be implemented using timers.)
We complete the proof by showing that $P \in \TCR{\tilde{e},I,\delta}$.
As $P_0$ and $P$ only differ by their response logic, there is a natural
isomorphism between $\RPZERO$ and $\RP$, which preserves the set of ND events.
Let $r \in \RP$, and denote by $r_0$ the run of $P_0$ matched to $r$ under this
isomorphism.
If $r \in \TRP$, then $r_0 \in \TR{P_0}$, and
we have $\timpl_r = \timpl_{r_0} + \ttildeimpl < \infty$.
Note that as $P_0 \in \TCR{\tilde{e},I,0}$, we obtain that $\timpl_{r_0}$ is
a constant function --- denote its value by $t_{r_0}$.
By \remarkref{implementation-properties}, $\timpl_r = t_{r_0} + \ttildeimpl$ is an
implementation of $\delta$ (which is, by definition, similar to $\ttildeimpl$).
If $r \notin \TRP$, then
$r_0 \notin \TR{P_0}$.
In this case, we have
$\timpl_{r_0}\equiv\infty$, and therefore $\timpl_r\equiv\infty$ as well.
\end{proof}

By the second part of \claimref{solvable-iff-implementable},
the study of the implementability
of a constraining function $\delta$ may also be thought of as the study of
solvability of $\TCR{\tilde{e},I,\delta}$ in contexts $\gamma$ in which
$\TCR{\tilde{e},I,0}$ is solvable. (By both parts of that claim,
contexts in which $\TCR{\tilde{e},I,0}$ is solvable may be thought of as ideal
for solvability, in the sense that instances of the
timely-coordinated response problem that are unsolvable therein are unsolvable
in any context.\footnote{
Ideality for solvability, in this sense, only reveals part of the whole picture,
as it disregards the question of how fast can the responses
be coordinated after an occurrence of $\tilde{e}$.
})

As a first step toward analyzing the implementability of a function, we define
a canonisation operation on constraining functions, which preserves the set
of implementations.
The canonical form of a constraining function will aid us
in other aspects of the analysis of the timely-coordinated response problem
as well, due to \remarkref{tcr-properties}. In order to define this
canonical form, we consider $\delta$ as a weight function on the edges
of a directed graph on $I$.

\begin{defn}
Given an implementation-spec $\implspec$,
we define the weighted directed graph of $\delta$ as
$G_{\delta} \eqdef (I,E_{\delta},\delta|_{E_{\delta}})$,
where $E_{\delta} \eqdef \{(i,j) \in \distinctpairs{I} \mid \delta(i,j) \ne \infty\}$.
\end{defn}

\begin{remark}
Let $\implspec$ be an implementation-spec. By the above definition:
\begin{enumerate}
\item
If $I=\{i,j\}$, then every $\apath \in \dpaths$ is either of the form
$(\LaTeXunderbrace{i,j,i,j,\ldots}_n)$ or of the form $(\LaTeXunderbrace{j,i,j,i,\ldots}_n)$, for some $n \in \mathbb{N}$.
(If $|I|>2$, then $\dpaths$ is much richer.)

\item
$\forall \apath \in \dpaths: \dlength(\apath) < \infty$.
\end{enumerate}

\end{remark}

\begin{defn}[Canonical Form]
Let $\implspec$ be an implementation-spec.
We define the ``canonical form'' of $\delta$ as
$\hat{\delta} \eqdef \delta_{G_{\delta}}$, the distance function of
$G_{\delta}$. By slight abuse of notation,
we allow ourselves to write $\hat{\delta}$ instead of
$\hat{\delta}|_{\distinctpairs{I}}$ on some occasions below.
\end{defn}

\begin{remark}\label{remark:canonical-form-properties}
Let $\implspec$ be an implementation-spec.
By the above definition, $\hat{\delta}$ satisfies:
\begin{itemize}

\item
$\forall i \in I: \hat{\delta}(i,i) \in \{0, -\infty\}$.
(Thus, by \remarkref{implementation-properties}, for implementable
$\delta$ we obtain $\hat{\delta}|_{\{(i,i)\mid i \in I\}} = 0$.)
Furthermore,
$\hat{\delta}(i,i)=-\infty$ iff $i$ is a vertex along a negative cycle in
$G_{\delta}$.

\item
Idempotence:
$\hat{\hat{\delta}}=\hat{\delta}$.

\item
Minimality: $\hat{\delta} \le \delta$.

\item
Triangle inequality: $\forall i,j,k \in I:
\hat{\delta}(i,k) \le \hat{\delta}(i,j) + \hat{\delta}(j,k)$.

\item
Equivalence:
$T(\delta) = T(\hat{\delta})$.
($\subseteq$:~by the triangle inequality for path lengths.
$\supseteq$:~by minimality and by
\remarkref{implementation-properties} (monotonicity of $T$).)

\item
Monotonicity:
Let $\delta':\distinctpairs{I} \rightarrow \Delta$.
If $\delta \le \delta'$, then $\hat{\delta} \le \widehat{\delta}'$.

\end{itemize}
\end{remark}

We are now ready to characterise the implementable functions from
$\distinctpairs{I}$ to $\Delta$. The first part of the following lemma
performs this task, while its second part shows that for every
implementable $\delta$, there exists an implementation that is minimal
in every coordinate --- a result that gives us hope to find an optimal
response logic for the timely-coordinated response problem
for every $\delta$.\footnote{
We do not wish, by any means, to imply that the existence of such
a minimal implementation implies the existence of such an optimal
response logic, nor even that in every run $r$
of a protocol endowed with such a response logic, $\timpl_r$ is
similar to this minimal implementation. We merely note that in the
hypothetical absence of such a minimal implementation for some implementable
$\delta$, it would have been possible to show that no optimal response logic
exists for the timely-coordinated response problem based on that $\delta$.
See also a discussion regarding
\remarkref{broom-implies-path-traversing-centipede} below.
}

\begin{lemma}\label{lemma:implementable-iff}
Let $\implspec$ be an implementation-spec.
\begin{enumerate}
\item
$\delta$ is implementable
iff $\hat{\delta}|_{\{i\} \times I}$ is bounded from below for every $i \in I$.
\item
If $\delta$ is implementable, then
$i \mapsto -\inf(\hat{\delta}|_{\{i\} \times I})$
is an implementation thereof, which is minimal in each coordinate.\footnote{
A quick glance at this formulation of the minimal implementation may raise
a suspicion that perhaps it would have been more natural to define $\delta$
as the negation (in each coordinate) of the definition we have given.
While it is indeed possible to define $\delta$ this way, and while doing so
would have indeed given a more natural definition of the minimal
implementation, it would have also required us to work with greatest path
lengths instead of distances, with a reverse triangle inequality and with
order-reversing monotonicity, which may somehow seem less natural.
}
\end{enumerate}
\end{lemma}

\begin{proof}
We first prove that if $\delta$ is implementable, then
every $\timpl \in T(\delta)$ satisfies
$\timpl(i) \ge -\inf(\hat{\delta}|_{\{i\} \times I}))$. This
implies the first direction (``$\Rightarrow$'') of the first part, and the
inequality in the second part.

Assume that $\delta$ is implementable and let $\timpl \in T(\delta)$
be an implementation thereof.
By \remarkref{canonical-form-properties} (equivalence),
$\timpl \in T(\hat{\delta})$ as well.
Let $i \in I$. By definition of an implementation, we obtain
\[ \forall j \in I \setminus \{i\}: \hat{\delta}(i,j) \ge \timpl(j) - \timpl(i) \ge 0 - \timpl(i) = -\timpl(i). \]
By \remarkref{canonical-form-properties}, $\hat{\delta}(i,i)=0\ge-\timpl(i)$.
Thus, we have $\hat{\delta}|_{\{i\}\times I} \ge -\timpl(i)$.
Taking the infimum over $I$ of both sides of this inequality
completes this part of the proof.

We now prove that if $\delta|_{\{i\} \times I}$ is bounded from below for
every $i \in I$, then the function defined in the second part
is indeed an implementation of $\delta$. This completes the proof of both
parts.

Define $\timpl:I\rightarrow\timeset$ by
$i \mapsto -\inf(\delta|_{\{i\} \times I}) < \infty$.
By \remarkref{canonical-form-properties}, $\delta(i,i)\le 0$, and therefore
indeed $\timpl \ge 0$. Let $(i,j) \in \distinctpairs{I}$.
Let $\apathfull \in \dpaths$ s.t.\ $p_1=j$. Define $p_0 \eqdef i$.
Note that
\[\inf(\delta|_{\{i\}\times I}) \le \dlength((p_m)_{m=0}^n) =
\delta(i,j) + \dlength((p_m)_{m=1}^n).\]
By taking the infimum of both sides over all $\apath \in \dpaths$ s.t.\
$p_1=j$, we obtain
$\inf(\delta|_{\{i\} \times I}) \le
\delta(i,j) + \inf(\delta|_{\{j\} \times I})$.
Thus, we have $\timpl(j) \le \timpl(i) + \delta(i,j)$, as required.
\end{proof}

For the case in which $I$ is finite, the first part of
\lemmaref{implementable-iff} yields the following, more tangible, implementation
criterion.

\begin{cor}\label{cor:implementable-iff-no-negative-cycles}
Let $\implspec$ be an implementation-spec s.t.\ $|I|<\infty$ and
$\delta>-\infty$. \linebreak
$\delta$ is implementable iff $G_{\delta}$ contains no negative cycles.
\end{cor}

For completeness, we now prove a uniqueness property
one may expect from the canonical form defined above, showing that
the equivalence classes of implementable constraining functions,
under the equivalence
relation $\delta_1 \sim \delta_2 \Leftrightarrow T(\delta_1)=T(\delta_2)$,
are in one-to-one, order-preserving, correspondence with canonical forms.
At the heart of the proof of this property lies the following lemma.

\begin{lemma}\label{lemma:canonical-form-bounds-attained}
Let $\implspec$ be an implementation-spec s.t.\ $\delta$ is implementable,
and let $\tilde{\imath},\tilde{\jmath} \in I$.
\begin{enumerate}
\item
If $\hat{\delta}(\tilde{\imath},\tilde{\jmath})<\infty$, then
there exists an implementation $\timpl \in T(\delta)$ satisfying
\mbox{$\timpl(\tilde{\jmath})-\timpl(\tilde{\imath})=\hat{\delta}(\tilde{\imath},\tilde{\jmath})$}.
\item
If $\hat{\delta}(\tilde{\imath},\tilde{\jmath})=\infty$,
then for every $K \in \timeset$,
there exists an implementation $\timpl \in T(\delta)$ satisfying
$\timpl(\tilde{\jmath})-\timpl(\tilde{\imath}) \ge K$.
\end{enumerate}
\end{lemma}

\begin{proof}
By \lemmaref{implementable-iff}, $\forall i \in I: \exists d_i \in \timeset:
\hat{\delta}|_{\{i\} \times I} \ge -d_i$. (For the time being, we may choose
$(-d_i)_{i \in I}$ to be the infima of the respective restrictions of
$\hat{\delta}$.) We define
$\delta':\distinctpairs{I} \rightarrow \Delta \setminus \{\infty\}$ by
\[ \forall i,j \in \distinctpairs{I}: \delta'(i,j) =
\begin{cases}
\delta(i,j) &\delta(i,j) < \infty \\
d_j         &\delta(i,j) = \infty.
\end{cases} \]

As $\delta' \le \delta$, by monotonicity of $T$ it is enough
to find an implementation of $\delta'$ that satisfies the conditions of
the lemma. By
\remarkref{canonical-form-properties} (minimality),
it is enough find such an implementation for
$\widehat{\delta'}$.

We first show that $\forall i \in I: \widehat{\delta'}|_{\{i\} \times I} \ge -d_i$.
Let $i \in I$ and let $\apathfull \in \paths{G_{\delta'}}$ s.t.\ $p_1 = i$.
Set $l=|\{k \in [n-1] \mid \delta(p_k,p_{k+1}) = \infty\}|$ --- 
the number of ``new'' edges in $\apath$, which do not exist in $G_{\delta}$.
We show, by induction on $l$, that $\dtaglength(\apath) \ge -d_i$.

Base: If $l=0$, then $\dtaglength(\apath) = \dlength(\apath) \ge -d_i$.

Induction step: Assume $l \ge 1$. Let $k \in [n-1]$ be maximal such that
$(p_k,p_{k+1})$ is a ``new'' edge (i.e.\ $\delta(p_k,p_{k+1})=\infty$).
By definition of $\delta'$, we have $\delta'(p_k,p_{k+1})=d_{p_{k+1}}$.
Thus, by the induction hypothesis, we obtain
\[
\dtaglength(\apath) =
\dtaglength((p_m)_{m=1}^k) + \delta'(p_k,p_{k+1})
+ \dlength((p_m)_{m=k+1}^n)
\ge -d_i + d_{p_{k+1}} - d_{p_{k+1}} = -d_i.
\]
and the proof by induction is complete. In particular, we conclude that
$\widehat{\delta'} > -\infty$, and by definition, also
$\widehat{\delta'} \le \delta' < \infty$.

We claim that
$t \eqdef d_{\tilde{\imath}} + \widehat{\delta'}(\tilde{\imath},\cdot) \ge 0$
is an implementation of
$\widehat{\delta'}$.
Indeed, for every $(j,k) \in \distinctpairs{I}$, by
\remarkref{canonical-form-properties} (triangle inequality), we have
\[\timpl(k) = d_{\tilde{\imath}} + \widehat{\delta'}(\tilde{\imath},k) \le
d_{\tilde{\imath}} + \widehat{\delta'}(\tilde{\imath},j) +
\widehat{\delta'}(j,k) = \timpl(j) +
\widehat{\delta'}(j,k).\]

If $\hat{\delta}(\tilde{\imath},\tilde{\jmath})<\infty$, we define
$K \eqdef \hat{\delta}(\tilde{\imath},\tilde{\jmath})$; otherwise,
let $K \in \timeset$ be arbitrarily large as in the conditions of the lemma.
As $\delta'$ is implementable, by \remarkref{canonical-form-properties} we
obtain $\widehat{\delta'}(\tilde{\imath},\tilde{\imath})=0$.
Therefore,
\[\timpl(\tilde{\jmath}) - \timpl(\tilde{\imath}) =
(d_{\tilde{\imath}} + \widehat{\delta'}(\tilde{\imath},\tilde{\jmath})) -
(d_{\tilde{\imath}} + \widehat{\delta'}(\tilde{\imath},\tilde{\imath})) =
\widehat{\delta'}(\tilde{\imath},\tilde{\jmath}).\]
Thus, if $\widehat{\delta'}(\tilde{\imath},\tilde{\jmath})
\ge K$, then the proof is complete.
(For the case in which $\hat{\delta}(\tilde{\imath},\tilde{\jmath})<\infty$, we obtain $\widehat{\delta'}(\tilde{\imath},\tilde{\jmath})\le K$
by \remarkref{canonical-form-properties} (monotonicity),
since $\delta' \le \delta$.)

Otherwise,
set $d \eqdef K - \widehat{\delta'}(\tilde{\imath},\tilde{\jmath}) >0$,
and define $d_i' \eqdef d_i + d > d_i$, for every $i \in I$.
Therefore, $-d_i' < -d_i \le \hat{\delta}|_{\{i\} \times I}$ for every $i \in I$.
Denote by $\delta''$ the function constructed from $\delta$
in the same way in which $\delta'$ was constructed from it,
but using the lower bounds $(-d_i')_{i \in I}$ rather than
$(-d_i)_{i \in I}$.
As explained above, in order to complete the proof it is enough to show that
$\widehat{\delta''}(\tilde{\imath},\tilde{\jmath}) \ge K$.
Let $\apathfull \in \paths{G_{\delta''}}$ s.t.\ $p_1 = \tilde{\imath}$ and
$p_n = \tilde{\jmath}$.
If $\forall k \in [n-1]: \delta(p_k, p_{k+1}) \ne \infty$, then
$\dtagtaglength(\apath) = \dlength(\apath) \ge \hat{\delta}(\tilde{\imath},\tilde{\jmath}) \ge K$.
Otherwise,
\begin{align*}
&\:\dtagtaglength(\apath) = & \text{by definitions of $\delta'$ and $\delta''$} \\
=&\:
\dtaglength(\apath) + d\cdot|\{k \in [n-1] \mid \delta(p_k, p_{k+1}) = \infty \}| \ge & \text{as this set is non-empty} \\
\ge&\: \dtaglength(\apath) + d \ge & \text{by definition of $\widehat{\delta'}$} \\
\ge&\: \widehat{\delta'}(\tilde{\imath},\tilde{\jmath}) + d = & \text{by definition of $K$} \\
=&\:  K.
\end{align*}
Either way, the proof is complete.
\end{proof}

While unimplementable functions whose canonical forms differ may exist (due
to $I$ not necessarily being finite, and due to the fact that $G_{\delta}$ needs
not necessarily be strongly connected), we now conclude, using
\lemmaref{canonical-form-bounds-attained}, that for implementable
functions, the map $\hat{\delta}\mapsto T(\delta)$ from the canonical
form of an implementable function $\delta$ to the set of implementations of
$\delta$ is a well-defined, order-preserving, monomorphism.

\begin{cor}\label{cor:canonical-form-order}
Let $I$ be a set and let
$\delta_1,\delta_2:\distinctpairs{I} \rightarrow \Delta$ s.t.\
$\delta_1$ is implementable.
$\hat{\delta}_1 \le \hat{\delta}_2$ iff $T(\delta_1) \subseteq T(\delta_2)$.
\end{cor}

\begin{proof}
$\Rightarrow$:
Assume that $\hat{\delta}_1 \le \hat{\delta}_2$.
By monotonicity of $T$ and by
\remarkref{canonical-form-properties} (equivalence), we have
$T(\delta_1) = T(\hat{\delta}_1) \subseteq T(\hat{\delta}_2) = T(\delta_2)$.

$\Leftarrow$:
Assume that $\hat{\delta}_1 \nleq \hat{\delta}_2$. Thus, there exist
$\tilde{\imath},\tilde{\jmath} \in I$ s.t.\
$\hat{\delta}_1(\tilde{\imath},\tilde{\jmath}) > \hat{\delta}_2(\tilde{\imath},\tilde{\jmath})$.
If $\hat{\delta}_1(\tilde{\imath},\tilde{\jmath})<\infty$, then
by \lemmaref{canonical-form-bounds-attained} there exists
$\timpl \in T(\delta_1)$ s.t.\ $\timpl(\tilde{\jmath})-\timpl(\tilde{\imath}) =
\hat{\delta}_1(\tilde{\imath},\tilde{\jmath}) > \hat{\delta}_2(\tilde{\imath},\tilde{\jmath})$,
and thus $\timpl \in T(\delta_1) \setminus T(\delta_2)$, and the proof is
complete.

If $\hat{\delta}_1(\tilde{\imath},\tilde{\jmath})=\infty$, then
$\hat{\delta}_2(\tilde{\imath},\tilde{\jmath})<\infty$ and thus there exists
$K \in \timeset$ s.t.\
$K > \hat{\delta}_2(\tilde{\imath},\tilde{\jmath})$.
Similarly to the proof of the
previous case, by \lemmaref{canonical-form-bounds-attained} there exists
$\timpl \in T(\delta_1)$ s.t.\ $\timpl(\tilde{\jmath})-\timpl(\tilde{\imath}) \ge K
> \hat{\delta}_2(\tilde{\imath},\tilde{\jmath})$.
Once again, we obtain that $\timpl \in T(\delta_1) \setminus T(\delta_2)$, and the
proof is complete.
\end{proof}

\begin{cor}\label{cor:canonical-form-uniqueness}
Let $I$ be a set and let
$\delta_1,\delta_2:\distinctpairs{I} \rightarrow \Delta$ s.t.\ at least
one of them is implementable.
$\hat{\delta}_1 = \hat{\delta}_2$ iff $T(\delta_1) = T(\delta_2)$.
\end{cor}

\begin{proof}
$\Rightarrow$: 
Assume that $\hat{\delta}_1=\hat{\delta}_2$. By applying \remarkref{canonical-form-properties} (equivalence) twice,
we obtain $T(\delta_1)=T(\hat{\delta}_1)=T(\hat{\delta}_2)=T(\delta_2)$.

$\Leftarrow$:
Assume that $T(\delta_1) = T(\delta_2)$. Thus, since at least one of
$\delta_1,\delta_2$ is implementable, they both are.
To complete the proof, we apply \corref{canonical-form-order} to
$T(\delta_1) \subseteq T(\delta_2)$ and to $T(\delta_2) \subseteq T(\delta_1)$.
\end{proof}

The above discussion gives rise to two alternative definitions (or rather,
characterisations) of the
canonical form of implementable functions:
The first one, non-constructive in nature, justifies
the name of the minimality property stated in
\remarkref{canonical-form-properties} and stems from this property when
combined with \corref{canonical-form-uniqueness}.
The second one, which constructively defines the inverse of the order-preserving
monomorphism $\hat{\delta}\mapsto T(\delta)$, stems directly from
\lemmaref{canonical-form-bounds-attained}.

\begin{cor}\label{cor:canonical-form-characterisations}
Let $\implspec$ be an implementation-spec s.t.\ $\delta$ is implementable.
\begin{enumerate}
\item
$\hat{\delta} = \min\{\delta' \in \Delta^{(\distinctpairs{I})} \mid T(\delta')=T(\delta) \}$.
(In particular, there exists a function in this set, which is minimal
in each coordinate, although this may be proven directly by means of a simpler
argument.)
\item
$\forall i,j \in I: \hat{\delta}(i,j) = \max\{\timpl(j)-\timpl(i) \mid \timpl \in T(\delta) \}$.
\end{enumerate}
\end{cor}

\begin{remark}
Implementability of $\delta$ is not required in
\corref{canonical-form-characterisations} if $|I|<\infty$ and if $G_{\delta}$ is
strongly connected. Indeed, under such conditions, if $\delta$ is
unimplementable, then $\hat{\delta} \equiv -\infty$, which coincides with the
function obtained in both parts of this corollary, when they are applied to any
unimplementable
$\delta$.
\end{remark}

By applying \claimref{solvable-iff-implementable} to the previous three
corollaries, we obtain similar results regarding the map
$\hat{\delta}\mapsto\TCR{\tilde{e},I,\delta}$ from the
canonical form of
an implementable function $\delta$ to the timely-coordinated response
problem that $\delta$ defines
with respect to a fixed external input.
We conclude this chapter by formulating these results.

\begin{cor}\label{cor:canonical-form-and-tcr}
Let $\gamma$ be a context, let $I \subseteq \agents$ and let
$\tilde{e} \in \externalinputs$ s.t.\ $\TCR{\tilde{e},I,0}$
is solvable.\footnote{
As noted above regarding \remarkref{re-solvable-iff-implementable},
\corref{tcr-iff-sr} shows that under certain conditions, the 0 function
in \corref{canonical-form-and-tcr} may be replaced
with a variety of other functions.
}
\begin{itemize}
\item
Let $\delta_1,\delta_2:\distinctpairs{I} \rightarrow \Delta$.
\begin{enumerate}
\item
If $\delta_1$ is implementable, then:
$\hat{\delta}_1 \le \hat{\delta}_2$ iff
$\TCR{\tilde{e},I,\delta_1} \subseteq \TCR{\tilde{e},I,\delta_2}$.
\item
If either $\delta_1$ or $\delta_2$ are implementable, then:
$\hat{\delta}_1 = \hat{\delta}_2$ iff
$\TCR{\tilde{e},I,\delta_1} = \TCR{\tilde{e},I,\delta_2}$.
\end{enumerate}
\item
Let $\delta:\distinctpairs{I} \rightarrow \Delta$ be implementable.
(Once again, implementability
of $\delta$ is not required for this part if $|I|<\infty$ and if $G_{\delta}$
is strongly connected.)
\begin{enumerate}
\item
$\hat{\delta} = \min\{\delta' \in \Delta^{(\distinctpairs{I})} \mid
\TCR{\tilde{e},I,\delta'}=\TCR{\tilde{e},I,\delta} \}$.
\item
$\forall i,j \in I: \hat{\delta}(i,j) =
\max\{\timpl_r(j)-\timpl_r(i) \mid r \in \cup_{P \in \TCR{\tilde{e},I,\delta}} \RP\}$.
\end{enumerate}
\end{itemize}
\end{cor}

Readers who found our choice from the previous chapter, to formally associate
a coordinated response problem with the set of solutions thereof
philosophically troubling, may find some justification for this choice in the
first part of \corref{canonical-form-and-tcr}. This part essentially shows that
in order to accept our choice, at least when $\TCR{\tilde{e},I,0}$ is solvable
and for implementable $\delta$, it is enough to accept that
$\TCR{\tilde{e},I,\hat{\delta}}$ is the same problem as
$\TCR{\tilde{e},I,\delta}$.

\chapter{The Syncausality Approach}\label{chapter:syncausality-approach}

In this chapter, we analyze the timely-coordinated response problem using tools
developed by Ben-Zvi and Moses\cite{bzm1,bzm2}, and generalize some previous
results obtained  by them\cite{bzm1,bzm2,bzm3,bzm4} using these tools.
The proofs that we give in this chapter, unlike the proofs in
\cite{bzm1,bzm2,bzm3,bzm4}, do not explicitly use the concept of knowledge.
We choose to phrase our proofs in this way in order to emphasize the
difference between the approach taken in this chapter and that of the next one.

\section{Background}

This section surveys previous definitions and results from
\cite{bzm1,bzm2,bzm3,bzm4}.

\subsection{Partial Orders on the Set of Agent-Time Pairs}

Recall that we work in a context consisting of a set of agents $\agents$ that
communicate
with each other solely via message passing, and that their communication
channels are modeled by the edges of the directed graph
$\contextgraph=(\agents,\contextneighbours,\contextbounds)$, each of which being
weighted according to the maximum delivery time of a message along it.

Ben-Zvi and Moses\cite{bzm1,bzm2} define two partial order relations on
the set of agent-time pairs $\agents \times \timeset$. The first, called
``syncausality'' (short for synchronous causality), is a synchronous
counterpart to
Lamport's ``happened-before'' causality relation\cite{lamport-causality}, and
similarly aims to capture information
flow. Intuitively, if $(i,t) \syncausal{r} (j,t')$, then in some sense,
$j$ potentially has, at $t'$ in $r$, information regarding the state of $i$ at
$t$ in $r$. (In
a full-information protocol, this intuition can be made more concrete:
the state of $i$ at $t$ in $r$ can be deduced with absolute certainty from the
state of $j$ at $t'$ in $r$.)

\begin{defn}[Syncausality]\label{defn:syncausality}
Let $\gamma$ be a context and let $r \in \runs$.
The ``Syncausality'' relation $\syncausal{r}$ is the minimal partial order
relation on $\agents \times \timeset$ satisfying (i.e.\ the transitive
closure of)
\begin{itemize}
\item
Locality: $\forall i \in \agents, t,t' \in \timeset : \ t' \ge t \ \Rightarrow \ (i,t) \syncausal{r} (i,t')$.

\item
Message delivery: If, in $r$, a message is sent from $i \in \agents$ at $t \in \timeset$ and
delivered to $j \in \agents$ at $t' > t$, then $(i,t) \syncausal{r} (j,t')$.

\item
Delivery guarantee: $(i,t) \syncausal{r} (j,t+\contextbounds(i,j))$,
for every $(i,j) \in \contextneighbours$ s.t.\linebreak
$\contextbounds(i,j) < \infty$,
and for every $t \in \timeset$.
\end{itemize}
\end{defn}

The syncausality relation is a refinement of Lamport's ``happened-before''
causality relation\cite{lamport-causality}, which is defined similarly,
with the only difference being the absence of the delivery guarantee property.
At first sight, this property may seem redundant
due to the message delivery property.
Indeed, the bound guarantee property is of importance only if a message is {\em not}
sent from $i$ to $j$ at $t$.
Intuitively, $j$ has a guarantee that $i$ did not
send it a message at $t$, only when the worst-case delivery time for
such a message has elapsed, i.e.\ at $t+\contextbounds(i,j)$.
Passing information by {\em not} sending a message was
first studied by Lamport\cite{lamport-null-messages}, who called such unsent
messages ``null messages''.

The second partial order relation on $\agents \times \timeset$,
called ``bound guarantee'', aims to capture
guaranteed information flow, and thus has no asynchronous counterpart.
Intuitively, if $(i,t) \boundguarantee (j,t')$,
then in some sense, it is not only that $j$ potentially has, at $t'$ in any run
$r$, information regarding $i$ at $t$ in $r$, but also that $i$ has some guarantee at $t$ in
$r$ that such information has the potential to reach $j$ by $t'$ in $r$ at the
latest. (In a full-information
protocol, this means that from the state of $i$ at $t$ in $r$, it can be deduced
(with absolute certainty) that this state may be deduced from the state of $j$
at $t'$ in $r$.)

\begin{defn}[Bound Guarantee]
Let $\gamma$ be a context.
The ``Bound Guarantee'' relation $\boundguarantee$ is the minimal partial
order relation on $\agents \times \timeset$ satisfying:
\begin{itemize}
\item
Locality: $\forall i \in \agents, t,t' \in \timeset: \ t' \ge t \ \Rightarrow\ (i,t) \boundguarantee (i,t')$.

\item
Delivery guarantee: $(i,t) \boundguarantee (j,t+\contextbounds(i,j))$,
for every $(i,j) \in \contextneighbours$ s.t.\linebreak
$\contextbounds(i,j) < \infty$,
and for every $t \in \timeset$.
\end{itemize}
\end{defn}

\begin{remark}
As would be expected by the intuitive descriptions of both
relations above, the bound guarantee relation is a subrelation of any
syncausality relation:
If~$(i,t) \boundguarantee (j,t')$, then
$\forall r \in \runs: (i,t) \syncausal{r} (j,t')$.
\end{remark}

\subsection{Additional Notation}

We now introduce some novel notation, which aims to capture the flow
of information regarding the occurrence of events, and the flow of
information which may affect the occurrence of an event.
This notation will both
aid us in more succinctly presenting some previous results of Ben-Zvi
and Moses in the next subsection, and in presenting our results
thereafter.

\begin{defn} Let $\gamma$ be a context and let $r \in \runs$.
\begin{enumerate}
\item Given an event $e \in E(r)$ and an
agent-time pair $(i, t) \in \agents \times \timeset$,
we write $e \syncausal{r} (i,t)$ (resp.\ $e \boundguarantee (i,t)$)
if $(i_e,t_e) \syncausal{r} (i,t)$ (resp.\ $(i_e,t_e) \boundguarantee
(i,t)$), where by $i_e$ we denote the immediate observer of $e$. (Recall
from \chapterref{discrete}, that $t_e$ is the occurrence time of $e$.)

\item Given two events $e, e' \in E(r)$, we write
$e \syncausal{r} e'$ (resp.\ $e \boundguarantee e'$) if either
$e = e'$ or $e'$ is a delivery of a message sent by an agent
$i \in \agents$ at time $t \in \timeset$ s.t.\ $e \syncausal{r} (i,t)$
(resp.\ $e \boundguarantee (i,t)$).

\end{enumerate}
\end{defn}

Once again, in a full-information protocol, some of the
implications of these definitions become very concrete, e.g.\
$e \syncausal{r} (i,t)$ guarantees that the occurrence of $e$ may be
deduced from the state of $i$ at $t$.
Similarly, $e \syncausal{r} e'$
guarantees that if $e'$ is a message
event, then the occurrence of $e$ may be deduced from the contents of the
message associated with $e'$. Moreover, if $e \syncausal{r} e'$ does {\em not}
hold, then the occurrence of $e'$ does not depend, in a sense, on the occurrence
of $e$.
We make this last observation precise in \corref{most-generalized-lemma-3}
below.

\subsection{Previous Results}

In this subsection, we survey the coordinated response problems defined and
studied by Ben-Zvi and Moses in \cite{bzm1,bzm2,bzm3,bzm4}, and their results
for these problems in discrete-time models. (The only coordinated response
problem from $\cite{bzm1,bzm2,bzm3,bzm4}$ that we do not survey in this
subsection, namely ``general ordered response'', is discussed in
\chapterref{previous}.)
We reformulate these problems, results, and the associated definitions
to match our notation, and to make use of our coordinated-response-theoretic
definitions.

While surveying all these coordinated response problems, and while remarking,
by defining an appropriate $\delta$ function, that the
timely-coordinated response problem extends each and every one of them (and also
extends general ordered response),
one property, which is common to
all these $\delta$ functions, should be spelled out explicitly:
$\hat{\delta}$ is
antisymmetric on each strongly-connected component of $G_{\delta}$.\footnote{
It is interesting to note, though, that some instances of the timely-coordinated
response problem, while having this property, are not instances of any of the
problems defined and studied by Ben-Zvi and Moses. We analyze such instances
in the second part of \corref{path-traversing-centipede-implies-centibroom}.
}
As we will see during this work, the absence of this property
in the timely-coordinated response problem introduces a significant amount of
complexity, both technically, and conceptually.

The first, most-basic coordinated response problem defined in \cite{bzm1,bzm2}
is that of ordered response.
\begin{defn}[Ordered Response]
Given a context $\gamma$, an external input $\tilde{e} \in \externalinputs$, $n \in \mathbb{N}$
and agents $\tuple{\imath}=(i_m)_{m=1}^n \in {\agents}^n$, define the ``ordered response''
problem $\OR{\tilde{e},\tuple{\imath}}\subseteq \ER{\tilde{e},\{i_m\}_{m=1}^n}$
as the set of all eventual-response protocols $P$ satisfying
$\timpl_r(i_{m+1}) \ge \timpl_r(i_m)$ for
every $m \in [n-1]$ and for every triggered run $r \in \TRP$.
\end{defn}

\begin{remark}
$\OR{\tilde{e},(i_m)_{m=1}^n} = \TCR{\tilde{e},\{i_m\}_{m=1}^n,\delta}$, for
\[ \delta(i_k,i_l) \eqdef
\begin{cases}
0 &k = l+1 \\
\infty &\text{otherwise.}
\end{cases} \]
\end{remark}

Ben-Zvi and Moses analyze this problem using a structure they call ``centipede''.

\fig{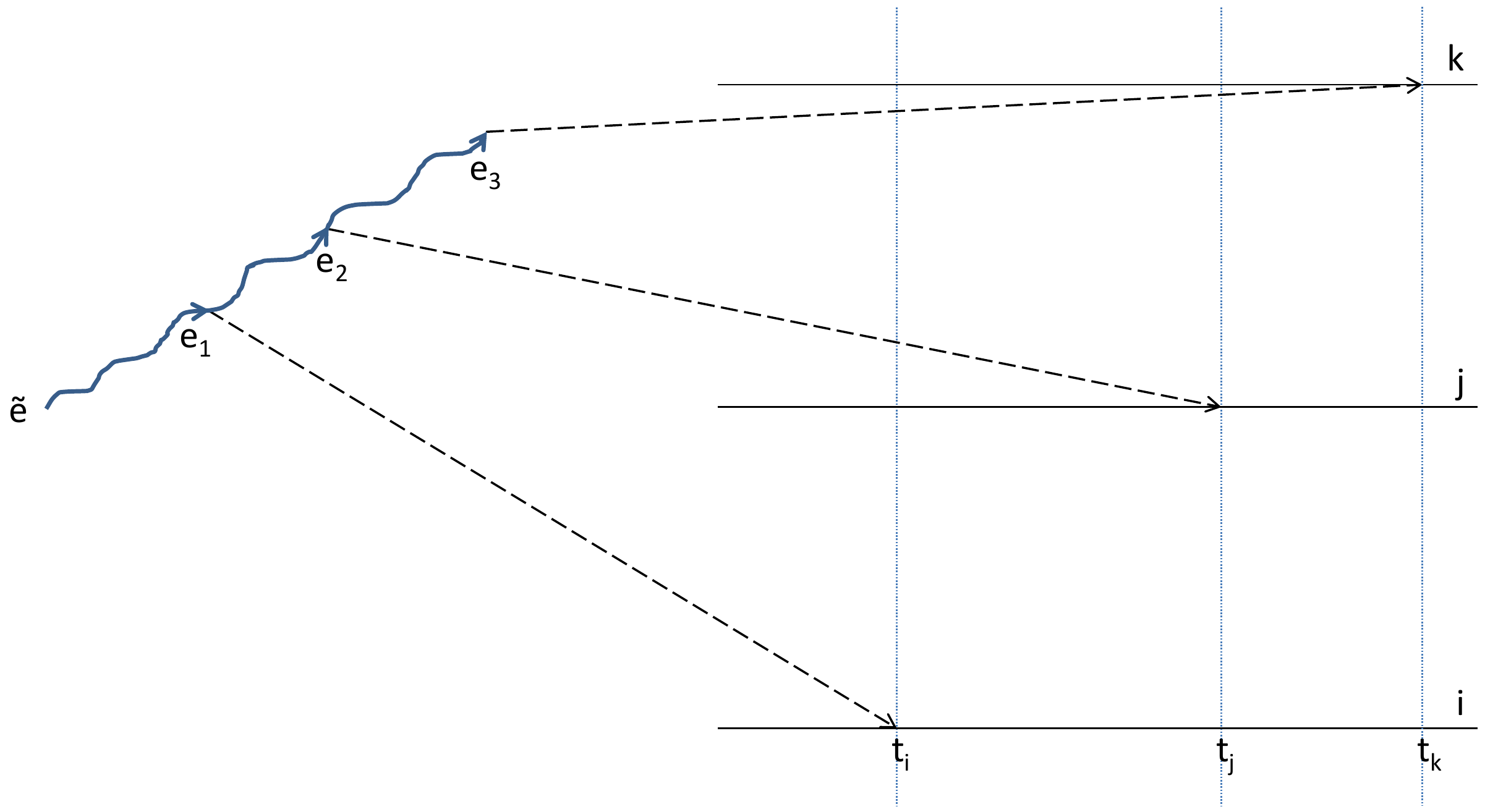}{Centipede}{
  $(e_1,e_2,e_3)$ is an $\tilde{e}$-centipede for $\{i,j,k\}$ by $(t_i,t_j,t_k)$.
}

\begin{defn}[Centipede --- see \figref{uneven-centipede}]
Given a context $\gamma$, a run $r \in \runs$, an external input
$\tilde{e} \in \externalinputs$, $n \in \mathbb{N}$,
and agents $\tuple{\imath} \in {\agents}^n$, with matching times
$\tuplefull{t} \in \timeset^n$,
we call an $n$-tuple of ND events $\tuplefull{e} \in \ND(r)^n$ an
``$\tilde{e}$-centipede'' for $\tuple{\imath}$ by $\tuple{t}$ if the
following hold:
\begin{itemize}
\item $\tilde{e} \syncausal{r} e_1$ and $\forall m \in [n-1]: e_m \syncausal{r} e_{m+1}$.
\item $\forall m \in [n]: e_m \boundguarantee (i_m,t_m)$.
\end{itemize}
Given a time $t \in \timeset$, we call an $n$-tuple of ND events
$\tuple{e} \in \ND(r)^n$ an
``$\tilde{e}$-centipede''
for $\tuple{\imath}$ by $t$, if it is an $\tilde{e}$-centipede for $\tuple{\imath}$ by $(t)^n$.
\end{defn}

\begin{thm}[Centipede]\label{thm:centipede}
In a discrete-time model,
let $\gamma$ be a context, let $n \in \mathbb{N}$,
let $\tuple{\imath} \in {\agents}^n$
and let $\tilde{e} \in \externalinputs$.
\begin{enumerate}
\item
Let $P \in \OR{\tilde{e},\tuple{\imath}}$. Each $r \in \TRP$ contains
an $\tilde{e}$-centipede for $\tuple{\imath}$ by $\timpl_r(i_n)$.
\item
In a shared-clock model, an optimal response logic for solving
$\OR{\tilde{e},\tuple{\imath}}$ is, for every $i_m$: ``respond at the
earliest time by which an $\tilde{e}$-centipede for
$(i_k)_{k=1}^m$ exists.''
\end{enumerate}
\end{thm}

Ben-Zvi and Moses prove the first part of \thmref{centipede} in two stages:
First, reducing to a response-recalling protocol
(a protocol in which the set of responses of an agent up until time $t$ may
may be deduced from its state at $t$),
they show that under the
conditions of that part of the theorem, at $\timpl_r(i_n)$ it holds that
$K_{i_n} \cdots K_{i_1} \tilde{e}$
(where $K_i$ means ``$i$ knows that\ldots'' --- this will be formally defined
in \chapterref{fixed-point-approach}). Second, conceptually following the
path of Chandy and Misra\cite{chandy-misra}, they deduce the existence of the
required centipede from this nested-knowledge formula.

Regarding the second part of \thmref{centipede}, it should be noted that
Ben-Zvi and Moses do not define a notion of optimality, but rather show
that a full-information protocol with the given response logic solves
$\OR{\tilde{e},\tuple{\imath}}$. (Optimality, under our definition, may
be derived from the combination of the two parts of \thmref{centipede},
which Ben-Zvi and Moses prove as separate theorems.)
A similar note holds for the second part of \thmref{broom} below.

The second problem presented in \cite{bzm1,bzm2} is the following variant of
the firing squad problem\cite{firing-squad1,firing-squad2}.

\begin{defn}[Simultaneous Response]
Given a context $\gamma$, an external input $\tilde{e} \in \externalinputs$ and set of agents
$I\subseteq\agents$, define the ``simultaneous response'' problem
$\SR{\tilde{e},I} \subseteq \ER{\tilde{e},I}$ as the set of all
eventual-response protocols $P$ for which $\timpl_r$ is a constant function
for each run $r \in \RP$. We denote, in this case, the constant value of
$\timpl_r$ by $t_r$.
\end{defn}

\begin{remark}
$\SR{\tilde{e},I} = \TCR{\tilde{e},I,0}$.
\end{remark}

Ben-Zvi and Moses analyze this problem using a structure they call ``broom''.

\fig{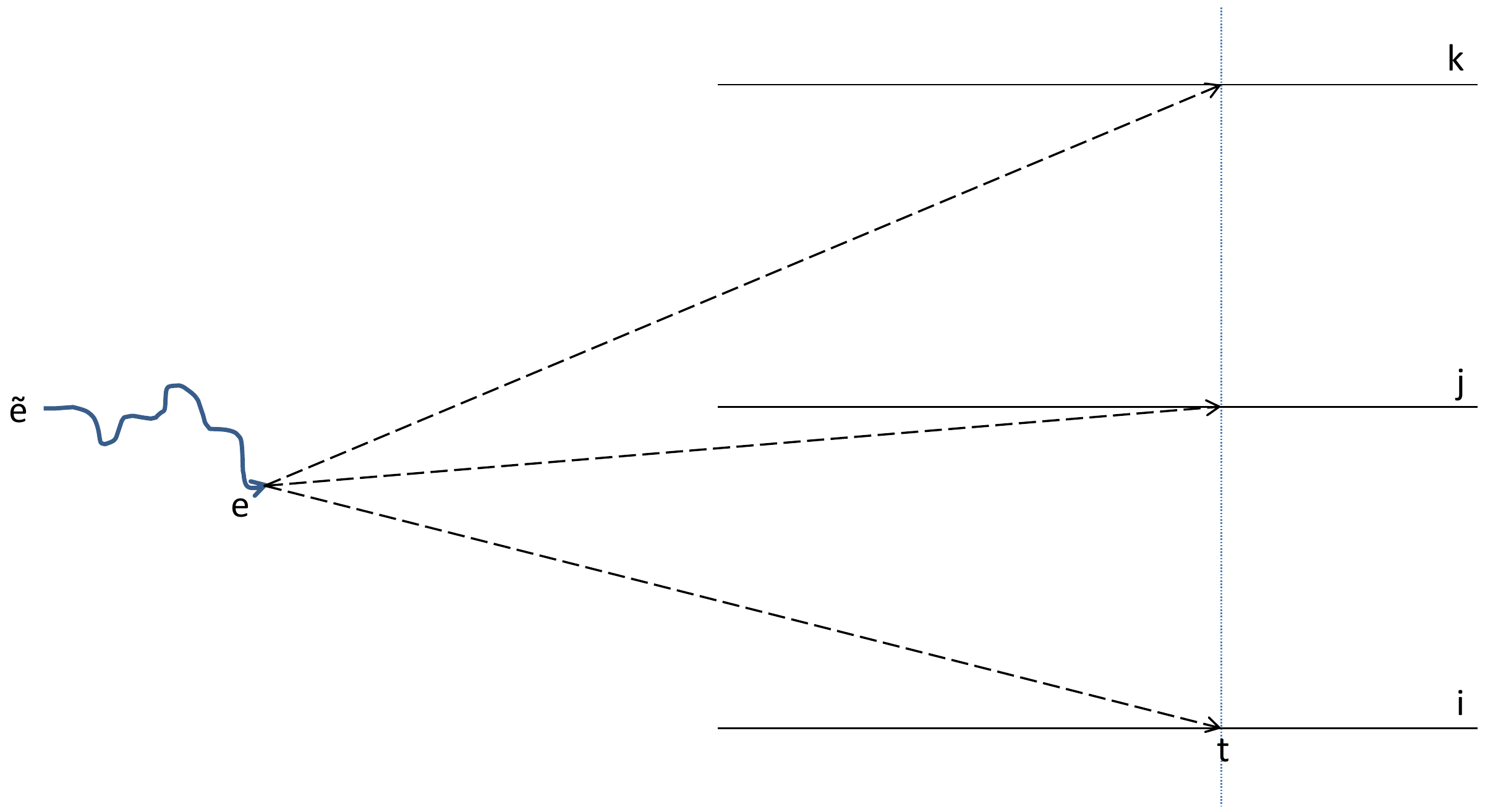}{Broom}{
  $(e_1,e_2,e_3)$ is an $\tilde{e}$-broom for $\{i,j,k\}$ by $t$.
}
\begin{defn}[Broom --- see \figref{broom}]
Given a context $\gamma$, a run $r \in \runs$, an external input
$\tilde{e} \in \externalinputs$
and a set of agents $I\subseteq\agents$ with matching times
$\tuple{t}=(t_i)_{i \in I} \in \timeset^I$,
we call an ND event $e \in \ND(r)$ an ``$\tilde{e}$-broom'' for $I$ by
$\tuple{t}$ if the following hold:
\begin{itemize}
\item $\tilde{e} \syncausal{r} e$
\item $\forall i \in I: e \boundguarantee (i,t_i)$.
We call $\{(i,t_i)\}_{i \in I}$ the set of ``end nodes'' of this broom,
and call $\max\{\tuple{t}\}$ the ``horizon'' of this broom.
\end{itemize}
Given a time $t \in \timeset$, we call an ND event $e \in \ND(r)$ an
``$\tilde{e}$-broom''
for $I$ by $t$, if it is an $\tilde{e}$-broom for $I$ by $(t)^I$.
\end{defn}

\begin{thm}[Broom]\label{thm:broom}
In a discrete-time model,\footnote{
This is a key requirement here.}
let $\gamma$ be a context, let $I \subseteq \agents$ be finite,
and let
$\tilde{e} \in \externalinputs$.
\begin{itemize}
\item
Let $P \in \SR{\tilde{e},I}$. Each $r \in \TRP$ contains
an $\tilde{e}$-broom for $I$ by $t_r$.
\item
In a shared-clock model, an optimal response logic for solving
$\SR{\tilde{e},I}$ is, for every $i \in I$: ``respond at the earliest
time by which an $\tilde{e}$-broom for $I$ exists.''
\end{itemize}
\end{thm}

Ben-Zvi and Moses prove the first part of \thmref{broom} by reducing to a
response-recalling protocol,
showing that under the conditions of that part of the
theorem, $\tilde{e}$ is common knowledge among all agents in $I$ at $t_r$,
and then using a reduction to the first part of \thmref{centipede}.
We give a direct proof of a slight generalization of \thmref{broom}
later in this work (see \thmref{concise-broom}).

The third and last problem presented in \cite{bzm2}, is the following
generalization of both ordered response and simultaneous response.

\begin{defn}[Ordered Joint Response]
Given a context $\gamma$, an external input $\tilde{e} \in \externalinputs$, $n \in \mathbb{N}$
and pairwise-disjoint non-empty sets of agents
$\tuplefull{I} \in (2^{\agents})^n$,
define the ``ordered joint response'' problem
$\OJR{\tilde{e},\tuple{I}} \subseteq \ER{\tilde{e},\cup_{m=1}^n I_m}$ as the set of all
eventual-response protocols $P$ satisfying, for every run $r \in \TRP$:
\begin{enumerate}
\item
$\timpl_r|_{I_m}$ is a constant function, for each $m \in [n]$.
We denote its value by $t_{r,m}$.
\item
$\forall m \in [n-1]: t_{r,m+1} \ge t_{r,m}$.
\end{enumerate}
\end{defn}

\begin{remark}
$\OJR{\tilde{e},\tuple{I}} = \TCR{\tilde{e},\cup_{m=1}^n I_m,\delta}$, for
\[ \delta(i,j) \eqdef
\begin{cases}
0 &\exists\, k,l \in [n]: i \in I_k \And j \in I_l \And k \ge l \\
\infty &\text{otherwise.}
\end{cases} \]
\end{remark}

Ben-Zvi and Moses analyze this problem using a structure they call
``centibroom'', which generalizes both a centipede and a broom.

\fig{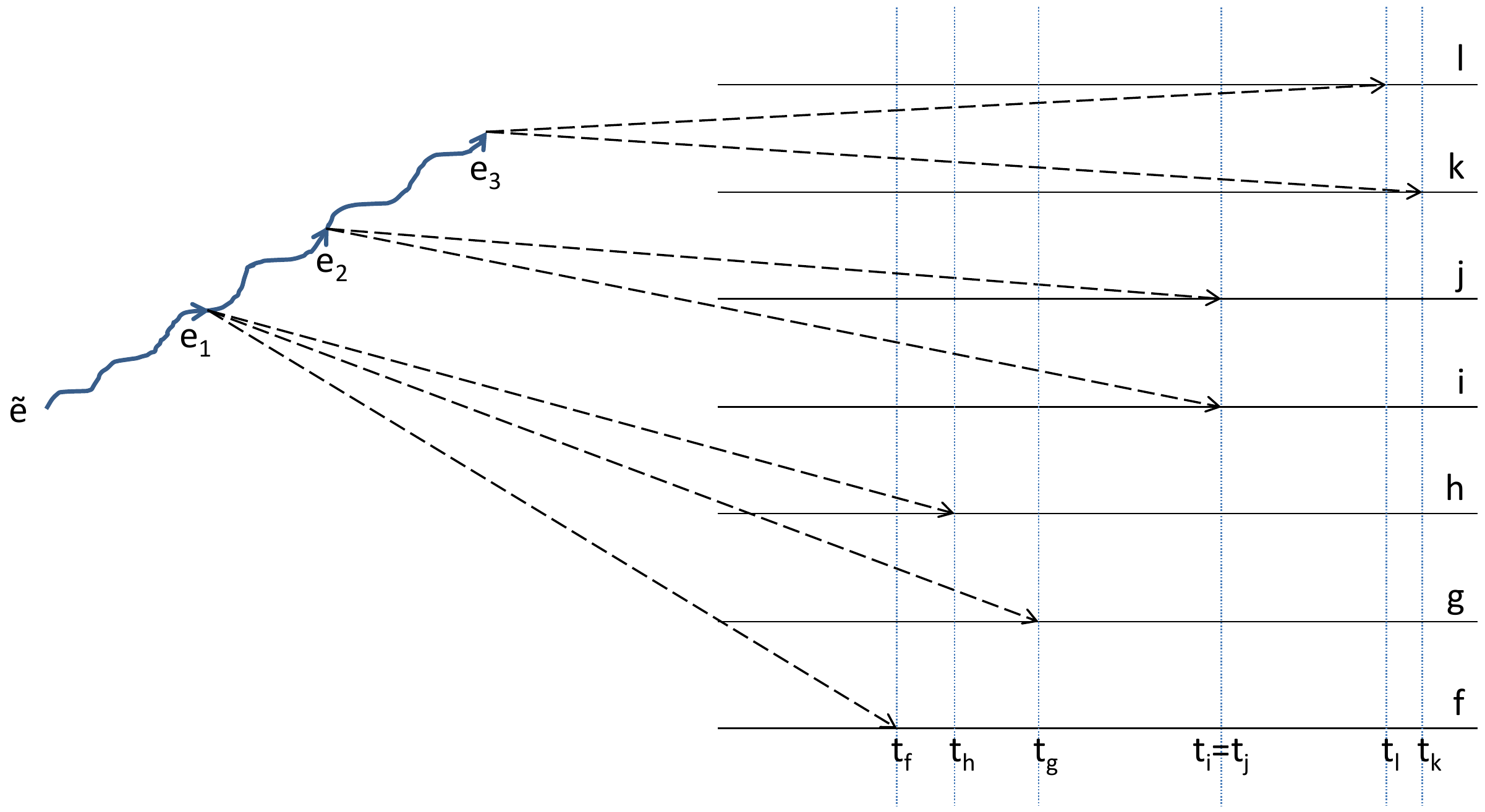}{Centibroom}{
  $(e_1,e_2,e_3)$ is an $\tilde{e}$-centibroom for $(\{f,g,h\},\{i,j\},\{k,l\})$ by
  $(t_f,t_g,t_h,t_i,t_j,t_k,t_l)$
}
\begin{defn}[Centibroom --- see \figref{uneven-centibroom}]
Given a context $\gamma$, a run $r \in \runs$, an external input
$\tilde{e} \in \externalinputs$, $n \in \mathbb{N}$,
and pairwise-disjoint non-empty sets of agents
$\tuple{I} \in (2^{\agents})^n$
with matching times $\tuple{t} \in \timeset^I$, where
$I \eqdef \cup_{m=1}^n I_m$,
we call an $n$-tuple of ND events $\tuple{e} \in \ND(r)^n$ an
``$\tilde{e}$-centibroom'' for $\tuple{I}$ by $\tuple{t}$ if the
following hold:
\begin{itemize}
\item $\tilde{e} \syncausal{r} e_1$ and $\forall m \in [n-1]: e_m \syncausal{r} e_{m+1}$.
\item $\forall m \in [n], i \in I_m: e_m \boundguarantee (i,t_i)$.
We call $\{(i,t_i)\}_{i \in I}$ the set of ``end nodes'' of this centibroom,
and call $\max\{\tuple{t}\}$ the ``horizon'' of this centibroom.
\end{itemize}
Given a time $t \in \timeset$, we call an $n$-tuple of ND events
$\tuple{e} \in \ND(r)^n$ an
``$\tilde{e}$-centibroom''
for $\tuple{I}$ by $t$, if it is an $\tilde{e}$-centibroom for $\tuple{I}$ by
$(t)^n$.
\end{defn}

\begin{thm}[Centibroom]\label{thm:centibroom}
In a discrete-time model,\footnote{
Once again, this is a key requirement here.}
let $\gamma$ be a context, let $n \in \mathbb{N}$,
let $\tuple{I} \in (2^{\agents})^n$
be pairwise-disjoint non-empty finite sets of agents, let
$\tilde{e} \in \externalinputs$ and
let $P \in \OJR{\tilde{e},\tuple{I}}$. Each $r \in \TRP$ contains
an $\tilde{e}$-centibroom for $\tuple{I}$ by $t_{r,n}$.
\end{thm}

Ben-Zvi and Moses prove \thmref{centibroom} by reducing to a
response-recalling protocol, showing that under the conditions of that part of
the theorem, at $t_{r,n}$ it holds that
$C_{I_n} \cdots C_{I_1} \tilde{e}$ (where $C_J$ means ``it is common
knowledge among $J$ that\ldots'' --- this will be formally defined in
\chapterref{fixed-point-approach}), and then using a
reduction to the first part of \thmref{centipede}.

In \cite{bzm4}, the following respective generalizations of ordered response
and simultaneous response were introduced:
(once again, we rephrase them to match the definitions and notation introduced
in this work.)

\begin{defn}[Weakly-Timed Response]
Given a context $\gamma$, an external input $\tilde{e} \in \externalinputs$, $n \in \mathbb{N}$,
agents $\tuple{\imath} \in {\agents}^n$ and finite time-differences
$\tuple{\varepsilon}=(\varepsilon_m)_{m=1}^{n-1} \in (\Delta \setminus \{-\infty, \infty\})^{n-1}$,
define the ``weakly-timed response''
problem $\WTR{\tilde{e},\tuple{\imath},\tuple{\varepsilon}} \subseteq
\ER{\tilde{e},\{i_m\}_{m=1}^n}$
as the set of all eventual-response protocols $P$ satisfying
$\timpl_r(i_{m+1}) \ge \timpl_r(i_m) + \varepsilon_m$ for
every $m \in [n-1]$ and for every triggered run $r \in \TRP$.
\end{defn}

\begin{remark}\label{remark:weakly-timed-is-timely}
$\WTR{\tilde{e},(i_m)_{m=1}^n,\tuple{\varepsilon}} = \TCR{\tilde{e},\{i_m\}_{m=1}^n,\delta}$, for
\[ \delta(i_k,i_l) \eqdef
\begin{cases}
-\varepsilon_l &k = l+1 \\
\infty &\text{otherwise.}
\end{cases} \]
\end{remark}

\begin{defn}[Tightly-Timed Response]
Given a context $\gamma$, an external input $\tilde{e} \in \externalinputs$
and a set of agents $I\subseteq\agents$ with
matching times $\tuple{t} \in \timeset^I$,
define the ``simultaneous response'' problem
$\TTR{\tilde{e},I,\tuple{t}} \subseteq \ER{\tilde{e},I}$ as the set of all
eventual-response protocols $P$ satisfying
$\timpl_r(i)-\timpl_r(j)=t_i-t_j$ for every
$i,j \in I$ and every run $r \in \RP$.
\end{defn}

\begin{remark}\label{remark:tightly-timed-is-timely}
$\TTR{\tilde{e},I,\tuple{t}} = \TCR{\tilde{e},I,\delta}$, for $\delta(i,j) \eqdef t_j-t_i$.
\end{remark}

Ben-Zvi and Moses present the following theorems in \cite{bzm4},
and prove them along the lines of their proofs of the first part of
\thmref{centipede} and the first part of \thmref{broom}, respectively.

\begin{thm}[Uneven Centipede]\label{thm:uneven-centipede}
In a discrete-time model,
let $\gamma$ be a context, let $n \in \mathbb{N}$, let
$\tuple{\imath} \in {\agents}^n$, let $\tuple{\varepsilon} \in (\Delta \setminus \{-\infty,\infty\}) ^{n-1}$,
let $\tilde{e} \in \externalinputs$
and let $P \in \OR{\tilde{e},\tuple{\imath}}$. Each $r \in \TRP$ contains
an $\tilde{e}$-centipede for $\tuple{\imath}$ by
$(\timpl_r(i_n)-\sum_{k=m}^{n-1} \varepsilon_k)_{m=1}^n$.
\end{thm}

\begin{thm}[Uneven Broom]\label{thm:uneven-broom}
In a discrete-time model,\footnote{
Yet again, this is a key requirement here.}
let $\gamma$ be a context, let $I \subseteq \agents$ be finite,
let $\tuple{t} \in \timeset^I$, let $\tilde{e} \in \externalinputs$
and let $P \in \SR{\tilde{e},I}$. Each $r \in \TRP$ contains
an $\tilde{e}$-broom for $I$ by $(\timpl_r(i))_{i \in I}$.
\end{thm}

\section{Adapting Some Machinery}\label{section:adapting-lemma-3}

Before approaching the timely-coordinated response problem
using the definitions surveyed in the previous section, we adapt some of the
machinery used by Ben-Zvi and Moses to obtain the results surveyed therein.
In order to do so, we introduce, yet again, some additional novel notation
and definitions.

\begin{defn}
Let $\gamma$ be a context and let $r \in \runs$.
\begin{itemize}
\item Given a time $t \in \timeset$,
we denote the set of all ND events occurring in $r$ no later
than $t$ by $\ND(r,t) \eqdef \{e \in \ND(r) \mid t_e \le t\}$.

\item Given an agent-time pair
$\theta \in \agents \times \timeset$ (resp.\ an event $\theta \in E(r)$),
we define the ``ND past'' of $\theta$ in $r$ as
$\PND{r}(\theta) \eqdef \{e \in \ND(r) \mid e \syncausal{r} \theta\}$.
Note that $\PND{r}(\theta) \subseteq \ND(r,t)$, where $\theta=(i,t)$
(resp.\ $t=t_\theta$). For an agent-time pair $(i,t)$, we sometimes write
$\PND{r}(i,t)$ instead of $\PND{r}((i,t))$, for readability.
\end{itemize}
\end{defn}

\begin{defn}[Subruns] Given a context $\gamma$, a protocol $P \in \protocols$,
a time $t \in \timeset$
and runs $r,r' \in \RP$,
we call $r'$ a ``$t$-subrun'' of $r$, and write $r' \tsubseteq r$, if
the initial states used for all agents in $r'$ and in $r$ are the same, and if
$\ND(r',t) \subseteq \ND(r,t)$.
(We omit
$P$ and $\gamma$ from this notation for readability, as they will be clear from
the discussion.)
For fixed $P$ and $t$, we note that $\tsubseteq$
is a quasi-order relation on $\RP$, in which two runs are in the same
equivalence class iff they are indistinguishable until $t$, inclusive.
\end{defn}

\begin{defn}[Retainable Subsets] Given a context $\gamma$, a protocol $P \in \protocols$,
a run $r \in \RP$ and a time $t \in \timeset$, we define the ``$t$-retainable'' subsets of
$\ND(r)$ as
\[\RND(r,t) \eqdef \{\ND(r',t) \mid r' \tsubseteq r \}
\subseteq 2^{\ND(r,t)}.\]
Furthermore, for every $E \in \RND(r,t)$, we denote
\[
r \tcap E \eqdef \{r' \tsubseteq r \mid \ND(r',t) = E \} \ne \emptyset,
\]
(Again, we omit $P$ and $\gamma$ from this notation as they will be inferred
from the discussion.)
We sometimes slightly abuse notation by using
$r \tcap E$ to refer to one such run and not to the whole
set, if the choice of representative is inconsequential. (This is often
the case, as $r \tcap E$ is an equivalence class of
$\tsubseteq$.)
\end{defn}

\begin{remark}\label{remark:retainable-properties}
Let $\gamma$ be a context, let $P \in \protocols$, let $r \in \RP$ and
let $t \in \timeset$.
By the above definitions:
\begin{itemize}
\item
$\forall t \in \timeset: \ND(r,t) \in \RND(r,t)$,
and $r \in r \tcap \ND(r,t)$.
\item
If $E \in \RND(r,t)$ and if $E' \in \RND(r \tcap E,t')$
for some $t' \in \timeset$ s.t.\ $t' \le t$,
then by definition, $E' \subseteq E$, $E' \in \RND(r,t')$ and
$(r \tcap E) \tprimecap E' = r \tprimecap E'$.
\end{itemize}
\end{remark}

\begin{claim}\label{claim:no-nd-between-tprime-t}
Let $\gamma$ be a protocol, let $P \in \protocols$, let
$r \in \RP$ and let $t,t' \in \timeset$.
If $t' \le t$, then $\RND(r,t') \subseteq \RND(r,t)$.
Furthermore, for every $E \in \RND(r,t')$, 
there exists $r' \in r \tcap E$ s.t.\ $\ND(r') \cap 
\externalinputs \subseteq E$.
\end{claim}

\begin{proof}[Proof Sketch]
Let $E \in \RND(r,t')$.
For the continuous-time model presented in \appref{continuous}, the claim
follows by applying the ``no foresight'' property to $r \tprimecap E$
at $t'$ and with $d=t$.
For the discrete-time model presented in \chapterref{discrete},
we construct a run $r' \in r \tcap E$ s.t.\
$\ND(r') \cap \externalinputs \subseteq E$, as follows:
$r'$ is identical
to $r \tprimecap E$ until $t'$, inclusive.
After $t'$, the agents behave in $r'$ according to $P$,
and the environment triggers no more ND events,
except for deliveries of sent messages that have an infinite
bound guarantee (as such non-deterministic deliveries must be triggered at some
time during the run, for the run to be legal).
Each such message is delivered at $\max\{t+1,t''+1\}$,
where $t''$ is the sending time of this message. It is straightforward to
inductively check that
the resulting run $r'$ is well defined and legal --- we omit
this cumbersome check,
which runs along similar lines of some of the proofs from \cite{bzm1},
in favor of the many, more interesting, pages ahead.
\end{proof}

We now present and adapt some machinery developed by
Ben-Zvi and Moses in their analysis\cite{bzm1,bzm2} of the ordered response
problem. Their discrete-time analysis essentially shows the following lemma,
which we rephrase using our notation. For the continuous-time model
presented in \appref{continuous}, the first part of this lemma is equivalent
to the ``no extrasensory perception'' property, and its second part follows
from the definition of a run.

\begin{lemma}\label{lemma:lemma-3}
Let $\gamma$ be a context, let $P \in \protocols$
and let $r \in \RP$. For every $t \in \timeset$,
\[
\RND(r,t) \supseteq \{ \PND{r}(i,t) \mid i \in \agents \},
\]
and for each $i \in \agents$, the state of $i$ at $t$ is
identical in $r$ and in all the runs $r \tcap \PND{r}(i,t)$.
\end{lemma}

By applying \claimref{no-nd-between-tprime-t}, we obtain the following
generalization of \lemmaref{lemma-3}.

\begin{cor}\label{cor:adjusted-lemma-3}
Let $\gamma$ be a context, let $P \in \protocols$
and let $r \in \RP$. For every $t \in \timeset$,
\[
\RND(r,t) \supseteq \{ \PND{r}(i,t') \mid i \in \agents \And t' \le t \},
\]
and for each $i \in \agents$ and each $t' \le t$, the state of $i$ at $t'$ is
identical in $r$ and in all the runs $r \tcap \PND{r}(i,t')$.
\end{cor}

We note, without a proof, that this result can be further generalized as
follows, at least for the cases listed below.

\begin{cor}\label{cor:most-generalized-lemma-3}
In a discrete-time model, or in a continuous-time model with finitely many
agents,
let $\gamma$ be a context, let $P \in \protocols$
and let $r \in \RP$.\footnote{
As discussed in \appref{continuous},
for certain ``nice'' protocols $P$, the requirement for only finitely many
agents in a continuous-time model may be relaxed to the requirement
that $\inf(\contextbounds)>0$.
}
For every $t \in \timeset$,
\[
\RND(r,t) \supseteq \{E \subseteq \ND(r,t) \mid
\forall e \in E:
E \supseteq \PND{r}(e) \},
\]
with equality if $P$ is a full-information protocol.
Furthermore, for every $E \in \RND(r,t)$ and for every
$(i,t') \in \agents \times \timeset$, if $\PND{r}(i,t') \subseteq E$, then 
the state of $i$ at $t'$
is identical in $r$ and in all the runs $r \tcap E$.
(If P is a full-information protocol, then the converse holds as well.)
\end{cor}

We do not require \corref{most-generalized-lemma-3}, though, as
\corref{adjusted-lemma-3} suffices for all the proofs that we give below.

\section{Analyzing Timely-Coordinated Response}

We now turn to define the structure that stands at the heart
of our syncausal analysis of the timely-coordinated response problem.

\fig{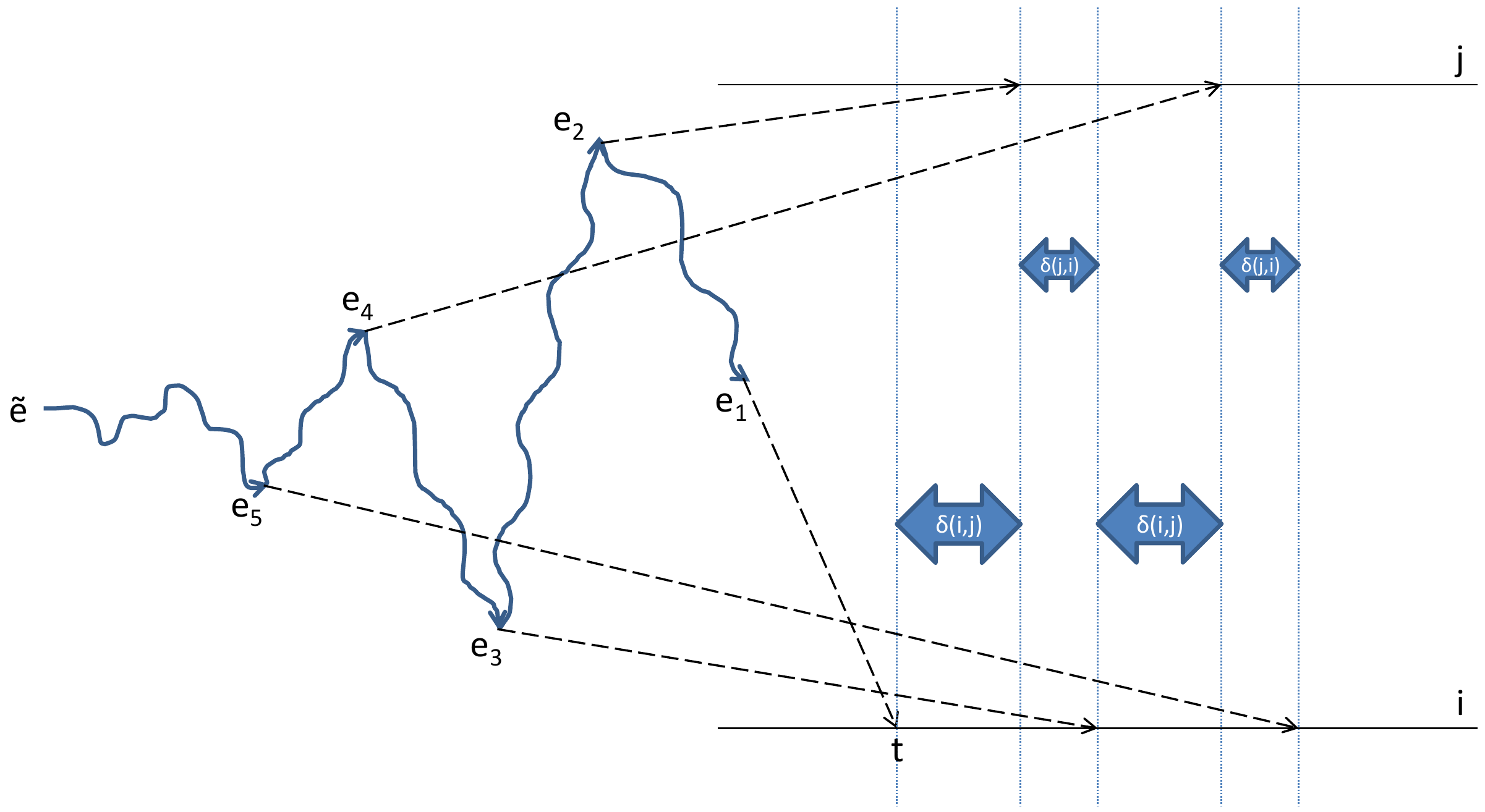}{Path-traversing centipede}{
  $(e_m)_{m=1}^5$ is an $((i,j,i,j,i),\delta)$-traversing $\tilde{e}$-centipede by $t$.
}
\begin{defn}[Path-Traversing Centipede --- see \figref{path-traversing-centipede}]
Given a TCR-spec $\TCRspec$,
a path $\apathfull \in \dpaths$
and a run $r \in \runs$,
we call an $n$-tuple of ND events $\tuple{e} \in \ND(r)^n$ a
``$(\apath,\delta)$-traversing $\tilde{e}$-centipede'' by $t$ if the
following hold:
\begin{itemize}
\item $\tilde{e} \syncausal{r} e_n$ and $\forall m \in [n-1]: e_{m+1} \syncausal{r} e_m$.
\item
$\forall m \in [n]: e_m \boundguarantee (p_m,t+\dlength((p_k)_{k=1}^m))$.
We call $\{(p_m,t+\dlength((p_k)_{k=1}^m))\}_{m=1}^n$ the set of ``end nodes''
of this centipede.
\end{itemize}
\end{defn}

\begin{remark}
$\tuple{e}^{\mathit{rev}}=(e_{n-m+1})_{m=1}^n$ is an
$\tilde{e}$-centipede for
$\apath^{\mathit{rev}}$
(as a tuple of agents) by $\tuple{t}^{\mathit{rev}}$, where for every
$m \in [n]$, $t_m \eqdef t+\dlength((p_k)_{k=1}^m)$.
Thus, $t_{m+1}=t_m + \delta(p_m,p_{m+1})$ for every $m \in [n-1]$.
\end{remark}

\begin{remark}\label{remark:path-traversing-centipede-properties}
Let $\TCRspec$ be a TCR-spec,
let $\apath \in \dpaths$ and
let $r \in \runs$.
By the above definition:
\begin{itemize}

\item
No $(\apath,\delta)$-traversing $\tilde{e}$-centipede exists in $r$,
if $\apath$ traverses an
edge with a weight of $-\infty$ in $G_{\delta}$.

\item Any $(\apath,\delta)$-traversing $\tilde{e}$-centipede by $t$ in $r$
is also a $(\apath,\delta')$-traversing $\tilde{e}$-centipede by $t'$ in $r$,
for every $\delta' \ge \delta$ and every $t' \ge t$, by the locality property
of bound guarantee. This justifies the phrasing ``path-traversing centipede
{\em by} $t$''.

\item
Let $\tuplefull{e}$ be a
$(\apath,\delta)$-traversing $\tilde{e}$-centipede
by $t$ in $r$, then $(e_m)_{m=k}^n$ is a
$((p_m)_{m=k}^n,\delta)$-traversing $\tilde{e}$-centipede by
$t + \dlength((p_m)_{m=1}^k)$ in $r$, for every $k \in [n]$.
We call this path-traversing centipede the ``$k$-suffix'' of $\tuple{e}$.

\end{itemize}
\end{remark}

The following theorem, once stated, may be proven using the tools that are
applied in \cite{bzm4} for proving \thmref{uneven-centipede}.
We provide a somewhat different and more concise proof here, also for the sake
of emphasizing the fact that the approach studied in this chapter requires no
direct use of the concept of knowledge.

\begin{thm}[Path-Traversing Centipede]\label{thm:path-traversing-centipede}
Let $\TCRspec$ be a TCR-spec and
let $P \in \TCR{\tilde{e},I,\delta}$.
Each $r \in \TRP$
contains a $(\apath,\delta)$-traversing $\tilde{e}$-centipede
by $\timpl_r(p_1)$, for every $\apath \in \dpaths$.
\end{thm}

\begin{proof}
By induction on $n$, the number of vertices in $\apath$. ($\apathfull$.)

Base: If $n=1$, denote $i\eqdef p_1$ (and thus, $\tuple{p} = (i)$).
Since $r \in \TRP$, we claim that $\tilde{e} \syncausal{r} (i,\timpl_r(i))$.
Indeed, by \corref{adjusted-lemma-3}
and by \claimref{no-nd-between-tprime-t},
there exists a run $r' \in \RP$ for which
$\ND(r',\timpl_r(i))=\PND{r}(i,\timpl_r(i))$ and
in which the only occurring external inputs are those that are
in $\PND{r}(i,\timpl_r(i))$.
Furthermore, both the state of $i$, and the events observed by it, are identical
in $r$ and in $r'$ up to and including $\timpl_r(i)$, and
thus $\timpl_{r'}(i)=\timpl_r(i)<\infty$.
By correctness of $P$, this implies $r' \in \TRP$,
and thus $\tilde{e} \in \ND(r) \cap \externalinputs
\subseteq \PND{r}(i,\timpl_r(i))$, as required.
Thus, there exists a syncausal path in $r$ from $\tilde{e}$ to $(i,\timpl_r(i))$.
Denote the latest among the ND event along this path by $e \in \ND(r)$.
By definitions
of syncausality and of bound guarantee,
$\tilde{e} \syncausal{r} e \boundguarantee (i,\timpl_r(i))$.
Thus, $(e)$ is a path-traversing centipede as required.

Induction step: Assume $n \ge 2$.
Set $i\eqdef p_1$, $j\eqdef p_2$,
$E \eqdef \PND{r}(i,\timpl_r(i))$ and
$t \eqdef \max\{\timpl_r(i),\timpl_r(i)+\delta(i,j)\} < \infty$.
As $\timpl_r(i) \le t$, by \corref{adjusted-lemma-3} we obtain that
$E \in \RND(r,t)$,
and that the state of $i$ at $\timpl_r(i)$ is the same in
$r \tcap E$ and in $r$. Therefore, $\timpl_{\smash{r \tcap E}}(i) = \timpl_r(i)$,
and thus, by correctness of $P$, we obtain
\[\timpl_{\smash{r \tcap E}}(j) \le \timpl_{\smash{r \tcap E}}(i) +
\delta(i,j) = \timpl_r(i) + \delta(i,j).\]

By the induction hypothesis, there exists a
$((p_m)_{m=2}^n,\delta)$-traversing
$\tilde{e}$-centipede $(e_m)_{m=2}^n$
by $\timpl_{\smash{r \tcap E}}(j)$ (and thus, by
\remarkref{path-traversing-centipede-properties},
by $\timpl_r(i) + \delta(i,j)$) in $r \tcap E$, and thus also in $r$
(as $\timpl_r(i) + \delta(i,j) \le t$).
To complete our proof, we note that
\begin{align*}
e_2 \in&\:\PND{r \tcap E}(j,\timpl_r(i) + \delta(i,j)) \subseteq & \text{by definition of $\mathit{PND}$} \\
\subseteq&\:\ND(r \tcap E, \timpl_r(i) + \delta(i,j)) \subseteq & \text{as $\timpl_r(i) + \delta(i,j) \le t$} \\
\subseteq&\:\ND(r \tcap E, t) = E = \PND{r}(i,\timpl_r(i)).
\end{align*}
and therefore $e_2 \syncausal{r} (i,\timpl_r(i))$.
As in the induction base,
there exists $e_1 \in \ND(r)$ s.t.\ $e_2 \syncausal{r} e_1
\boundguarantee (i,\timpl_r(i))$. Thus, $(e_m)_{m=1}^n$
is a path-traversing centipede as required.
\end{proof}

It should be noted that by \corref{canonical-form-and-tcr}, a
$(\apath,\hat{\delta})$-traversing
$\tilde{e}$-centipede is also implied by \thmref{path-traversing-centipede}
for every $\apath \in \dhatpaths$ under the conditions of that theorem.
Furthermore, this result is at least as strong as the verbatim result of that
theorem for $\delta$,
by minimality of the canonical form,
by \remarkref{path-traversing-centipede-properties} and as
$\dpaths \subseteq \dhatpaths$.
It may be readily verified that if the
distance between every pair of agents in $G_{\delta}$ is attained,
then these results are
in fact equivalent, as any
path-traversing centipede guaranteed by \thmref{path-traversing-centipede}
for $\hat{\delta}$ is a (possibly trivial) subcentipede (i.e.\ subtuple)
of a path-traversing centipede directly guaranteed by it
for $\delta$ (for a possibly
different path). Henceforth, whenever minimizing the times of the end nodes
of the guaranteed path-traversing centipede is of the essence (as is the
case in e.g.\ Corollaries \ref{cor:path-traversing-centipede-implies-broom}
and \ref{cor:path-traversing-centipede-implies-centibroom}
and \claimref{finite-influence-implies}), we indeed apply
\thmref{path-traversing-centipede} using $\hat{\delta}$.

We now apply \thmref{path-traversing-centipede} to deduce an optimal
response logic for the timely-coordinated response problem
in shared-clock models.

\begin{cor}\label{cor:path-traversing-centipede-optimality}
In a shared-clock model, let $\TCRspec$ be a TCR-spec.
An optimal response logic for solving $\TCR{\tilde{e},I,\delta}$ is, for every
$i \in I$: ``respond at the earliest time by which
a $(\apath,\delta)$-traversing $\tilde{e}$-centipede exists for
every path $\apath \in \dpaths$ starting at $p_1=i$''.
\end{cor}

\begin{proof}
Assume that $\TCR{\tilde{e},I,\delta}$ is solvable and let
$P \in \TCR{\tilde{e},I,\delta}$.
W.l.o.g., $P$ is a full-information protocol.
Let $P'$ be the (full-information)
protocol obtained by endowing $P$ with the above-defined
response logic.
We first prove the optimality of $P'$ in each triggered run, and then
prove that it indeed solves $\TCR{\tilde{e},I,\delta}$.

Let $r \in \TRP$ and $r' \in \TR{P'}$ be two runs matched under the natural
isomorphism between $\RP$ and $\RPTAG$ and let $i \in I$.
By \thmref{path-traversing-centipede}, all path-traversing
centipedes required for $i$ to respond according to $P'$ exist in $r$
(and hence in $r'$) by $\timpl_r(i)$ (at the latest), and therefore
$\timpl_{r'}(i) \le \timpl_r(i)$.
(The fact that $P'$ is a full-information protocol, together with the existence
of a shared clock,
guarantees that if such path-traversing centipedes exist by some
$t \in \timeset$, then $i$ can deduce this at $t$.
For the continuous-time model presented in \appref{continuous}, the fact
that $i$ is enabled at every supremum of times at which it observes events
allows $i$ to respond at the required time.)

We now prove that $P' \in \TCR{\tilde{e},I,\delta}$:
Let $r' \in \RPTAG$. Obviously, no path-traversing $\tilde{e}$-centipedes
exist in $r'$ if $\tilde{e}$ does not occur.
Therefore, if $r' \notin \TR{P'}$ then $\timpl_r \equiv \infty$.
We are left with the case in which $r' \in \TR{P'}$.
Denote by $r \in \TRP$ the run of $P$ matching $r'$ under the natural
isomorphism between $\RP$ and $\RPTAG$.
By the first part of this proof, $\timpl_{r'} \le \timpl_r < \infty$.
Let $(i,j) \in \distinctpairs{I}$ s.t.\ $\delta(i,j) < \infty$.
For every $\apath \in \dpaths$ s.t.\ $p_1=j$,
denote by $\apath' \in \dpaths$ the path commencing at $i$ and whose
$2$-suffix is $\apath$.
By definition of $P'$, a $(\apath',\delta)$-traversing
$\tilde{e}$-centipede exists in $r'$ by
$\timpl_{r'}(i)$ (at the latest), and thus its $2$-suffix, which
is a $(\apath,\delta)$-traversing $\tilde{e}$-centipede,
exists in $r'$ by $\timpl_{r'}(i) + \delta(i,j)$ (at the latest).
Hence, all path-traversing centipedes
required for $j$'s response according to $P'$ exist in $r'$ by
$\timpl_{r'}(i) + \delta(i,j)$, and thus
$\timpl_{r'}(j) \le \timpl_{r'}(i) + \delta(i,j)$, as required.\footnote{
The attentive reader may notice a conceptual similarity between this
argument and the second part of the proof of \lemmaref{implementable-iff}.}
\end{proof}

We conclude this chapter with a practical note that motivates some of our
discussion in \chapterref{practical}. In that chapter, we derive somewhat more
practical results from the above discussion, for some naturally-occurring
models that we define.

The results of Ben-Zvi and Moses that are surveyed in the beginning of this
chapter imply that for each of the coordinated response problems
they have studied,
it is enough for an agent $i$ to deduce the existence of a single,
simple, syncausal
structure in order to respond according to the optimal response logic for this
problem.
In contrast, from the definition of the optimal response logic for the
timely-coordinated response problem from
\corref{path-traversing-centipede-optimality}, it may seem
that in the case of a general constraining function $\delta$ (i.e.\ $\delta$
that does
not reduce to e.g.\ one of the special cases studied by Ben-Zvi and Moses),
for an agent $i$ to respond according to this logic, $i$ is always required
to check for infinitely many, arbitrarily long, path-traversing centipedes
(using infinitely many
facts stored in the memory/state of $i$).
While in the general case this is true, the following remark shows that
for any constraining function $\delta$, in some
cases finitely many syncausal
structures may imply the existence of all
(infinitely many) path-traversing centipedes required for $i$'s response.
Conceptually, this means that finitely
many checks (of finitely many facts) may provide enough
information for $i$ to respond at a specific time according to this logic.
In \chapterref{practical}, we show that under certain practical
assumptions, finitely many not-much-more-complicated checks always suffice.

\begin{remark}\label{remark:broom-implies-path-traversing-centipede}
Let $\TCRspec$ be a TCR-spec s.t.\ $\delta$ is implementable, and
let $r \in \runs$.
If $e \in \ND(r)$ is an $\tilde{e}$-broom for $I$ in $r$
by $\tuple{t} \in \timeset^I$
s.t.\ $\sup(\tuple{t})<\infty$,
then for every $\apathfull \in \dpaths$,
$(e)^n$ is a $(\apath,\delta)$-traversing
$\tilde{e}$-centipede in $r$ by
\[
\max_{k \in [n]}\{t_{p_k}-\dlength((p_m)_{m=1}^k)\} \le
\max_{k \in [n]}\{t_{p_k}-\hat{\delta}(p_1,p_k)\} \le
\sup(\tuple{t}) - \inf(\hat{\delta}|_{\{p_1\} \times I}) < \infty.
\]
In particular, for every agent $i \in I$,
$(e)^n$ is a $(\apath,\delta)$-traversing $\tilde{e}$-centipede
by \linebreak $\sup(\tuple{t}) - \inf(\hat{\delta}|_{\{i\} \times I})$ (which
is finite by implementability of $\delta$ and by \lemmaref{implementable-iff}),
for every $n \in \mathbb{N}$ and every $\apathfull \in \dpaths$ starting
at $p_1=i$.
Thus, by this time $i$ will
have received information guaranteeing that $e$ had occurred, which will have
given $i$ enough information in order to respond by that time according to the
optimal response logic presented in
\corref{path-traversing-centipede-optimality}.\footnote{
We note that it is possible to construct an alternative argument as to
why there exists, under certain conditions, a solving finite-memory protocol
according to which each $i$ responds by that time. Such an argument may
be constructed by combining
a variation of the second part of \thmref{broom}, with
the second part of \lemmaref{implementable-iff} and with
the second part of \claimref{solvable-iff-implementable}. Conversely,
\remarkref{broom-implies-path-traversing-centipede} may be used
to construct an alternative proof for parts of \lemmaref{implementable-iff}.
}
\end{remark}

\chapter{The Fixed-Point Approach}\label{chapter:fixed-point-approach}

We now set aside, for the moment, the results of \chapterref{syncausality-approach}
and embark on a parallel, independent analysis of the timely-coordinated
response problem
using fixed-point analysis. While the basis of this analysis follows
the lines of \cite[Section~11.6]{book}, we formalize it here using events,
along the approach of Aumann\cite{aumann}, instead of
using the temporal-epistemic logic tools used in \cite{book}.
(The treatment in either form is analogous, although it is somewhat more
concise for our case with the notation used below, which facilitates
the study of fixed points.)

\section{Background}
In this section, we survey previous definitions and results from \cite{book},
upon which our analysis below is founded, reformulating them using events,
and adapting them to our notation.

\subsection{Events, Knowledge and Common Knowledge}

In order to begin our discussion, we define the space in which we work.

\begin{defn}[Space]
Let $\gamma$ be a context and let $R \subseteq \runs$.
We define $\points \eqdef R \times \timeset$ and $\pointsets \eqdef 2^{\points}$.
\end{defn}

As in probability theory, we represent events using the set of points
(i.e.\ run-time pairs)
in which they hold. (In contrast to probability theory, though,
we do not need to define a measure
on the set of events, so we choose to allow any subset of $\points$
to constitute an event.)
For example, we may define the event ``$i$ is responding''
for some $i \in \agents$, which is formally associated with all points
$(r,t) \in \points$ s.t.\ $i$ responds at $t$ in $r$.

We now incorporate the concept of knowledge into our discussion.
Given an event $\psi \in \pointsets$, we wish to define, for some agent $i$,
the event ``$i$ knows that $\psi$ holds'' (e.g.\ $i$ knows that $j$ is
responding). Choosing how to formalize a concept as abstract and as subjective
as knowledge is not a simple issue. We present below what has become a standard
definition for knowledge, and avoid discussing its relation to the abstract,
philosophical, concept of knowledge.
Intuitively, by this definition, at $(r,t)$ $i$ knows
$\psi$ iff $\psi$ holds at all possible points $(s,t')$ that $i$
cannot distinguish from $(r,t)$.

\begin{defn}[Knowledge]\label{defn:knowledge}
Let $\gamma$ be a context, let $R \subseteq \runs$ and let $i \in \agents$.
\begin{enumerate}
\item
We partition $\points$ into equivalence classes according to the state of
$i$, s.t.\ $p,q \in \points$ are in the same equivalence class iff
the state of $i$ is the same in $p$ and in $q$.
In a shared-clock model, we additionally demand that the time be the same at
$p$ and at $q$.
For $p \in \points$, we denote the equivalence class of $p$ by
$\mathcal{S}_i(p)$.
\item
Define
\vspace*{-1em} 
\functiondefn{K_i}{\pointsets}{\pointsets}{\psi}{\{p \in \points \mid \mathcal{S}_i(p) \subseteq \psi\}.}
\end{enumerate}
While both definitions, of $\mathcal{S}_i$ and of $K_i$, depend on
$R$, we omit $R$ from these notations, for readability, as the set of runs
will be clear from the discussion. We follow this convention when presenting
some other definitions in this, and in the following, chapter as well.
\end{defn}

We now present a few immediate (and well-known)
properties of the knowledge operator.
The first one, sometimes referred to as the ``Truth Axiom for Knowledge'',
intuitively means that ``whenever anyone knows something, then
it is true''. The second property, sometimes referred to as the
``Positive Introspection Axiom'', which
intuitively means ``whenever $i$ knows
something, then $i$ knows that it knows it'',
has been the subject of quite a few philosophically discussions.
As we have done when presenting the definition for knowledge, and as we
will continue to do below, we present it, and avoid
discussing any philosophical consequences thereof.

\begin{remark}\label{remark:knowledge-properties}
Let $\gamma$ be a context, let $R \subseteq \runs$ and let $i \in \agents$.
By \defnref{knowledge}, the knowledge operator $K_i$ satisfies:
\begin{itemize}
\item
Truth Axiom for Knowledge:
$K_i(\psi) \subseteq \psi$, for every event $\psi \in \pointsets$.
\item
Positive Introspection Axiom: $K_i(K_i(\psi)) = K_i(\psi)$, for every event $\psi \in \pointsets$.
\item
Monotonicity:
$\ \psi\subseteq\phi \ \Rightarrow\ 
K_i(\psi) \subseteq K_i(\phi),\ $ for every two events $\psi,\phi \in \pointsets$.
\item $K_i$ commutes with intersection:
$K_i(\cap\Psi) = \bigcap\{K_i(\psi) \mid \psi \in \Psi\}$, for every set of events $\Psi \subseteq \pointsets$.
\end{itemize}
\end{remark}

We now build upon \defnref{knowledge} and define the notion of
``everybody knows''.

\begin{defn}
Let $\gamma$ be a context, let $R \subseteq \runs$ and let $I \subseteq \agents$
be a set of agents.
Define
\vspace*{-1em} 
\functiondefn{E_I}{\pointsets}{\pointsets}{\psi}{\bigcap_{i \in I} K_i(\psi).}
\end{defn}

A truth axiom, analogous to the one presented in \remarkref{knowledge-properties}, readily holds for $E_I$ as well. In addition, $E_I$ is monotone and
commutes with intersection. However, it is not idempotent.

We are now ready to define common knowledge. One classic, constructive
definition of common knowledge\cite{common-knowledge-nested}
is the following, defining
that an event is common knowledge to a set of agents when all know
it, all know that all know it, etc.

\begin{defn}[Common Knowledge]\label{defn:common-knowledge}
Let $\gamma$ be a context, let $R \subseteq \runs$ and let $I \subseteq \agents$.
Define
\vspace*{-1em} 
\functiondefn{C_I}{\pointsets}{\pointsets}{\psi}{\bigcap_{n=1}^{\infty} {E_I}^n(\psi),}
where ${E_I}^0(\psi)=\psi$ and ${E_I}^{n}(\psi)=E_I({E_I}^{n-1}(\psi))$ for
every $n \in \mathbb{N}$.
\end{defn}

It may be readily verified that the common knowledge operator satisfies the
obvious analogues of all properties of the knowledge operator that are
presented in \remarkref{knowledge-properties} (including idempotence\footnote{
In
fact, $C_I$ is the (coordinate-wise) greatest idempotent operator s.t.\
$C_I \subseteq E_I$.}).

A classic result\cite{halpern-moses-1990} regarding common knowledge
is that it relates tightly to simultaneous response, in the sense that e.g.\
in order to coordinate a
simultaneous response among a set of agents, they must all have common knowledge
of the response when it occurs. Conversely, whenever
common knowledge of a fact arises among a set of agents, it does so
simultaneously for all agents. \defnref{perfect-coordination} and
\thmref{perfect-coordination} below formalize this intuition, but before
we present them, we turn to a few more definitions.

As noted above, the truth axiom for knowledge, presented in
\remarkref{knowledge-properties}, implies that whenever some fact is known
to someone, the fact is true as well.
Certain events, such as, for some agent $i$, ``$i$ is responding (right now)'',
have the converse property as
well, i.e.\ they are known to $i$ whenever they hold.
The following definition characterises such events.

\begin{defn}[Local Event]
Let $\gamma$ be a context, let $R \subseteq \runs$ and let $i \in \agents$.
An event $\psi \in \pointsets$
is said to be ``local'' to $i$ if $K_i(\psi) = \psi$.
\end{defn}

\begin{remark}\label{remark:knowledge-is-local}
By the positive introspection axiom presented in \remarkref{knowledge-properties},
$K_i(\psi)$ is local to $i$, for every $i \in \agents$ and for every
$\psi \in \pointsets$.
\end{remark}

An important property of events that are local to $i$,
for some agent $i \in \agents$, is that $i$
may act upon them, i.e.\ the response logic of $i$ in a protocol may be defined
by specifying that $i$ should respond whenever some given local event for $i$
holds. ($i$ may do that, as locality of this event guarantees that whether it
holds or not at some time $t \in \timeset$ is determined 
by the state of $i$ at $t$.) As we are interested in coordination,
though, we are usually interested in specifying a joint response logic for a
set of agents $I \subseteq \agents$, i.e.\ a
response logic for each $i \in I$, in which the response times of the various
agents are coordinated in some way. One way to specify such a joint response
logic, therefore, is to specify, for each $i \in I$, a local event for $i$.
Such a collection of specifications is called an ``ensemble''.

\begin{defn}[Event Ensemble]
Let $\gamma$ be a context, let $R \subseteq \runs$ and let $I \subseteq \agents$. An $I$-tuple of events $\tuple{\ensemble}=(\ensemble_i)_{i \in I} \in {\pointsets}^I$ is called an
ensemble if $\ensemble_i$ is local to $i$ for each $i \in I$.
\end{defn}

It should be noted, though, that ensembles are useful beyond specifying response
logics, as they may be used to study the coordination of passive events as well,
i.e.\ coordination of times at which agents become aware of some fact, or
at which they observe an event (e.g.\ receive some message).

We are now ready to survey the results of \cite{halpern-moses-1990} relating
common knowledge and simultaneity, as formulated for ensembles in
\cite[Section~11.6]{book}. While phrasing the following theorem, and
henceforth, we use the following shorthand notation: $\cup\tuple{\xi}\eqdef
\cup_{i \in I} \xi_i$, for every $\tuple{\xi} =\xi_{i \in I} \in {\pointsets}^I$.

\begin{defn}[Perfect Coordination]\label{defn:perfect-coordination}
Let $\gamma$ be a context, let $R \subseteq \runs$ and
let $I \subseteq \agents$.
An ensemble $\tuple{\ensemble} \in {\pointsets}^I$ is said to be
``perfectly coordinated'' if $\ensemble_i=\ensemble_j$ for every $i,j \in I$.
\end{defn}

\begin{thm}\label{thm:perfect-coordination}
Let $\gamma$ be a context, let $R \subseteq \runs$ and let $I \subseteq \agents$.
\begin{enumerate}
\item
For every event $\psi$, the ensemble $(K_i(C_I(\psi)))_{i \in I}$
is perfectly coordinated.
\item
If $\tuple{\ensemble} \in {\pointsets}^I$ is a perfectly coordinated ensemble, then
$\ensemble_i \subseteq K_i(C_I(\cup\tuple{\ensemble}))$ for every~$i \in I$.
\item
If $\tuple{\ensemble} \in {\pointsets}^I$ is a perfectly coordinated ensemble, then
$\cup\tuple{\ensemble} \subseteq C_I(\cup\tuple{\ensemble})$.
\end{enumerate}
\end{thm}

In \chapterref{equivalence}, we show that the analysis of
\chapterref{syncausality-approach}
has led us, in a sense, to a definition that is
similar to \defnref{common-knowledge}.

\subsection{Fixed-Point Analysis}

Another classic definition\cite{common-knowledge-fixed-point}
for common knowledge, which is known to be equivalent, is the following,
defining it as the greatest fixed point of a function on
$\pointsets$.

\begin{thm}[Common Knowledge as a Greatest Fixed Point]\label{thm:common-knowledge-fixed-point}
Let $\gamma$ be a context, let $R \subseteq \runs$ and let $I \subseteq \agents$.
$C_I(\psi)$ is the greatest fixed point of
the function $x \mapsto E_I(\psi \cap x)$, for every event
$\psi \in \pointsets$.
\end{thm}

Based on this definition, Moses and Halpern\cite{halpern-moses-1990} defined
two variants of Common Knowledge, matching two weaker forms of coordination.
We now present their results, as formulated for ensembles in
\cite[Section~11.6]{book}.

\begin{defn}[Eventual Coordination]\label{defn:eventual-coordination}
Let $\gamma$ be a context, let $R \subseteq \runs$ and let
$I \subseteq \agents$.
An ensemble $\tuple{\ensemble} \in {\pointsets}^I$ is said to be
``eventually coordinated'' if
for every $i,j \in I$ and for every $(r,t) \in \ensemble_i$, there exists $t' \in \timeset$
s.t.\ $(r,t') \in \ensemble_j$.
\end{defn}

\begin{defn}
Let $\gamma$ be a context and let $R \subseteq \runs$.
Define
\functiondefn{\sometime}{\pointsets}{\pointsets}{\psi}{\{(r,t) \in \points \mid
\exists\, t' \in \timeset : (r,t') \in \psi\}.\footnotemark}
($\sometime(\psi)$ is the event ``$\psi$ eventually holds at some time during 
the current run, be it past, present or future.'')
\footnotetext{
We use the symbol $\sometime$ instead of the standard temporal logic
notation $\mnmeddiamond$, in order to emphasize that $t'$ may be smaller than $t$.
}
\end{defn}

\begin{thm}\label{thm:eventual-coordination}
Let $\gamma$ be a context, let $R \subseteq \runs$ and
let $I \subseteq \agents$.
\begin{enumerate}
\item
For every $\psi \in \pointsets$, the function
$x \mapsto \cap_{i \in I}\sometime(K_i(\psi \cap x))$ has a greatest
fixed point. Denote it by $C_I^{\sometime}(\psi)$ (``eventual common knowledge''
of $\psi$ by $I$).
\item
For every event $\psi$, the ensemble $(K_i(C_I^{\sometime}(\psi)))_{i \in I}$
is eventually coordinated.
\item
If $\tuple{\ensemble} \in {\pointsets}^I$ is an eventually coordinated ensemble, then
$\ensemble_i \subseteq K_i(C_I^{\sometime}(\cup\tuple{\ensemble}))$ for every $i \in I$.
\item
If $\tuple{\ensemble} \in {\pointsets}^I$ is an eventually coordinated ensemble, then
$\cup\tuple{\ensemble} \subseteq C_I^{\sometime}(\cup\tuple{\ensemble})$.
\end{enumerate}
\end{thm}

Another variant of common knowledge, also defined and studied by Halpern and
Moses, relates to an approximation of perfect coordination.

\begin{defn}[$\upvarepsilon$\mbox{-}Coordination]\label{defn:epsilon-coordination}
Let $\gamma$ be a context, let $R \subseteq \runs$, let $I \subseteq \agents$
and let $\varepsilon\ge0$.
An ensemble $\tuple{\ensemble} \in {\pointsets}^I$ is said to be
``$\varepsilon$\mbox{-}coordinated'' if
for every $i \in I$ and for every $(r,t) \in \ensemble_i$, there exists an interval
$T \subseteq \timeset$ of length at most $\varepsilon$, s.t.\ $t \in T$
and s.t.\ for every $j \in I$ there exists $t' \in T$ s.t.\ $(r,t') \in \ensemble_j$.
\end{defn}

We note that $0$\mbox{-}coordination is the same as perfect coordination,
and thus the following theorem also implies \thmref{perfect-coordination}
as a special case thereof.

\begin{thm}\label{thm:epsilon-coordination}
Let $\gamma$ be a context, let $R \subseteq \runs$, let $\varepsilon\ge0$
and let $I \subseteq \agents$.
Define
\functiondefn{E_I^{\varepsilon}}{\pointsets}{\pointsets}{\psi}{\left\{(r,t) \in \points \:\middle|\: \exists T \subseteq \timeset : \begin{matrix} t \in T \And \sup\{T-T\} \le \varepsilon \And \\ \forall i \in I \: \exists\, t' \in T : (r,t') \in K_i(\psi)\end{matrix}\right\}.}
\begin{enumerate}
\item
For every $\psi \in \pointsets$, the function
$x \mapsto E_I^{\varepsilon}(\psi \cap x))$ has a
greatest fixed point. Denote it by $C_I^{\varepsilon}(\psi)$ (``$\varepsilon$\mbox{-}common knowledge'' of $\psi$ by $I$).
\item
For every event $\psi$, the ensemble
$(K_i(C_I^{\varepsilon}(\psi)))_{i \in I}$ is $\varepsilon$\mbox{-}coordinated.
\item
If $\tuple{\ensemble} \in {\pointsets}^I$ is an $\varepsilon$\mbox{-}coordinated ensemble, then
$\ensemble_i \subseteq K_i(C_I^{\varepsilon}(\cup\tuple{\ensemble}))$ for every $i \in I$.
\item
If $\tuple{\ensemble} \in {\pointsets}^I$ is an $\varepsilon$\mbox{-}coordinated ensemble, then
$\cup\tuple{\ensemble} \subseteq C_I^{\varepsilon}(\cup\tuple{\ensemble})$.
\end{enumerate}
\end{thm}

The attentive reader may notice, by now, a pattern forming in the similarity
between \defnref{perfect-coordination} and the function defined in
\thmref{common-knowledge-fixed-point},
between \defnref{eventual-coordination} and the function defined in
\thmref{eventual-coordination}
and between \defnref{epsilon-coordination} and the function defined in
\thmref{epsilon-coordination}.

\section{\texorpdfstring{$\updelta$}{Delta}-Common Knowledge}

Having completed our survey of some previous, relevant, results and definitions,
we are now ready to start extending them.
Recall that the interdependencies between the response times of different agents
in the timely-coordinated response problem are captured by an
implementation-spec $\implspec$, i.e.\ a set of agents $I$ and a constraining
function
$\delta:\distinctpairs{I}\rightarrow\Delta$ from ordered pairs of distinct
agents, to maximum allowed time differences.
Before we turn to analyze the timely-coordinated response problem in the next
section, we first define a form of coordination exhibiting similar constraints,
and analyze it.

\begin{defn}[$\updelta$\mbox{-}Coordination]\label{defn:delta-coordination}
We call a quadruplet $\dcspec$ a ``$\updelta$\mbox{-}coor\-dination-spec''\footnote{
Note the difference between the italic $\delta$ that indicates a specific
constraining function, and the roman (i.e.\ upright)
$\updelta$ that generally refers to the form of coordination that we define,
regardless of any concrete constraining function. For example,
$\delta$\mbox{-}coordination, for a specific constraining function $\delta$,
is an instance of $\updelta$\mbox{-}coordination.
},
if $\gamma$ is a context, $R \subseteq \runs$ is a set of runs,
and $\implspec$ is an implementation-spec s.t.\ $I \subseteq \agents$.
Given a $\updelta$\mbox{-}coordination-spec $\dcspec$,
we say that an ensemble $\tuple{\ensemble} \in {\pointsets}^I$ is
``$\delta$\mbox{-}coordinated'' if
for every $(i,j) \in \distinctpairs{I}$ and for every $(r,t) \in \ensemble_i$,
there exists $t' \in T$ s.t.\ $t' \le t + \delta(i,j)$ and
$(r,t') \in \ensemble_j$.
\end{defn}

It should be noted that from this point on, whenever dealing with coordinated
ensembles, we always assume, for ease of presentation, that $|I|>1$, i.e.\
that the ensemble is defined over more than a single agent.
Adjusting our results for the case
in which this does not hold is neither hard, nor interesting.

Before we commence our analysis of $\updelta$\mbox{-}coordination, we define, given an
event $\psi \in \pointsets$, notation standing for the event
``$\psi$ holds at some (past, present, or
future) time, no later than $\varepsilon$ time units from now''.

\begin{defn}\label{defn:nolaterthan}
Let $\gamma$ be a context, let $R \subseteq \runs$ and let
$\varepsilon \in \Delta$.
We define
\functiondefn{\nolaterthan{\varepsilon}}{\pointsets}{\pointsets}{\psi}{
\{(r,t) \in \points \mid \exists\, t' \subseteq \timeset: t' \le t
+ \varepsilon \And (r,t') \in \psi\}.\footnotemark}
\end{defn}
\footnotetext{
As with our usage of $\sometime$, we use the symbol $\mncircledcirc$
instead of the standard temporal logic notation
$\mnmedcircle$, in order to emphasize that $t'$ may be smaller than $t$.
}
\vspace*{-1em} 

\begin{remark}\label{remark:nolaterthan-properties}
By \defnref{nolaterthan}:
\begin{itemize}
\item $\nolaterthan{\infty}=\sometime$.
\item $\nolaterthan{-\infty}(\psi)=\emptyset$, for every $\psi \in \pointsets$.
\item $\nolaterthan{0}(\psi)$, for an event $\psi \in \pointsets$, means ``$\psi$ has occurred, either now or in
the past''.
\item
Additivity: $\nolaterthan{\varepsilon_1}(\nolaterthan{\varepsilon_2}(\psi))=
\nolaterthan{\varepsilon_1+\varepsilon_2}(\psi)$ for every $\varepsilon_1,\varepsilon_2 \in \Delta \setminus \{-\infty\}$ and for every event $\psi \in \pointsets$.
\item
Monotonicity:
$\ (\varepsilon_1\le\varepsilon_2 \And \psi\subseteq\phi) \ \Rightarrow
\ \nolaterthan{\varepsilon_1}(\psi) \subseteq
\nolaterthan{\varepsilon_2}(\phi),\ $
for every $\varepsilon_1,\varepsilon_2 \in \Delta$ and for every two events $\psi,\phi \in \pointsets$.
\item
$\nolaterthan{\varepsilon}(\cap\Psi) \subseteq \bigcap\{\nolaterthan{\varepsilon}(\psi) \mid \psi \in \Psi\}$, for every $\varepsilon \in \Delta$ and for every set of events $\Psi \subseteq \pointsets$.
\end{itemize}
\end{remark}

We are now ready to
analyze $\updelta$\mbox{-}coordination along the lines of the results surveyed in the
previous section.
First, we define a common-knowledge analogue for
this case.
As in the results presented in the previous section,
given an event $\psi \in \pointsets$,
we use $\psi$ to define a function $f_{\psi}^{\delta}$, s.t.\
knowledge of the greatest fixed point of $f_{\psi}^{\delta}$
by each agent constitutes a $\delta$\mbox{-}coordinated ensemble
with several desired properties.
Nonetheless, we face several additional technical challenges along the way.
A main technical challenge is that $\updelta$\mbox{-}coordination lacks the symmetry
among the different agents in $I$,
which manifests in the common knowledge variants
presented in the previous section. Thus, in general there is no natural way to
define a single
event $\psi$ for which $(K_i(\psi))_{i \in I}$ is the ensemble we are looking
for, i.e.\ our ensemble should be defined in terms of knowledge of different
events for different agents.
For this reason, we somewhat generalize our strategy:
Instead of searching for a fixed point of a function
on $\pointsets$, we define a function on
${\pointsets}^I$ --- the set of $I$-tuples of events.
We denote the greatest fixed point of this function
by $\dck(\psi)$ (this is an $I$-tuple of events), and show that
$(K_i(\dck(\psi)_i))_{i \in I}$ is the desired ensemble, i.e.\ each
coordinate of this fixed point is the event that $i \in I$ should know in this
ensemble.\footnote{
While vectorial fixed points may alternatively be captured by nested fixed
points~\cite[Chapter 1]{mu-calculus}, in our case we argue that the
vectorial representation
better parallels the underlying intuition.
}
To our knowledge, such a technique was never utilized in this field before.

Before we define the above-described function, we define a lattice structure on
${\pointsets}^I$, which gives precise meaning to the concept of a greatest
fixed point of a function on ${\pointsets}^I$.

\begin{defn}[Lattice Structure on ${\pointsets}^I$]\label{defn:tuples-lattice}
Let $\gamma$ be a context, let $R \subseteq \runs$ and let $I \subseteq \agents$.
We define the following lattice structure on ${\pointsets}^I$: (In the
following definitions,
$\tuple{\varphi}\eqdef(\varphi_i)_{i \in I} \in {\pointsets}^I$ and
$\tuple{\xi}\eqdef(\xi_i)_{i \in I} \in {\pointsets}^I$.)
\begin{enumerate}
\item
Order: $\ \tuple{\varphi} \le \tuple{\xi}\ \ $ iff $\ \ \forall i \in I: \varphi_i \subseteq \xi_i$.
\item
Join: $\ \tuple{\varphi} \vee \tuple{\xi}\ \eqdef\ (\varphi_i \cup \xi_i)_{i \in I}$.
\item
Meet: $\ \tuple{\varphi} \wedge \tuple{\xi}\ \eqdef\ (\varphi_i \cap \xi_i)_{i \in I}$.
\end{enumerate}
\end{defn}

\begin{remark}\label{remark:complete-lattice}
${\pointsets}^I$, with the lattice structure defined above, constitutes a
complete lattice, i.e.\ every subset of ${\pointsets}^I$ has both a supremum and
an infimum.
\end{remark}

Now, for every $\psi \in \pointsets$, we turn to define the function on
${\pointsets}^I$, whose
greatest fixed point we denote by $\dck(\psi)$.

\begin{defn}[$\updelta$\mbox{-}Common Knowledge]\label{defn:f-psi}
Let $\dcspec$ be a $\updelta$\mbox{-}coordination-spec.
For each $\psi \in \pointsets$, we define
\functiondefn{f_{\psi}^{\delta}}{{\pointsets}^I}{{\pointsets}^I}{(x_i)_{i \in I}}
{\left(\smashoperator[r]{\bigcap_{j \in I\setminus\{i\}}}\, \nolaterthan{\delta(i,j)}(K_j(\psi \cap x_j))\right)_{i \in I},\footnotemark}
and denote its greatest fixed point by $\dck(\psi)$
(``$\delta$\mbox{-}common knowledge'' of $\psi$ by $I$).
\footnotetext{
This definition may be usefully generalized by allowing $\psi$ to depend
on $j$ and even on $i$ as well. Our results concerning this generalization
are outside the scope of this work.
}
\end{defn}

We now show that $\dck(\psi)$ is well defined.
Furthermore, we prove some basic properties of $\dck(\psi)$,
as well as of $\dck$, which constitutes a function from events
$\psi \in \pointsets$ to $I$-tuples of events $\tuple{\varphi} \in {\pointsets}^I$.

\begin{lemma}\label{lemma:delta-common-knowledge}
Let $\dcspec$ be a $\updelta$\mbox{-}coordination-spec and
let $\psi \in \pointsets$.
\begin{enumerate}
\item
$\dck(\psi)$ is well defined, i.e.\
$f_{\psi}^{\delta}$ has a greatest fixed point.

\item
Let $\tuple{\xi} \in {\pointsets}^I$.
If $\tuple{\xi} \le f_{\psi}^{\delta}(\tuple{\xi})$, then $\tuple{\xi} \le \dck(\psi)$.

\item
$\dck$ is monotone:
$\psi\subseteq\phi \Rightarrow
\dck(\psi) \le \dck(\phi)$ for every two events $\psi,\phi \in \pointsets$.
\end{enumerate}

\end{lemma}

\begin{proof}
By monotonicity of $K_i$ for every $i \in \agents$
and by monotonicity of $\nolaterthan{\varepsilon}$ for every
$\varepsilon \in \Delta$, we obtain that $f_{\psi}^{\delta}$ is monotone. By 
\remarkref{complete-lattice}, and by Tarski's fixed point theorem\cite{tarski},
the set of fixed points of $f_{\psi}^{\delta}$
has a greatest element, which equals
$\bigvee\{\tuple{\xi} \in {\pointsets}^I \mid \tuple{\xi} \le f_{\psi}^{\delta}(\tuple{\xi})\}$. This proves the first two parts of the lemma.

To prove monotonicity of $\dck$,
let $\psi,\phi \in \pointsets$ s.t.\
$\psi \subseteq \phi$.
Once again, by monotonicity of $K_i$ for every $i \in \agents$
and by monotonicity of $\nolaterthan{\varepsilon}$ for every
$\varepsilon \in \Delta$, we obtain that
$f_{\psi}^{\delta}(\tuple{\varphi}) \le f_{\phi}^{\delta}(\tuple{\varphi})$
for every $\tuple{\varphi} \in {\pointsets}^I$.
By substituting $\tuple{\varphi}\eqdef \dck(\phi)$,
and by definition of $\dck$,
we obtain:
$\dck(\psi) = f_{\psi}^{\delta}(\dck(\psi)) \le
f_{\phi}^{\delta}(\dck(\psi))$. By directly applying the second part
of the lemma, we obtain that $\dck(\psi)
\le \dck(\phi)$.
\end{proof}

It is now time to prove an equivalent of Theorems \ref{thm:perfect-coordination}, \ref{thm:eventual-coordination} and \ref{thm:epsilon-coordination}, for
$\updelta$\mbox{-}common knowledge.

\begin{thm}\label{thm:delta-coordination}
Let $\dcspec$ be a $\updelta$\mbox{-}coordination-spec.
\begin{enumerate}
\item
For every event $\psi$, the ensemble
$(K_i(\dck(\psi)_i))_{i \in I}$ is $\delta$\mbox{-}coordinated.
\item
If $\tuple{\ensemble} \in {\pointsets}^I$ is a $\delta$\mbox{-}coordinated ensemble,
then
$\ensemble_i \subseteq K_i(\dck(\cup\tuple{\ensemble})_i)$ for every $i \in I$.
\item
If $\tuple{\ensemble} \in {\pointsets}^I$ is a $\delta$\mbox{-}coordinated ensemble,
then $\cup\tuple{\ensemble} \subseteq \cup \dck(\cup\tuple{\ensemble})$.
\end{enumerate}
\end{thm}

\begin{proof}
We begin the proof of the first part by
noting that by \remarkref{knowledge-is-local},
$\tuple{\ensemble}\eqdef(K_i(\dck(\psi)_i))_{i \in I}$ is indeed
an ensemble for $I$.
Let $(i,j) \in \distinctpairs{I}$ and let
$(r,t) \in \ensemble_i$.
By definition of $\ensemble_i$ and by the truth axiom for knowledge,
$(r,t) \in \dck(\psi)_i$.
By definition of $\dck$,
\[
\dck(\psi)_i\:=\smashoperator[r]{\bigcap_{k \in I \setminus \{i\}}}\,\nolaterthan{\delta(i,k)}(K_k(\psi\cap\ensemble_k)) \subseteq \nolaterthan{\delta(i,j)}(K_j(\psi\cap\ensemble_j)).\]
Thus, we obtain
$(r,t) \subseteq \nolaterthan{\delta(i,j)}(K_j(\psi\cap\ensemble_j))$.
By definition of $\nolaterthan{\delta(i,j)}$, there exists $t' \in \timeset$
s.t.\ $t'\le t+\delta(i,j)$ and $(r,t') \in K_j(\psi\cap\ensemble_j)$.
By monotonicity of $K_j$ and by locality of $\ensemble_j$ to $j$,
we obtain $(r,t') \in K_j(\ensemble_j)=\ensemble_j$, and the
proof of the first part is complete.\footnote{
The attentive reader may notice a conceptual similarity between the above
argument and the proof of \corref{path-traversing-centipede-optimality}.
Furthermore, as noted there, this similarly extends to the second part of the
proof of \lemmaref{implementable-iff} as well.
}

We move on to proving the second part. Let
$\tuple{\ensemble}$ be as defined in this part of the theorem.
First, we show that $\tuple{\ensemble} \le f_{\cup\tuple{\ensemble}}^{\delta}(\tuple{\ensemble})$.
Let $i \in I$.
Let $(r,t) \in \ensemble_i$ and let $j \in I \setminus \{i\}$.
Since $\tuple{\ensemble}$ is $\delta$\mbox{-}coordinated,
there exists $t' \in \timeset$ s.t.\ $t' \le t+\delta(i,j)$
and $(r,t') \in \ensemble_j$. By definition of an ensemble, $\ensemble_j$ is
local to $j$, and thus $\ensemble_j = K_j(\ensemble_j)$. Therefore,
$(r,t') \in K_j(\ensemble_j)$. By definition of $\nolaterthan{\delta(i,j)}$,
we obtain $(r,t) \in \nolaterthan{\delta(i,j)}(K_j(\ensemble_j))$. Thus,
\[\ensemble_i \:\subseteq
\smashoperator[r]{\bigcap_{j \in I\setminus\{i\}}}\,\nolaterthan{\delta(i,j)}(K_j(\ensemble_j))\:=
\smashoperator[r]{\bigcap_{j \in I\setminus\{i\}}}\,\nolaterthan{\delta(i,j)}(K_j((\cup\tuple{\ensemble}) \cap \ensemble_j))=
f_{\cup\tuple{\ensemble}}^{\delta}(\tuple{\ensemble})_i.\]
By the second part of \lemmaref{delta-common-knowledge}, we thus have
$\tuple{\ensemble} \le \dck(\cup \tuple{\ensemble})$.
For every $i \in I$, by monotonicity of $K_i$ we obtain
$K_i(\ensemble_i) \subseteq K_i(\dck(\cup \tuple{\ensemble})_i)$, and by locality
of $\ensemble_i$ to $i$, we complete the proof of the second part of
\thmref{delta-coordination}, as $\ensemble_i = K_i(\ensemble_i)$.
Let $i \in I$. As we have just shown that
$\ensemble_i \subseteq \dck(\cup\tuple{\ensemble})_i$,
we also have $\ensemble_i \subseteq \cup \dck(\cup\tuple{\ensemble})$.
As this holds for every $i \in I$, and as the r.h.s.\ does not depend on
$i$, we obtain $\cup\tuple{\ensemble} \subseteq \cup \dck(\cup\tuple{\ensemble})$,
completing the third, and last, part of the proof.
\end{proof}

\thmref{delta-coordination}, which we have just proved, provides us with some
key properties of $\delta$\mbox{-}common knowledge: The first part of the theorem
says that the ensemble defined by it is $\delta$\mbox{-}coordinated. The second part
says that regardless of the way a $\delta$\mbox{-}coordinated ensemble is formed (be
it using $\delta$\mbox{-}common knowledge of some event $\psi$, or otherwise),
the fact that its $i$'th coordinate holds implies that $i$ knows the
$i$'th coordinate of $\delta$\mbox{-}common knowledge of (the disjunction of) this
ensemble. The third
part, similarly, says that in this case the fact that any coordinate of such
an ensemble holds implies that at least one coordinate of $\delta$\mbox{-}common
knowledge of (the disjunction of) this ensemble holds.

While \thmref{delta-coordination} does indeed provide us with several key
properties of $\updelta$\mbox{-}common knowledge, a second, deeper look at
this theorem (resp.\ at its analogues from \cite[Section~11.6]{book} surveyed
in the previous section) reveals that it does not characterise $\updelta$\mbox{-}common
knowledge (resp.\ common knowledge, eventual common knowledge, or
$\upvarepsilon$\mbox{-}common knowledge). Indeed, this theorem (resp.\ all its analogues)
would still hold if we defined each coordinate of $\updelta$\mbox{-}common knowledge
(resp.\ common knowledge, eventual common knowledge, or $\upvarepsilon$\mbox{-}common
knowledge) simply as $\pointsets$, i.e.\ the event ``True''.

\enlargethispage{.5em} 
We remark that it may be verified that
the ensemble defined by common knowledge of an event $\psi$ is the
greatest perfectly coordinated ensemble $\tuple{\ensemble}$ satisfying
$\cup \tuple{\ensemble} \subseteq \psi$.\footnote{
Greatest, here, is in the sense of the lattice structure defined in
\defnref{tuples-lattice}.
In particular, this characterisation implies that such a greatest ensemble
exists.
This characterisation of the ensemble defined by common knowledge
is an immediate consequence of Aumann's definition\cite{aumann} of common
knowledge, which may be rephrased as follows:
$C_I(\psi)$, for $I \subseteq \agents$ and $\psi \in \pointsets$,
is the greatest event $\phi \subseteq \psi$ that is local to each
$i \in I$.}
Similarly, we remark without a proof (as a proof would be similar to that of
\claimref{delta-coordination-characterisation} below) that it may be verified
that the ensemble defined by eventual common knowledge of an event of the form
$\sometime(\psi)$ is the greatest
eventually-coordinated ensemble $\tuple{\ensemble}$ satisfying
$\cup \tuple{\ensemble} \subseteq \sometime(\psi)$.
Analogous characterisations, for the ensembles defined by
$\upvarepsilon$\mbox{-}common knowledge and by
$\updelta$\mbox{-}common knowledge, are, however, more elusive to phrase (as is
an analogous characterisation of the ensemble defined by eventual-common
knowledge of an arbitrary event.)
For this reason, we now only characterise the ensemble defined by
$\updelta$\mbox{-}common knowledge of events that we call ``atemporal''.
(It may be readily verified, along the same lines, that an
analogous characterisation for the ensemble defined by $\upvarepsilon$\mbox{-}common
knowledge of an atemporal event holds as well.) While
this characterisation is similar to that of the ensemble defined by
eventual-common knowledge, which
we phrased above, it is not analogous in that it conceptually does a
significantly less adequate job in capturing the essence of $\updelta$\mbox{-}common
knowledge (and of $\upvarepsilon$\mbox{-}common knowledge).
Nonetheless, this characterisation suffices for our analysis of
the timely-coordinated response problem in the next section.

\begin{defn}[Atemporal Event]\label{defn:atemporal}
Let $\gamma$ be a context and let $R \subseteq \runs$.
We call an event $\psi \in \pointsets$ ``atemporal'' if
$\psi$ holding at some time during a run implies that it holds at all times
throughout that run. Formally, $\psi$ is atemporal iff it is of the form
$R' \times \timeset$ for some $R' \subseteq R$.
\end{defn}

\begin{remark}\label{remark:atemporal-properties}
By \defnref{atemporal}:
\begin{itemize}
\item
$\psi$ is atemporal iff $\psi=\sometime(\psi)$.
\item
By \remarkref{nolaterthan-properties} (additivity),
$\sometime=\nolaterthan{\infty}$ is idempotent.
Thus, $\sometime(\psi)$ is atemporal for every $\psi \in \pointsets$.
\end{itemize}
\end{remark}

We conclude this section with the following claim, which, together with
\remarkref{atemporal-properties}, provides a characterisation of
the ensemble defined by $\delta$\mbox{-}common knowledge of an atemporal event
$\psi$: It is
the greatest $\delta$\mbox{-}coordinated ensemble $\tuple{\ensemble}$ satisfying
$\cup \tuple{\ensemble} \subseteq \psi$. (In the
next section, we conclude that this characterisation holds also for
the ensemble defined by $\delta$\mbox{-}common knowledge of a temporal event
of the form $\nolaterthan{0}(\tilde{e})$, where $\tilde{e}$ is an ND event.)

\begin{claim}\label{claim:delta-coordination-characterisation}
Let $\dcspec$ be a $\updelta$\mbox{-}coordination-spec,
let $\psi \in \pointsets$ and
let $\tuple{\ensemble} \in {\pointsets}^I$ be the ensemble defined by
$\ensemble_i \eqdef K_i(\dck(\psi)_i)$ for every $i \in I$.
\begin{enumerate}
\item $\tuple{\ensemble}' \le \tuple{\ensemble}$, for every $\delta$\mbox{-}coordinated ensemble
$\tuple{\ensemble}' \in \pointsets^I$ satisfying
$\cup \tuple{\ensemble}' \subseteq \psi$.
\item $\cup \tuple{\ensemble} \subseteq \sometime(\psi)$.
(In particular, for atemporal $\psi$, by \remarkref{atemporal-properties}, $\cup \tuple{\ensemble} \subseteq \psi$.)
\end{enumerate}
\end{claim}

\begin{proof}
We begin by proving the first part. We have
\begin{align*}
&\:\ensemble'_i \subseteq & \text{by the second part of \thmref{delta-coordination}} \\
\subseteq& \:K_i(\dck(\cup \tuple{\ensemble}')_i) \subseteq & \text{by monotonicity of $K_i$ and of $\dck$} \\
\subseteq& \:K_i(\dck(\psi)_i) = & \text{by definition of $\ensemble_i$} \\
=& \:\ensemble_i,
\end{align*}
as required.

To prove the second part, let $i \in I$ and let $j \in I \setminus \{i\}$.
By monotonicity of
$\nolaterthan{\delta(i,j)}$, we have
\begin{align*}
& \: \dck(\psi)_i \subseteq
& \text{by definition of $\dck$} \\
\subseteq& \nolaterthan{\delta(i,j)}(K_j(\psi\cap \dck(\psi)_j)) \subseteq
& \text{by monotonicity of $K_j$} \\
\subseteq&
\nolaterthan{\delta(i,j)}(K_j(\psi)) \subseteq 
& \text{by the truth axiom for knowledge} \\
\subseteq& \nolaterthan{\delta(i,j)}(\psi) \subseteq &
\text{by \remarkref{nolaterthan-properties} (monotonicity)} \\
\subseteq& \sometime(\psi).
\end{align*}
By monotonicity of $K_i$ and by the truth axiom for knowledge, we obtain that
$\ensemble_i \subseteq K_i(\sometime(\psi)) \subseteq \sometime(\psi)$,
completing the proof of the second part.
\end{proof}

\section{Analyzing Timely-Coordinated Response}

We now relate the machinery developed in the previous section to
the timely-coordinated response problem. We begin by formally introducing
external inputs as events in $\pointsets$, and by
formalizing the relationship between ND events and knowledge.

\begin{defn}
Given a context $\gamma$ and an external input $\tilde{e} \in \externalinputs$,
we formally associate $\tilde{e}$ with the event
``$\tilde{e}$ is occurring (right now)'', i.e.\ with the set of all points
at which $\tilde{e}$ occurs.
\end{defn}

\begin{remark}\label{remark:nd-cannot-be-foreseen}
As noted in \chapterref{discrete}, since $\tilde{e}$ is an ND event,
it cannot be foreseen (i.e.\ known to occur) by any agent before it occurs.
In the notation
of this chapter, this may be formalized as follows:
$K_i(\sometime(\tilde{e})) \subseteq \nolaterthan{0}(\tilde{e})$
for every $i \in \agents$.
(This follows straight from applying
\claimref{no-nd-between-tprime-t} to the first part of
\remarkref{retainable-properties}, at $t_{\tilde{e}}$.\footnote{
For the continuous time model presented in \appref{continuous},
our analyses from both this and the previous chapter depend,
through \claimref{no-nd-between-tprime-t}, on the ``no foresight'' property.
The dependence of the analysis of this chapter on it, though, is not
fundamental, in the sense that it may be readily dropped by replacing
$\nolaterthan{0}(\tilde{e})$ with $\sometime(\tilde{e})$ in
\corref{delta-common-knowledge-optimality}. (In contrast, it is not clear
that the results of
the previous chapter may be easily modified to hold in the absence this
property.)
Moreover, unlike the
analysis of the previous chapter, which heavily relies on
the ``no extrasensory perception'' property of the continuous-time model,
the analysis of this chapter does not rely on it at all. (This is not
surprising given the fact that this analysis is agnostic to the methods
of information gain by agents. Indeed, not even once do we mention messages
in this analysis.)
These are both examples of the ability of the higher-level approach of this
chapter to mask the details of the model in question by phrasing its
results in terms of knowledge. (The meaning that the knowledge operator takes
on in a specific model, though, depends of course on such properties.)
})
\end{remark}

\begin{cor}\label{cor:ensembles-and-nd}
Let $\gamma$ be a context and let $\tilde{e} \in \externalinputs$.
\begin{enumerate}
\item $K_i(\sometime(\tilde{e})) = K_i(\nolaterthan{0}(\tilde{e}))$, for every
$i \in \agents$.
\item Let $I \subseteq \agents$ and let $\tuple{\ensemble} \in {\pointsets}^I$ be an ensemble. If
$\cup\tuple{\ensemble} \subseteq \sometime(\tilde{e})$, then
$\cup\tuple{\ensemble} \subseteq \nolaterthan{0}(\tilde{e})$.
\end{enumerate}
\end{cor}

\begin{proof}
Let $i \in \agents$. By \remarkref{nolaterthan-properties} (monotonicity)
and by monotonicity of
$K_i$, we have $K_i(\nolaterthan{0}(\tilde{e})) \subseteq K_i(\sometime(\tilde{e}))$.
Conversely, by the positive introspection axiom, by monotonicity of $K_i$ and by
\remarkref{nd-cannot-be-foreseen}, we have $K_i(\sometime(\tilde{e})) =
K_i(K_i(\sometime(\tilde{e}))) \subseteq K_i(\nolaterthan{0}(\tilde{e}))$,
and the proof of the first part is complete.

To prove the second part, let $i \in I$.
By locality of $\ensemble_i$ and by monotonicity of $K_i$,
we have $\ensemble_i=K_i(\ensemble_i)\subseteq K_i(\sometime(\tilde{e}))$.
By \remarkref{nd-cannot-be-foreseen}, the proof is complete.
\end{proof}

We conclude this chapter by applying the machinery developed throughout it to
obtain an optimal response logic for the timely-coordinated response problem.

\begin{cor}\label{cor:delta-common-knowledge-optimality}
Let $\TCRspec$ be a TCR-spec.
An optimal response logic for solving $\TCR{\tilde{e},I,\delta}$ is,
for every $i \in I$:
``respond when $K_i (\dck(\nolaterthan{0}(\tilde{e}))_i)$ holds for the
first time''.
\end{cor}

\begin{proof}
Assume that $\TCR{\tilde{e},I,\delta}$ is solvable and let
$P \in \TCR{\tilde{e},I,\delta}$. W.l.o.g., $P$ is a full-information
protocol.\footnote{
A much weaker assumption regarding $P$ suffices as well, actually.
}
For every $i \in I$, define $\rensemble_i \in \protocolpointsets$ as
the set of all points at which $i$ responds according to $P$.
Since $P$ is a full-information protocol, the actions (and in particular,
the responses) of each agent $i \in I$ at each time $t \in \timeset$ may be
deduced from its state at $t$. Therefore, $\tuple{\rensemble}$ is an ensemble.
In addition, define an ensemble $\tuple{\ensemble} \in {\protocolpointsets}^I$ by
$\ensemble_i \eqdef K_i (\dck(\nolaterthan{0}(\tilde{e}))_i)$. (This is indeed an ensemble,
by \remarkref{knowledge-is-local}, and thus $i$ may indeed respond according to
it.)

We first prove the optimality of responding according to $\tuple{\ensemble}$.
Let $i \in I$. We have to show that
for every $r \in \TRP$, $i$ would respond in $r$, according to the
response logic defined above, no later than $\timpl_r(i)$ (the response time of
$i$ in $r$ according to $P$).
Formally, this amounts to showing that
$\rensemble_i \subseteq \nolaterthan{0}(\ensemble_i)$.
Since $P \in \TCR{\tilde{e},I,\delta}$, we have
$\cup\tuple{\rensemble} \subseteq \sometime(\tilde{e})$ and by the second part
of \corref{ensembles-and-nd},
$\cup\tuple{\rensemble} \subseteq \nolaterthan{0}(\tilde{e})$.
Thus, by the first part of \claimref{delta-coordination-characterisation},
$\tuple{\rensemble} \le \tuple{\ensemble}$. 
(This may seem like a slightly stronger statement than the required
$\forall i \in I: \rensemble_i \subseteq \nolaterthan{0}(\ensemble_i)$,
however, we will show in \corref{delta-common-knowledge-ensemble-stable}
that these are in fact equivalent for full-information protocols.)

We now prove that responding according to $\tuple{\ensemble}$ solves
$\TCR{\tilde{e},I,\delta}$. By the first part of
\thmref{delta-coordination}, $\tuple{\ensemble}$
is $\delta$\mbox{-}coordinated.
Let $i \in I$.
Correctness of $P$ implies that $\rensemble_i$
holds at some time along any $r \in \TRP$. Using the notation of this chapter,
this is formulated as $\sometime(\tilde{e}) \subseteq \sometime(\rensemble_i)$.
Therefore, by the first part of this proof and by monotonicity of $\sometime$,
we obtain $\sometime(\tilde{e}) \subseteq \sometime(\ensemble_i)$ as well.
To complete the proof,
we note that by the second part of
\claimref{delta-coordination-characterisation} and by
\remarkref{nolaterthan-properties} (additivity),
$\cup\tuple{\ensemble} \subseteq \sometime(\nolaterthan{0}(\tilde{e})) = \sometime(\tilde{e})$,
i.e.\ $\tilde{e}$ occurs at some time along every run during which
any coordinate of $\tuple{\ensemble}$ holds.\footnote{
By \corref{ensembles-and-nd}, we have
$\cup\tuple{\ensemble} \subseteq \nolaterthan{0}(\tilde{e})$ as well, explaining
our previous statement that \corref{ensembles-and-nd} proves
\remarkref{response-not-before}.
}
\end{proof}

\chapter{The Equivalence of Both Approaches}\label{chapter:equivalence}
Corollaries
\ref{cor:path-traversing-centipede-optimality} and
\ref{cor:delta-common-knowledge-optimality}
both present optimal response
logics for solving the timely-coordinated response problem presented in
\chapterref{tcr-exhibition}. An obvious consequence
is that the response logics defined in both corollaries must somehow be
equivalent, at least for full-information protocols, which are assumed
in our proofs of these corollaries,\footnote{
The attentive reader may also notice some similarity in
the way the second part of the proof of each of these corollaries makes use of
its first part.} and in a shared-clock model, which is assumed in
our proof of \corref{path-traversing-centipede-optimality}.
In the next chapter, we apply our results from the previous chapters
to obtain specialized, somewhat more practical, versions of these results
for some naturally-occurring models.
However, before starting to do so, we
prove the equivalence of Corollaries \ref{cor:path-traversing-centipede-optimality} and \ref{cor:delta-common-knowledge-optimality} in a somewhat more
constructive manner in this chapter,
which also sheds some more light on the fixed-point analysis of the previous
chapter, and makes
the notion of $\updelta$\mbox{-}common knowledge more concrete. Our aim is to prove the
following result, by searching for a more constructive (yet equivalent)
definition of
$\updelta$\mbox{-}common-knowledge, along the lines of the nested-knowledge definition
of common knowledge given in \defnref{common-knowledge}, rather than those of
its fixed-point definition given in \thmref{common-knowledge-fixed-point}.

\begin{thm}\label{thm:path-traversing-centipedes-iff-delta-common-knowledge}
In a shared-clock model, let $\TCRspec$ be a TCR-spec
s.t.\linebreak $\TCR{\tilde{e},I,\delta}$ is solvable.
The response logics for $\TCR{\tilde{e},I,\delta}$ defined in Corollaries
\ref{cor:path-traversing-centipede-optimality} and
\ref{cor:delta-common-knowledge-optimality}
are equivalent when applied to full-information protocols, i.e.\ these two
response logics yield the exact same responses at the same times in each run.
\end{thm}

In order to prove
\thmref{path-traversing-centipedes-iff-delta-common-knowledge}, we
perform an analysis of $\updelta$\mbox{-}common knowledge of events of the form
$\nolaterthan{0}(\tilde{e})$ in full-information protocols.
To make our analysis somewhat cleaner and more generic,
we first aim to formally capture the properties of such
protocols and of events of the form $\nolaterthan{0}(\tilde{e})$, which are of
interest to us.

\section{Background}

In this section, we review two definitions and some basic properties thereof,
from \cite[Chapter~4]{book}. We rephrase these to match the notation we have
introduced so far.

\begin{defn}[Stability]\label{defn:stability}
Let $\gamma$ be a context and let $R \subseteq \runs$.
An event $\psi \in \pointsets$ is said to be ``stable'' if
once $\psi$ holds at some time during a run,
it continues to hold for the duration of that run. Formally, using our notation,
$\psi$ is stable iff $\psi=\nolaterthan{0}(\psi)$.
\end{defn}

\begin{remark}\label{remark:stability-properties}
By \defnref{stability}:
\begin{itemize}
\item
By \remarkref{nolaterthan-properties} (additivity),
$\nolaterthan{0}$ is idempotent. Thus, $\nolaterthan{0}(\psi)$ is a stable event for every $\psi \in \pointsets$.
\item
$\psi \cap \phi$ is a stable event for every two stable events $\psi,\phi \in \pointsets$.
\end{itemize}
\end{remark}

Indeed, the property of $\nolaterthan{0}(\tilde{e})$ that we utilize in this
chapter is its stability. We now present the second definition based upon
\cite[Chapter~4]{book}, which we utilize in this chapter.

\begin{defn}[Perfect Recall]
Let $\gamma$ be a context.
A set of runs $R \subseteq \runs$ is said to exhibit ``perfect recall'' if for
every $r \in R$, for every $i \in \agents$ and for every $t,t' \in \timeset$
s.t.\ $t'\le t$, the state of $i$ at $t$ in $r$ uniquely determines the state
of $i$ at $t'$ in $r$.
\end{defn}

\begin{remark}\label{remark:full-information-then-perfect-recall}
$\RP$ exhibits perfect recall for every full-information protocol\linebreak
$P \in \protocols$.
\end{remark}

We now distill the property of full-information protocols that is of interest
to us, namely that in sets of runs that exhibit perfect recall (and thus,
by \remarkref{full-information-then-perfect-recall}, also in
full-information protocols), knowledge of a stable event is itself stable.
The following is given in \cite[Exercise~4.18(b)]{book}, and its proof
follows directly from the definitions of stability and of knowledge.

\begin{claim}\label{claim:perfect-recall-property}
Let $\gamma$ be a context, let $R \subseteq \runs$ be a set of runs
exhibiting perfect recall and let $\psi \in \pointsets$.
If $\psi$ is stable, then $K_i(\psi)$ is stable as well,
for every $i \in \agents$.
\end{claim}

\section{A Constructive Proof}

Returning to our results and working toward proving
\thmref{path-traversing-centipedes-iff-delta-common-knowledge},
we first derive a stability property for
$\updelta$\mbox{-}common knowledge.

\begin{claim}\label{claim:delta-common-knowledge-stable}
Let $\dcspec$ be a $\updelta$\mbox{-}coordination-spec.
For every $\psi \in \pointsets$,
all coordinates of $\dck(\psi)$ are stable.
\end{claim}

\begin{proof}
Let $i \in I$. By \defnref{nolaterthan}, it is enough to show that
$\nolaterthan{0}(\dck(\psi)_i) \subseteq \dck(\psi)_i$.
Indeed, we have
\begin{align*}
&\ \:\:\nolaterthan{0}(\dck(\psi)_i)=
& \text{by definition of $\dck$} \\
=&\ \:\:\nolaterthan{0}\left(\smashoperator[r]{\bigcap_{j \in I\setminus\{i\}}}\,\nolaterthan{\delta(i,j)}(
K_j(\psi \cap \dck(\psi)_j))\right)
\subseteq 
& \text{by \remarkref{nolaterthan-properties}} \\
\subseteq&
\smashoperator[r]{\bigcap_{j \in I\setminus\{i\}}}\,\nolaterthan{0}(\nolaterthan{\delta(i,j)}(
K_j(\psi \cap \dck(\psi)_j)))=
& \text{by \remarkref{nolaterthan-properties} (additivity)} \\
=&
\smashoperator[r]{\bigcap_{j \in I\setminus\{i\}}}\,\nolaterthan{\delta(i,j)}(
K_j(\psi \cap \dck(\psi)_j))=
& \text{by definition of $\dck$} \\
=&\ \ \: \dck(\psi)_i.&\qedhere
\end{align*}
\end{proof}

It should be noted that stability of $\updelta$\mbox{-}common knowledge, as guaranteed
by \claimref{delta-common-knowledge-stable}, does not generally guarantee
stability of the ensemble defined by it
in the first part of \thmref{delta-coordination}. Nonetheless, combining
Claims \ref{claim:perfect-recall-property} and
\ref{claim:delta-common-knowledge-stable}, we obtain stability of this
ensemble in the presence of perfect recall.

\begin{cor}\label{cor:delta-common-knowledge-ensemble-stable}
Let $\dcspec$ be a $\updelta$\mbox{-}coordination-spec s.t.\ $R$ exhibits perfect
recall.
For every $\psi \in \pointsets$, all coordinates of the ensemble
$(K_i(\dck(\psi)_i))_{i \in I}$ are stable.
\end{cor}

Claims \ref{claim:perfect-recall-property} and \ref{claim:delta-common-knowledge-stable}
and the proof of \corref{delta-common-knowledge-optimality}
lead us to consider, for stable $\psi$ and given perfect recall, a slightly
different definition for $f_{\psi}^{\delta}$ than the one given in
\defnref{f-psi}. In order to phrase this definition,
we first define, given an event $\psi \in \pointsets$, notation standing for the
event ``$\psi$ holds at exactly $\varepsilon$ time units from now''.

\enlargethispage{1.5em} 
\begin{defn}\label{defn:atexactly}
Let $\gamma$ be a context, let $R \subseteq \runs$ and let
$\varepsilon \in \Delta \setminus \{-\infty,\infty\}$.
We define
\vspace*{-1em} 
\functiondefn{\atexactly{\varepsilon}}{\pointsets}{\pointsets}{\psi}
{\{(r,t) \in \points \mid (r,t+\varepsilon) \in \psi\}.\footnotemark}
\vspace*{-1.5em} 
\footnotetext{
Once again, we use the symbol $\mncircledcirc$ instead of the standard
temporal logic notation
$\mnmedcircle$, in order to emphasize that $\varepsilon$ may be nonpositive.
}
\end{defn}

\begin{remark}\label{remark:atexactly-properties} By \defnref{atexactly}, for every event $\psi \in \pointsets$ we have:
\begin{itemize}
\item
$\atexactly{\varepsilon_1}(\nolaterthan{\varepsilon_2}(\psi))=
\nolaterthan{\varepsilon_1}(\atexactly{\varepsilon_2}(\psi))=
\nolaterthan{\varepsilon_1+\varepsilon_2}(\psi)$, for every $\varepsilon_1,\varepsilon_2 \in \Delta \setminus \{-\infty,\infty\}$.
\item
$\atexactly{\varepsilon}(\psi) \subseteq \nolaterthan{\varepsilon}(\psi)$, for every $\varepsilon \in \Delta \setminus \{-\infty,\infty\}$.
\item $\atexactly{\varepsilon}$ commutes with intersection for every
$\varepsilon \in \Delta \setminus \{-\infty,\infty\}$:
$\atexactly{\varepsilon}(\cap\Psi) = \bigcap\{\atexactly{\varepsilon}(\psi) \mid \psi \in \Psi\}$ for every set of events $\Psi \subseteq \pointsets$.
\end{itemize}
\end{remark}

We now present our slightly modified definition of $f_{\psi}^{\delta}$,
which differs from the definition of $f_{\psi}^{\delta}$ given in
\defnref{f-psi} by the use of $\atexactly{\delta(i,j)}$
instead of $\nolaterthan{\delta(i,j)}$, and by intersecting with
$\sometime(\psi)$ instead of intersecting over eventual knowledge
requirements.\footnote{
The intersection with $\sometime(\psi)$ has any effect only if the intersection
following it is empty.
}

\begin{defn}
Let $\dcspec$ be a $\updelta$\mbox{-}coordination-spec s.t.\ $\delta > -\infty$.
For each $\psi \in \pointsets$, we define
\functiondefn{g_{\psi}^{\delta}}{{\pointsets}^I}{{\pointsets}^I}{(x_i)_{i \in I}}
{\left(\sometime(\psi) \cap\: \smashoperator{\bigcap_{\substack{j \in I\setminus\{i\}\\\delta(i,j)\ne\infty}}}\, \atexactly{\delta(i,j)}(K_j(\psi \cap x_j))\right)_{i \in I},}
and denote its greatest fixed point by $\dckg(\psi)$.
\end{defn}

Using an argument completely analogous
to the proof of \lemmaref{delta-common-knowledge},
it may be shown that
$\dckg(\psi)$ is well defined.
Furthermore, the same argument shows that $\dckg(\psi)$ also
satisfies the obvious analogues of the second and third parts of
\lemmaref{delta-common-knowledge}, with regard to $g_{\psi}^{\delta}$.

We now present a key observation, which stands at the heart of our proof of
\thmref{path-traversing-centipedes-iff-delta-common-knowledge}.
While, even in full-information protocols and when $\psi$ is stable,
$g_{\psi}^{\delta} \ne f_{\psi}^{\delta}$ (e.g.\ when applied to certain
unstable events), it so happens that under certain conditions,
the greatest fixed points of both of these functions coincide.

\begin{lemma}\label{lemma:g-eq-f}
Let $\dcspec$ be a $\updelta$\mbox{-}coordination-spec s.t.\ $R$ exhibits perfect
recall and s.t.\ $\delta>-\infty$,
and let $\psi \in \pointsets$.
If $\psi$ is stable, and if
$\sometime(\psi) \subseteq \sometime(\dck(\psi)_i)$ for
every $i \in I$,
then $\dckg(\psi)=\dck(\psi)$.
\end{lemma}

\begin{proof}[Proof.\footnotemark]\footnotetext{
It should be noted that we could have saved ourselves some hardship in this
proof by replacing $\cap_{j \in I\setminus\{j\}}\ldots$ with
$\sometime(\psi) \cap (\cap_{j \in I,\delta(i,j)<\infty}\ldots)$ when
defining $f_{\psi}^{\delta}$, which would still have allowed us to obtain
\corref{delta-common-knowledge-optimality}. While this is indeed true,
in this case many of our results regarding $\updelta$\mbox{-}common knowledge would
have required the additional assumption that
$\sometime(\psi) \subseteq \sometime(\dck(\psi)_i)$, reducing from their
generality and usefulness.
The added strength of the approach we have chosen presents itself
both in \corref{bounded-syncausal-path-or-zeno}, and while
discussing eventual common knowledge in \chapterref{previous}.
}
$\ge$:
Let $i \in I$.
By \claimref{delta-common-knowledge-stable}, $\dck(\psi)_j$ is stable
for every $j \in I$. Since $\psi$ is stable as well,
\remarkref{stability-properties} yields that
$\psi \cap \dck(\psi)_j$ is stable for every $j \in I$.
We also note that for every $j \in I \setminus \{i\}$, by the truth axiom
for knowledge and by
\remarkref{nolaterthan-properties} (monotonicity), we have
\begin{equation}\label{g-eq-f-internal1}
\nolaterthan{\delta(i,j)}(K_j(\psi \cap \dck(\psi)_j)) \subseteq
\sometime(\psi).
\end{equation}
Thus, we obtain
\begin{align*}
&\ \ \: \dck(\psi)_i = & \text{by definition of $\dck$} \\
=& \smashoperator[r]{\bigcap_{j \in I\setminus\{i\}}}\,\nolaterthan{\delta(i,j)}(
K_j(\psi \cap \dck(\psi)_j))=
& \text{by \eqref{g-eq-f-internal1}} \\
=& \ \:\sometime(\psi) \cap
\:\smashoperator{\bigcap_{j \in I\setminus\{i\}}}\,\nolaterthan{\delta(i,j)}(
K_j(\psi \cap \dck(\psi)_j))\subseteq
& \text{intersecting over fewer events} \\
\subseteq& \ \:\sometime(\psi) \cap
\:\smashoperator{\bigcap_{\substack{j \in I\setminus\{i\} \\ \delta(i,j)<\infty}}}\,\nolaterthan{\delta(i,j)}(
K_j(\psi \cap \dck(\psi)_j))=
& \text{by \remarkref{atexactly-properties}} \\
=&\ \:\sometime(\psi) \cap
\:\smashoperator{\bigcap_{\substack{j \in I\setminus\{i\} \\ \delta(i,j)<\infty}}}\,\atexactly{\delta(i,j)}(\nolaterthan{0}(
K_j(\psi \cap \dck(\psi)_j)))=
&
\smash{\begin{aligned}
\text{by \claimref{perfect-recall-property} and by} \\
\text{stability of $\psi \cap \dck(\psi)_j$}
\end{aligned}} \\
=&\ \:\sometime(\psi) \cap
\:\smashoperator{\bigcap_{\substack{j \in I\setminus\{i\} \\ \delta(i,j)<\infty}}}\,\atexactly{\delta(i,j)}(
K_j(\psi \cap \dck(\psi)_j))=
& \text{by definition of $g_{\psi}^{\delta}$} \\
=&\ \ \: g_{\psi}^{\delta}(\dck(\psi))_i.
\end{align*}
Thus, by the analogue of the second part of \lemmaref{delta-common-knowledge}
for $g_{\psi}^{\delta}$, we obtain $\dck(\psi)\le\dckg(\psi)$,
as required.

$\le$: For every $i \in I$, we have
\begin{align*}
&\:\dckg(\psi)_i=& \text{by definition of $\dckg$} \\
=&
\:\sometime(\psi) \cap \:\smashoperator{\bigcap_{\substack{j \in I\setminus\{i\}\\ \delta(i,j) \ne \infty}}}\,\atexactly{\delta(i,j)}(K_j(\psi \cap \dckg(\psi)_j))\subseteq & \text{by \remarkref{atexactly-properties}} \\
\subseteq&\:
\sometime(\psi) \cap \:\smashoperator{\bigcap_{\substack{j \in I\setminus\{i\}\\ \delta(i,j) \ne \infty}}}\,\nolaterthan{\delta(i,j)}(K_j(\psi \cap \dckg(\psi)_j)) \subseteq &
\text{as $\sometime(\psi) \subseteq \sometime(\dck(\psi)_i)$} \\[.25em]
\subseteq&\:
\sometime\left(\smashoperator[r]{\bigcap_{j \in I\setminus\{i\}}}\,\nolaterthan{\delta(i,j)}(K_j(\psi \cap \dck(\psi)_j))\right) \cap \\*
&\qquad\qquad \cap\left(\smashoperator[r]{\bigcap_{\substack{j \in I\setminus\{i\}\\ \delta(i,j) \ne \infty}}}\,\nolaterthan{\delta(i,j)}(K_j(\psi \cap \dckg(\psi)_j))\right) \subseteq  & \text{by monotonicity of $\sometime$} \\[.25em]
\subseteq&\:
\sometime\left(\smashoperator[r]{\bigcap_{\substack{j \in I\setminus\{i\} \\ \delta(i,j)=\infty}}}\,\nolaterthan{\delta(i,j)}(K_j(\psi \cap \dck(\psi)_j))\right) \cap \\*
&\qquad\qquad \cap\left(\smashoperator[r]{\bigcap_{\substack{j \in I\setminus\{i\}\\ \delta(i,j) \ne \infty}}}\,\nolaterthan{\delta(i,j)}(K_j(\psi \cap \dckg(\psi)_j))\right) \subseteq  & \text{by \remarkref{nolaterthan-properties}} \\[.25em]
\subseteq&\:
\left(\smashoperator[r]{\bigcap_{\substack{j \in I\setminus\{i\} \\ \delta(i,j)=\infty}}}\,\nolaterthan{\delta(i,j)}(K_j(\psi \cap \dck(\psi)_j))\right) \cap
\\*
&\qquad\qquad \cap\left(\smashoperator[r]{\bigcap_{\substack{j \in I\setminus\{i\}\\ \delta(i,j) \ne \infty}}}\,\nolaterthan{\delta(i,j)}(K_j(\psi \cap \dckg(\psi)_j))\right) \subseteq & \begin{aligned}\text{by monotonicity of $\nolaterthan{i,j}$} \\ \text{and of $K_j$, and by the} \\ \text{first part of this proof}\end{aligned} \\[.25em]
\subseteq&\:
\left(\smashoperator[r]{\bigcap_{\substack{j \in I\setminus\{i\} \\ \delta(i,j)=\infty}}}\,\nolaterthan{\delta(i,j)}(K_j(\psi \cap \dckg(\psi)_j))\right) \cap \\*
&\qquad\qquad \cap\left(\smashoperator[r]{\bigcap_{\substack{j \in I\setminus\{i\}\\ \delta(i,j) \ne \infty}}}\,\nolaterthan{\delta(i,j)}(K_j(\psi \cap \dckg(\psi)_j))\right) = \\[.5em]
=&
\smashoperator[r]{\bigcap_{j \in I\setminus\{i\}}}\,\nolaterthan{\delta(i,j)}(K_j(\psi \cap \dckg(\psi)_j))= &\text{by definition of $f_{\psi}^{\delta}$} \\[.25em]
=&\:f_{\psi}^{\delta}(\dckg(\psi))_i.
\end{align*}

\noindent Thus, by the second part of \lemmaref{delta-common-knowledge}, we have
$\dckg(\psi) \le \dck(\psi)$.
\end{proof}

One may wonder why we have worked so hard to obtain $\delta$\mbox{-}common knowledge,
under the conditions of this chapter, as a fixed point of $g_{\psi}^{\delta}$
rather than
of $f_{\psi}^{\delta}$. The answer is simple: $g_{\psi}^{\delta}$ commutes with
the meet operation, while $f_{\psi}^{\delta}$ does not. (Moreover, as a result,
$g_{\psi}^{\delta}$ is downward-continuous while $f_{\psi}^{\delta}$, even
in a discrete-time model, is not.)
This fact paves our way toward proving
\thmref{path-traversing-centipedes-iff-delta-common-knowledge}.

\begin{proof}[Proof of \thmref{path-traversing-centipedes-iff-delta-common-knowledge}]
Let $\TCRspec$ be a TCR-spec
s.t.\ $\TCR{\tilde{e},I,\delta}$ is solvable, and let $P \in \protocols$
be a full-information protocol. In this proof, we work in $\protocolpoints$.
Note that by \remarkref{stability-properties},
$\nolaterthan{0}(\tilde{e})$ is stable.

By \remarkref{implementation-properties},
solvability of $\TCR{\tilde{e},I,\delta}$ implies
$\delta>-\infty$. Furthermore, as shown in the proof of
\corref{delta-common-knowledge-optimality},
solvability of $\TCR{\tilde{e},I,\delta}$ implies
$\sometime(\tilde{e}) \subseteq
\sometime(K_i(\dck(\nolaterthan{0}(\tilde{e}))_i))$ for
every $i \in I$. By \remarkref{nolaterthan-properties} (additivity),
by the truth axiom for knowledge and by monotonicity of $\sometime$, we have
$\sometime(\nolaterthan{0}(\tilde{e}))=\sometime(\tilde{e}) \subseteq
\sometime(\dck(\nolaterthan{0}(\tilde{e}))_i)$ for every $i \in I$ as
well.
Thus,
by \remarkref{full-information-then-perfect-recall} and by \lemmaref{g-eq-f},
we obtain
$\dck(\nolaterthan{0}(\tilde{e}))=
\dckg(\nolaterthan{0}(\tilde{e}))$.

It is easy to verify that $g_{\nolaterthan{0}(\tilde{e})}^{\delta}$ commutes with both finite, and infinite,
meet. Thus, it is downward-continuous and by Kleene's fixed point
theorem\footnote{
This fixed point theorem seems to be popularly named after Kleene, as the idea
of using the orbit of an extremal element to obtain a fixed point was first
used in his proof of his first recursion theorem\cite[p.~348]{kleene}.
For a definition of this theorem that is phrased in terms of lattices,
continuity and greatest fixed point, see \cite{kolodner}.
}, we obtain
\[
\dckg(\nolaterthan{0}(\tilde{e})) = \bigwedge_{n \in \mathbb{N}} {g_{\nolaterthan{0}(\tilde{e})}^{\delta}}^n({\protocolpoints}^I).
\]
By $\atexactly{\varepsilon}$ commuting
with intersection for every $\varepsilon \in \Delta$, and by $K_i$ commuting
with intersection for every $i \in I$, we thus obtain, for every $i \in I$, that
\enlargethispage{2em} 
\begin{align*}
&\:\,\dck(\nolaterthan{0}(\tilde{e}))_i = \\[0.5em]
=&\bigcap_{n \in \mathbb{N}} {g_{\nolaterthan{0}(\tilde{e})}^{\delta}}^n({\protocolpoints}^I)_i = \\
=&\:
\sometime(\nolaterthan{0}(\tilde{e})) \cap \left(\smashoperator[r]{\bigcap_{\substack{j \in I\setminus\{i\} \\ \delta(i,j)<\infty}}}\, \atexactly{\delta(i,j)} (K_j(\nolaterthan{0}(\tilde{e})))\right)
\cap \\*
&\qquad\qquad\cap\left(\smashoperator[r]{\bigcap_{\substack{j \in I\setminus\{i\} \\ \delta(i,j)<\infty}}}\, \atexactly{\delta(i,j)} \left(K_j\left(\nolaterthan{0}(\tilde{e}) \cap \:\smashoperator{\bigcap_{\substack{k \in I\setminus\{j\} \\ \delta(j,k)<\infty}}}\, \atexactly{\delta(j,k)} (K_k (\nolaterthan{0}(\tilde{e})))\right)\right)\right)
\cap \\*
&\qquad\qquad\qquad\qquad\cap\cdots = \\
=&\:
\sometime(\tilde{e}) \cap \left(\smashoperator[r]{\bigcap_{\substack{j \in I\setminus\{i\} \\ \delta(i,j)<\infty}}}\, \atexactly{\delta(i,j)} (K_j(\nolaterthan{0}(\tilde{e})))\right)
\cap \\*
&\qquad\qquad\cap\left(\smashoperator[r]{\bigcap_{\substack{j \in I\setminus\{i\} \\ \delta(i,j)<\infty}}}\, \atexactly{\delta(i,j)} \left(K_j \left(\smashoperator[r]{\bigcap_{\substack{k \in I\setminus\{j\} \\ \delta(j,k)<\infty}}}\, \atexactly{\delta(j,k)}(K_k (\nolaterthan{0}(\tilde{e})))\right)\right)\right) \cap \\*
&\qquad\qquad\qquad\qquad\cap \cdots = \\[.25em]
=&\:
\sometime(\tilde{e}) \cap \: \smashoperator{\bigcap_{\substack{\apath \in \dpaths \\ p_1=i \\ \apath \ne (i)}}}\,
\atexactly{\delta(p_1,p_2)}(K_{p_2}(
\cdots
(\atexactly{\delta(p_{n-1},p_n)}(K_{p_n}(\nolaterthan{0}(\tilde{e}))))\cdots)).
\end{align*}

By \corref{ensembles-and-nd}, $K_i(\sometime(\tilde{e}))=K_i(\nolaterthan{0}(\tilde{e}))$
and thus, by $K_i$ commuting with intersection, we obtain (omitting henceforth some parentheses for readability)
\begin{equation}\label{eq:delta-common-knowledge-nested}
K_i(\dck(\nolaterthan{0}(\tilde{e}))_i) \:=
\smashoperator[r]{\bigcap_{\substack{\apath \in \dpaths \\ p_1=i}}}
\,K_{p_1}\atexactly{\delta(p_1,p_2)}K_{p_2}
\cdots
\atexactly{\delta(p_{n-1},p_n)}K_{p_n}(\nolaterthan{0}(\tilde{e})).
\end{equation}
Thus, the response logic from \corref{delta-common-knowledge-optimality} is
equivalent, for every $i \in I$,
to: ``respond as soon as 
\begin{equation}\label{eq:nested-knowledge-relative}
K_{p_1}\atexactly{\delta(p_1,p_2)}K_{p_2}\atexactly{\delta(p_2,p_3)}
\cdots
K_{p_{n-1}}\atexactly{\delta(p_{n-1},p_n)}K_{p_n}(\nolaterthan{0}(\tilde{e}))
\end{equation}
holds for every path $\apath \in \dpaths$ starting at $p_1=i$.''

By directly applying the methods of Ben-Zvi and
Moses\cite{bzm1,bzm2,bzm3,bzm4},\footnote{
Ben-Zvi and Moses show this in a discrete-time model. However, it may be
verified that
the ``no foresight'' and ``no extrasensory perception'' properties of
the continuous-time model presented in \appref{continuous} suffice in
order to adapt their argument, without fundamental change, to this model
as well.
} it can be seen that since $P$ is a full-information protocol in a shared-clock
model, a
$(\apath,\delta)$-traversing $\tilde{e}$-centipede by $t \in \timeset$
in a run of $P$
is equivalent to the following holding during that run, expressed by means
of their absolute-time modal-logic notation from \cite{bzm4}:
\begin{equation}\label{eq:nested-knowledge-absolute}
K_{(p_1,t_1)} K_{(p_2,t_2)} \cdots K_{(p_n,t_n)} \tilde{e},
\end{equation}
for $t_k\eqdef t+\dlength((p_m)_{m=1}^k)$ for every $k \in [n]$, and where
$\tilde{e}$ is a proposition corresponding to our
$\nolaterthan{0}(\tilde{e})$ event.

As \eqref{eq:nested-knowledge-relative} holding at $t$ is a different notation for
\eqref{eq:nested-knowledge-absolute}, the proof is complete.
\end{proof}

We conclude this chapter with an observation.
If $|I|<\infty$ and if $G_{\delta}$ has only trivial (i.e.\ singleton)
strongly connected
components, then there are only finitely many paths in $G_{\delta}$.
In this case, \thmref{path-traversing-centipede} and
\corref{path-traversing-centipede-optimality} imply that a timely-coordinated
response hinges on only finitely many path-traversing centipedes. (This is
indeed the case for the ordered response and weakly-timed response problems
studied by Ben-Zvi and Moses, as is shown in \chapterref{previous}.)
This observation may seem, at first glance, to clash
with the infinite nature of fixed points in general, and of greatest fixed
points in particular.
It is worthwhile to note that what reconciles these is that
in this case, ${g_{\psi}^{\delta}}^{|I|}$ is constant and therefore its value,
which is a finite intersection of nested-knowledge events, is its only fixed
point, and thus its greatest fixed point. Furthermore, by
\corref{delta-common-knowledge-optimality},
solvability of $\TCR{\tilde{e},I,\delta}$ implies that
$\sometime(\tilde{e}) \subseteq \sometime(\dck(\nolaterthan{0}(\tilde{e})))$
and thus, as noted above, we would still have obtained
\corref{delta-common-knowledge-optimality}
had we defined $f_{\psi}^{\delta}$ similarly to $g_{\psi}^{\delta}$, but
using $\nolaterthan{\delta(i,j)}$ instead of $\atexactly{\delta(i,j)}$.
In this case, the function ${f_{\psi}^{\delta}}^{|I|}$
would have also been constant, and the above insight would have held for
it as well.

\chapter{Results for Practical Models}\label{chapter:practical}

The analysis of the timely-coordinated response problem in Chapters
\ref{chapter:syncausality-approach} and \ref{chapter:fixed-point-approach} is a
general one, assuming very little regarding the model in which we work. 
The advantage of such a general analysis is that the results it yields hold
for a vast variety of models and situations. One disadvantage, which
we noted in the discussion concluding \chapterref{syncausality-approach},
is that the gap between these results and their consequences for practical
situations is quite large.
In this chapter, we derive, from the general analysis of the timely-coordinated
response problem
from the previous chapters, various stronger results for some special,
yet naturally-occurring, cases that we introduce below.
We also discuss some possible practical applications of our observations.

\section{Bounded-Syncausal-Path Contexts}

When presenting Theorems \ref{thm:broom}, \ref{thm:centibroom} and
\ref{thm:uneven-broom} above,
we noted that the proofs that Ben-Zvi and
Moses present for them strongly rely on time being modeled discretely.
It is worthwhile, in this context, to recall
the classic ``coordinated attack'', or ``two generals'',
problem\cite{two-generals-gangsters,two-generals}.
This problem describes a hypothetical situation, in which
two army generals, each camped on top a different hill overlooking some village,
wish to coordinate a simultaneous attack of this village
(i.e.\ reach common knowledge of an attack time
that was not agreed upon in advance), by communicating solely via messengers.
While this problem is unsolvable in a discrete-time
model\cite{two-generals-gangsters},
Fagin et al.\cite[p.~386]{book}\ note that
even in the lack of any delivery guarantee,
the generals may successfully coordinate a simultaneous attack if they have
access to a messenger who can make infinitely many trips
between one general's camp and the other's in finite time,
by doubling
her speed each time she reaches one of the camps.

In fact, \thmref{path-traversing-centipede} implies that in our continuous-time
model, even in the absence of any bound guarantee between two disjoint sets
of agents $I,J$
(i.e.\ when $\hat{\delta}_{\contextgraph}|_{\{I \times J\} \cup \{J \times I\}} \equiv \infty$),
a simultaneous response of two agents $i \in I$ and $j \in J$ may be achieved
even if no such ``infinite syncausal paths''
from $\tilde{e}$ to each agent at the time of its
response exist, as long as such paths that alternate between these sets of agents an
arbitrarily large number of times exist.\footnote{
This is thus also possible in the lack of any delivery
guarantee.}
In the above scenario, this means that the generals may also
coordinate a simultaneous attack if they have access to infinitely many
messengers, such that in a finite time frame, for any arbitrarily
large $N \in \mathbb{N}$, there exists a messenger who alternates between their camps
at least $N$ times. In \corref{bounded-syncausal-path-or-zeno},
we formalize the intuition that
there are no other ways in which these generals may coordinate even
an approximately simultaneous attack.

In order to have any hope of generalizing Theorems \ref{thm:broom},
\ref{thm:centibroom} and \ref{thm:uneven-broom} for the timely-coordinated
response problem
in a continuous-time model, we therefore first have to define some restriction
that prevents such Zeno-paradoxical situations from taking place.
By doing so, we effective force the infinitely many associated path-traversing
centipedes to degenerate to a broom-like, or centibroom-like,
structure. This intuition is formalized in \defnref{bounded-syncausal-path}.

\begin{defn}
Given a context $\gamma$, a run $r \in \runs$ and two agent-time pairs
$\theta_1,\theta_2 \in \agents \times \timeset$,
we denote by $L_r(\theta_1,\theta_2)$ the supremum of the number of ND events in
a syncausal path $\theta_1 \syncausal{r} \theta_2$.
If $\theta_1 \notsyncausal{r} \theta_2$, then we define
$L_r(\theta_1,\theta_2)\eqdef\infty$.
\end{defn}

\begin{defn}[Bounded-Syncausal-Path Context]\label{defn:bounded-syncausal-path}
We say that a context $\gamma$ is a ``bounded-syncausal-path'' context if
$\theta_1 \syncausal{r} \theta_2$ implies
$L_r(\theta_1,\theta_2) < \infty$,
for every run $r \in \runs$ and every two agent-time pairs
$\theta_1,\theta_2 \in \agents \times \timeset$.
\end{defn}

\begin{remark}
Some naturally-occurring bounded-syncausal-path contexts include contexts with
the following properties, which are customarily taken as axioms:
\begin{itemize}
\item
Any context in which a universal positive lower bound on all delivery times 
holds, including all contexts of the discrete-time model presented in
\chapterref{discrete}.
\item
Any context in which only finitely many messages may be sent (or rather,
may be delivered early) in any bounded time frame.
\end{itemize}
\end{remark}

While, due to the continuous nature of time and to the possibility of infinitely
many agents, finite-memory and finite-processing-power models need not
guarantee bounded-syncausal-path contexts (at least when dealing with protocols
that are not necessarily full-information ones), we show in the next
section that many results that hold for bounded-syncausal-path contexts,
still hold when the memory or processing power of each agent is limited.

We are now ready to generalize
Theorems \ref{thm:broom} and \ref{thm:uneven-broom} and
the proofs of Ben-Zvi and Moses\cite{bzm1,bzm2,bzm4} for these theorems,
for the timely-coordinated response problem
in a continuous-time, yet bounded-syncausal-path, context. The following is,
in a sense, a converse of \remarkref{broom-implies-path-traversing-centipede}
for such contexts.

\begin{cor}\label{cor:path-traversing-centipede-implies-broom}
Let $\TCRspec$ be a TCR-spec,
let $P \in \TCR{\tilde{e},I,\delta}$ and let $r \in \TRP$.
Let $J \subseteq I$ be a finite subset of $I$ that is contained entirely
within one strongly-connected component of $G_{\delta}$.

\begin{enumerate}
\item
If there exists $i \in J$ s.t.\
$L_r(\tilde{e},(i,\timpl_r(i))) < \infty$,\footnote{
In particular, this holds for every $i \in J$, if $\gamma$ is a
bounded-syncausal-path context.
} then there exists an event $e \in \PND{r}(i,\timpl_r(i))$
that is an $\tilde{e}$-broom for $J$ in $r$.
Furthermore, the horizon of this broom may be bounded by a finite bound of the
form
\[
\tilde{b}(\timpl_r(i),\max(\hat{\delta}|_{J^2}),|J|,L_r(\tilde{e},(i,\timpl_r(i)))).
\]
\item
If, in addition, $\hat{\delta}|_{J^2}$ is antisymmetric, then the broom
guaranteed by the first part of this corollary is by $(\timpl_r(j))_{j \in I}$,
implying the results of
Theorems \ref{thm:broom} and \ref{thm:uneven-broom}
for bounded-syncausal-path contexts in a continuous-time model.

\end{enumerate}
\end{cor}

\begin{proof}
Denote $n \eqdef |J|$.
The fact that $J$ resides within one strongly-connected
component of $G_{\delta}$ implies $\hat{\delta}|_{J^2} < \infty$,
and thus, a Hamiltonian cycle exists in the subgraph of $G_{\hat{\delta}}$
induced by $J$. Let $(p_1,\ldots,p_n,p_1)$ be such a cycle, for which $p_1=i$.
We now concatenate this cycle to itself enough times to obtain a path of
\mbox{$l \eqdef (n-1) \cdot L_r(\tilde{e},(i,\timpl_r(i))) + 1$} vertices, which we denote by
\[
\apath'=(\LaTeXunderbrace{p_1,\ldots,p_n,p_1,\ldots,p_n,p_1,\ldots,p_n,p_1,\ldots,p_k}_{l=(n-1) \cdot L_r(\tilde{e},(i,\timpl_r(i))) + 1}).
\]
By \thmref{path-traversing-centipede}, $r$
contains a $(\apath',\hat{\delta})$-traversing $\tilde{e}$-centipede
by $\timpl_r(i)$ --- denote it by $\tuple{e}=(e_m)_{m=1}^{l}$.
By definition, $\tuple{e}$ contains at most
$L_r(\tilde{e},(i,\timpl_r(i)))$ distinct events. Thus, each distinct event
contained in $\tuple{e}$ appears in it, on average, at least
$l / L_r(\tilde{e},(i,\timpl_r(i))) > n-1$ times. Thus, by the pigeonhole
principle, there exists an event
$e \in \PND{r}(i,\timpl_r(i))$ that appears in $\tuple{e}$ at least $n$ times.
By definition of a path-traversing centipede and by antisymmetry of
the syncausality relation,
these appearances are consecutive, and thus we obtain that there are $n$
consecutive vertices
in $\apath'$ to which there exists a delivery guarantee from $e$.
As any set of $n$ consecutive vertices \linebreak in $\apath'$ exactly equals $J$,
we obtain that $e$ is an $\tilde{e}$-broom for $J$ in $r$
by $\timpl_r(i) + \dlength(\apath')$.

We complete the proof of the first part of the corollary
by bounding this time:
\begin{align*}
&\:\timpl_r(i) + \dlength(\apath') \le & \text{by definition of $\dlength$} \\
\le&\: \timpl_r(i) + (l-1) \cdot \max(\hat{\delta}|_{J^2}) = & \text{by definition of $l$} \\
=&\: \timpl_r(i) + (|J|-1) \cdot L_r(\tilde{e},(i,\timpl_r(i))) \cdot \max(\hat{\delta}|_{J^2}) < & \quad\text{by finiteness of all elements} \\
<&\: \infty.
\end{align*}
We note that for large $L_r(\tilde{e},(i,\timpl_r(i)))$, obtaining the shortest
possible $\apath'$ frequently involves choosing a Hamiltonian cycle of minimal
length.

We now move on to proving the second part of the corollary.
If $\hat{\delta}|_{J^2}$ is antisymmetric, then the length of any
$\apathfull \in \paths{G_{\hat{\delta}|_{J^2}}}$ is precisely
$\hat{\delta}(p_n) - \hat{\delta}(p_1)$.
Therefore, for any $r \in \TRP$, any
end node of any $(\apath,\hat{\delta})$-traversing
$\tilde{e}$-centipede by $\timpl_r(p_1)$
is of the form $(p_k,\timpl_r(p_1) + L_{G_{\hat{\delta}}}((p_m)_{m=1}^k) =
(p_k, \timpl_r(p_1) + \hat{\delta}(p_1,p_k)) = (p_k, \timpl_r(p_k)$,
and the proof of the second part of the corollary is complete.
\end{proof}

An analogous proof gives rise to the following corollary, which generalizes
both \thmref{centibroom} and
\corref{path-traversing-centipede-implies-broom}.

\begin{cor}\label{cor:path-traversing-centipede-implies-centibroom}
Let $\TCRspec$ be a TCR-spec,
let $P \in \TCR{\tilde{e},I,\delta}$ and let $r \in \TRP$.
Let $n \in \mathbb{N}$ and
let $\tuplefull{J} \in (2^I)^n$ be a tuple of finite subsets of $I$, each of
which is contained entirely within one strongly-connected component of
$G_{\delta}$. Assume, furthermore, that no two of these subsets
are contained within the same strongly-connected component
of $G_{\delta}$, and that for every $m \in [n-1]$, there exists a path from
$J_m$ to $J_{m+1}$ in $G_{\delta}$.

\begin{enumerate}
\item
If there exists $i \in J_1$ s.t.\
$L_r(\tilde{e},(i,\timpl_r(i))) < \infty$,\footnote{
As before, this holds for every $i \in J_1$ if $\gamma$ is a
bounded-syncausal-path context.
}
then there exists an $\tilde{e}$-centibroom for
$\tuple{J}^{\mathit{rev}}$ in $r$,
consisting entirely of events from $\PND{r}(i,\timpl_r(i))$.
Furthermore, the horizon of this centibroom may be
bounded by a finite bound of the form
\[
\tilde{b}(\timpl_r(i),n,(\max(\hat{\delta}|_{{J_m}^2}))_{m=1}^n,(\min(\hat{\delta}|_{J_m \times J_{m+1}}))_{m=1}^{n-1},(|J_m|)_{m=1}^{n},L_r(\tilde{e},(i,\timpl_r(i)))).
\]
\item
For every $m \in [n]$, choose an arbitrary $j_m \in J_m$.
If, in addition to the conditions of the previous part,
for every $m \in [n]$, $\hat{\delta}|_{{J_m}^2}$ is antisymmetric,
then the end nodes of the centibroom guaranteed by the first part of this
corollary are
$\{(j,\timpl_r(j)) \mid j \in J_1\}$,
$\{(j,\timpl_r(j_1) + \hat{\delta}(j_1,j)) \mid j \in J_2\}$,
$\{(j,\timpl_r(j_1) + \hat{\delta}(j_1,j_2) + \hat{\delta}(j_2,j)) \mid j \in J_3\}$,
etc.\footnote{
By antisymmetry of $\hat{\delta}$ on each $J_m$, these are invariant to the
choice of representatives $(j_m)_{m=1}^n$.}
In particular, this also implies the result of \thmref{centibroom} for
bounded-syncausal-path contexts in a continuous-time model.\footnote{
The scenario studied in the second part of
\corref{path-traversing-centipede-implies-centibroom}
is, in a sense, a timed generalization of ordered joint response,
in that it is to tightly-timed response and to weakly-timed response,
as ordered joint response is to simultaneous response and to ordered response.
This scenario generalizes all the response problems studies by Ben-Zvi and
Moses that are surveyed in \chapterref{syncausality-approach}.
As noted in that chapter, the most general form of this
scenario does not fall within the scope of any of the coordinated response
problems defined and studied
by Ben-Zvi and Moses\cite{bzm1,bzm2,bzm3,bzm4}, even though no weak
mutual dependencies between response time exist in it.
}
\end{enumerate}
\end{cor}

\begin{proof}[Proof sketch]
For each $m \in [n]$, choose a Hamiltonian cycle in the subgraph of
$G_{\hat{\delta}}$ induced by $J_m$, and concatenate it to itself enough times
to obtain a path of
$l \eqdef (|J_m|-1) \cdot L_r(\tilde{e},i,\timpl_r(i)) + 1$ vertices.
Now, apply \thmref{path-traversing-centipede} to the concatenation of all
these paths, in ascending order of $m$.
The rest of the proof is analogous to the
proof of \corref{path-traversing-centipede-implies-broom}.
\end{proof}

\fig{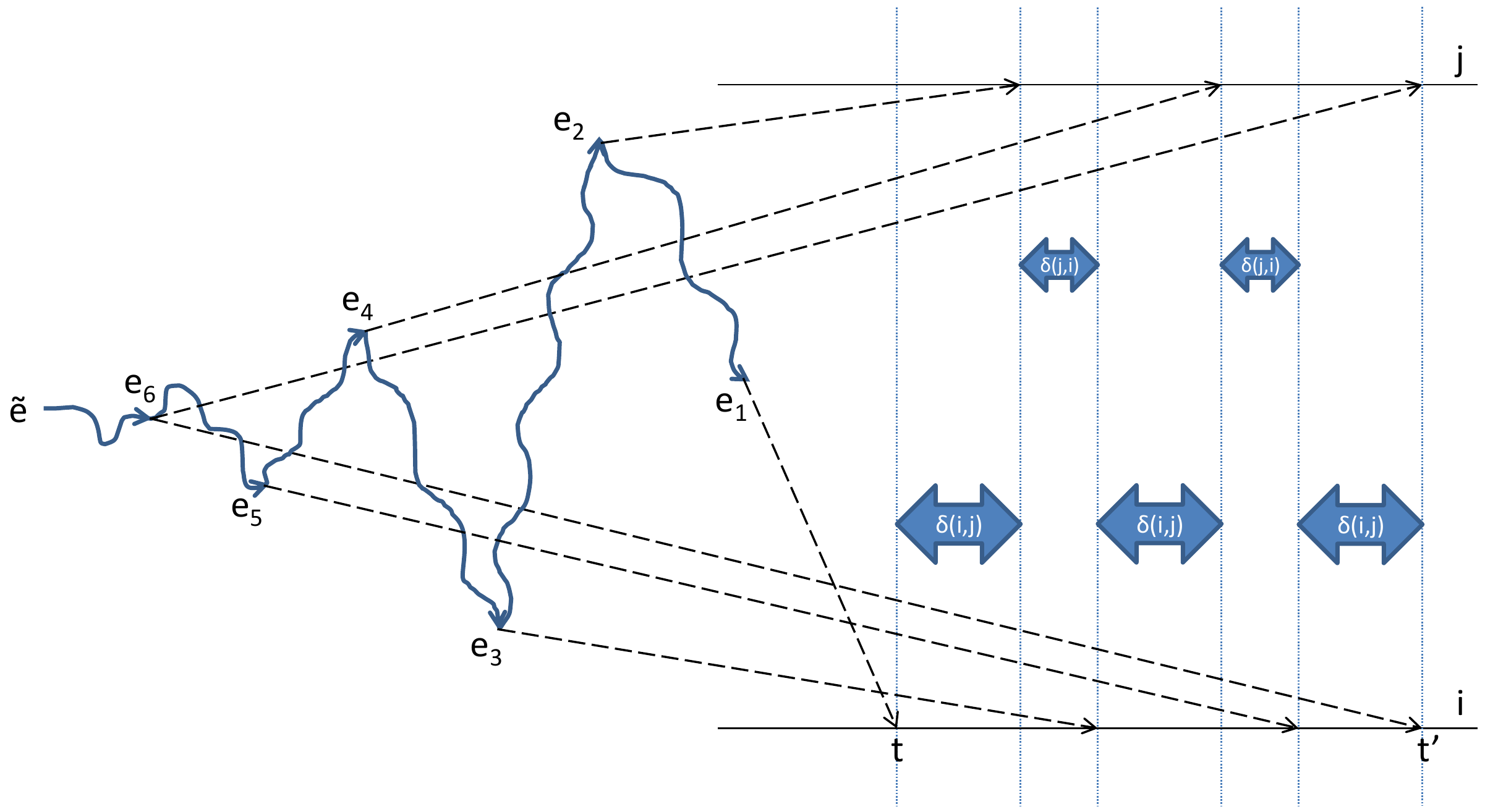}{Path-traversing centipede with a broom suffix}{
  When $I=\{i,j\}$, an $((i,j,i,j,i,j),\delta)$-traversing
$\tilde{e}$-centipede by $t$ ($(e_m)_{m=1}^6$) with an $\tilde{e}$-broom
suffix by $t'$ ($e_6$) provides
sufficient data for $i$ to respond at $t$ according to the optimal response
logic for $\TCR{\tilde{e},I,\delta}$, as it implies a
$(\apath,\delta)$-traversing $\tilde{e}$-centipede by $t$ for any
$\apath \in \dpaths$ starting at $p_1=i$.

It should be noted that in this example, common knowledge
of the occurrence of $\tilde{e}$ is only attained at $t'$. As the time at which
common knowledge of this occurrence is attainable in a full-information protocol
is independent of $\delta$, this intuitively illustrates a property of
the timely-coordinated response problem hinted
to by the first part of \remarkref{path-traversing-centipede-properties} and
by \exref{acme} opening this work in \chapterref{introduction}:
intuitively, the greater $\delta$ is, the better a chance there is to respond
earlier in many cases.
}

The proof of \corref{path-traversing-centipede-implies-broom} gives rise to the
following observation, stated using the notation of that proof, and under the
assumptions thereof:
Every path $\apath \in \paths{G_{\hat{\delta}|_{J^2}}}$ is a prefix of some
path $\apath'$, for which there exists a $(\apath',\delta)$-traversing
$\tilde{e}$-centipede that
has a suffix that constitutes an $\tilde{e}$-broom for $J$. Thus, this path-traversing
centipede yields a path-traversing centipede for every path of which $\apath'$
is a prefix.
While, as noted in the closing remarks of \chapterref{syncausality-approach},
implementing the optimal response logic from
\corref{path-traversing-centipede-optimality} may entail, in the most general
setting, checking for infinitely many path-traversing centipedes (using
infinitely many facts stored in memory),
this observation, and the more general analogous observation stemming from
the proof of \corref{path-traversing-centipede-implies-centibroom},
provide a practical and straightforward approach for implementing this optimal
response logic
in bounded-syncausal-path contexts, as illustrated in
\figref{path-traversing-centipede-within-a-broom},
and by the following example.

\begin{ex}\label{ex:algorithm}
Let $\TCRspec$ be a TCR-spec s.t.\ $\gamma$ is a bounded-syncausal-path
context, s.t.\ $|I|<\infty$ and s.t.\
$\TCR{\tilde{e},I,\delta}$ is solvable.
For simplicity of this example, assume for the time being that
$G_{\delta}$ is strongly connected and that $\delta>0$.
Let $i \in I$ and $t \in \timeset$. In a full-information protocol,
the following algorithm may be applied by
$i$ at $t$ to decide whether it should respond at that time according to the
optimal response logic for $\TCR{\tilde{e},I,\delta}$.

{\bf The Algorithm:}
First, check if any $e \in \PND{r}(i,t)$ is an $\tilde{e}$-broom for $I$
by any past, present, or future time. (This only depends on
the observer of $e$, so we may compute this efficiently with the aid of
a precalculated lookup table. Moreover, this may be computed once for each $e$,
storing the result in the state of $i$.)
Denote the set of all such broom events by $B$.\footnote{
While $\PND{r}(i,t)$ may be
infinite, implying that calculation of $B$ may require infinite processing
power and memory, we show in \claimref{finite-influence-implies}
in the next section that if $\TCR{\tilde{e},I,\delta}$ is solvable by any
finite-memory or finite-processing-power protocol, then
it is enough to consider only finitely many events from $\PND{r}(i,t)$
at this stage.
}
If $B = \emptyset$, then $i$ should not respond at $t$.

For each $e \in B$, denote by $\mathit{br}_e$ the earliest time by which
$e$ is an $\tilde{e}$-broom for $I$. Thus, by locality of bound guarantees
and as $\delta>0$, we obtain that
$(e)^n$ is a $(\apath,\delta)$-traversing $\tilde{e}$-centipede by
$\mathit{br}_e$, and by any later time,
for every $\apathfull \in \dpaths$.
(Note that $\mathit{br}_e-t_e$
is a constant depending solely on the observer of $e$, so once again, some advance
computation allows an efficient calculation of $\mathit{br}_e$, which, once calculated, may
be stored in the state of $i$.)

Denote $\widetilde{\mathit{br}}\eqdef\sup_{e \in B}\{\mathit{br}_e\}$.
This is a finite quantity, due to
$\mathit{br}_e-t_e$ depending only on the observer of $e$, and by finiteness of $I$.
Denote the set of paths $\apath \in \dpaths$ s.t.\
$\dlength(\apath)<\widetilde{\mathit{br}}-t$ by $\mathcal{P}$.
As $\delta>0$, and as $|I|<\infty$, we have $|\mathcal{P}|<\infty$ as well.

$i$ should respond at $t$ iff for each $\apathfull \in \mathcal{P}$,
there exists, by $t$, a $(\apath,\delta)$-traversing $\tilde{e}$-centipede
$\tuplefull{e} \in \PND{r}(i,t)^n$
s.t.\ there exists $e \in B$ satisfying $e \syncausal{r} e_n$. (This
check may be implemented efficiently using backtracking, and accelerated using
some precalculations.)

{\bf Dropping Unneeded Assumptions:}
As noted above, the assumptions that $\delta>0$ and that $G_{\delta}$ is strongly
connected are not required. Handling a situation in which they do not
hold is not inherently different, albeit significantly more cumbersome.
We now overview the key points of difference regarding these cases.

First, let us drop the assumption that $\delta>0$. This introduces two obstacles
for the above algorithms, which we now rectify.

The first obstacle is that, due to the possibility of $\delta$ taking negative
values, it may no longer hold that
$(e)^n$ is a $(\apath,\delta)$-traversing $\tilde{e}$-centipede by $\mathit{br}_e$ for every $e \in B$ and $\apathfull \in \dpaths$.
For this to hold again, we redefine $\mathit{br}_e$, for every $e \in B$,
as the earliest time satisfying
$e \boundguarantee (i,\mathit{br}_e+\min(\hat{\delta}|_{\{i\} \times I}))$
for every $i \in I$.
(See also \remarkref{broom-implies-path-traversing-centipede}.)
Note that $\mathit{br}_e$ is finite, by \lemmaref{implementable-iff}.
Also note that $\mathit{br}_e-t_e$
still depends only on the observer of $e$.
We accordingly redefine $\mathcal{P}$ as the set of paths
$\apath \in \dpaths$, s.t.\
the length of $\apath$, and of every prefix thereof, is less than
$\widetilde{br}-t$.

The second obstacle is that $G_{\delta}$ may contain
nontrivial cycles of zero length, which implies that $\mathcal{P}$ may be
of infinite cardinality. The adjustment of the algorithm for this case is somewhat
less straightforward. We partition $I$ into pairwise-disjoint equivalence
classes s.t.\ $i,j \in I$ are in the same equivalence class iff
$\hat{\delta}(i,j)=-\hat{\delta}(j,i)$.
It may be readily verified that these equivalence classes are exactly the
subsets
of $I$ on which $\hat{\delta}$ is antisymmetric, and that are maximal with
regard to this property. Another characterisation of these classes, which
is of key importance to us, is that a cycle in $G_{\delta}$ is of length 0 iff
all its vertices belong to the same equivalence class.
Let $J \subseteq I$ be a set of representatives for all such equivalence classes.
Denote, for each $j \in J$, its equivalence class (which it represents) by
$I_j$.
We restrict $\mathcal{P}$ to paths containing only ``representative''
vertices $j \in J$. (Thus, by the above characterisation of equivalence
classes using cycles lengths, $|\mathcal{P}|<\infty$ once again.)
Finally, for each $\apathfull \in \mathcal{P}$, we require that the matching
$(\apath,\delta)$-traversing $\tilde{e}$-centipede be, in a sense,
a ``path-traversing centibroom'' for $I_{p_1},\ldots,I_{p_n}$.
To be more precise, we require a path-traversing centipede
that, for each $k \in [n]$, does not merely satisfy
$e_k \boundguarantee (p_k,t+\dlength((p_m)_{m=1}^k))$
(as in the path-traversing centipede definition), but also
$e_k \boundguarantee (j,t+\dlength((p_m)_{m=1}^k)+\hat{\delta}(p_k,j))$,
for every $j \in I_{p_k}$.
(Checking for the existence of all required path-traversing
centipedes/centibrooms may still be efficiently implemented using the same
techniques as in the simpler case above.)

Finally, we sketch the key point of adapting the above algorithm for the case in
which $G_{\delta}$ is not necessarily strongly connected.
In this case, the syncausal structure underlying this algorithm is somewhat more
complex. Instead of the algorithm revolving around
$(\apath,\delta)$-traversing
$\tilde{e}$-centipedes/centibrooms $(e_m)_{m=1}^n$ for which $e_n$ is an
$\tilde{e}$-broom for all agents in $I$,
the algorithm searches for $(\apath,\delta)$-traversing $\tilde{e}$-centipedes
$(e_m)_{m=1}^n$
for which, for each strongly-connected component $I'$ of $G_{\delta}$ that is
visited by $\apath$, the event $e_{\max\{m \mid p_m \in I'\}}$ may be an
$\tilde{e}$-broom for $I'$.
\end{ex}

When presenting coordinated response problems in \chapterref{tcr-exhibition},
we noted that for many such problems, it is possible
to obtain a characterisation for solvability from an optimal response logic.
Indeed, \corref{path-traversing-centipede-optimality} implies that a necessary
and sufficient condition for solvability of the timely-coordinated response
problem in a shared-clock model
is a guarantee that in every triggered run, infinitely many path-traversing
centipedes (one for each path in $G_{\delta}$) occur by some finite time.
While this indeed fully characterises solvability of
the timely-coordinated response problem, the complexity of this
characterisation renders it not very usable.
\corref{path-traversing-centipede-implies-centibroom} also allows us to
present a surprisingly simpler characterisation for solvability of the
timely-coordinated response problem
of finitely many agents in a shared-clock model. We first present a special
case thereof,
which stems from \corref{path-traversing-centipede-implies-broom}.

\begin{cor}\label{cor:tcr-iff-sr}
In a shared-clock model, let $\gamma$ be a context, let $I \subseteq \agents$
be finite, and let $\tilde{e} \in \externalinputs$.
the following conditions are equivalent:
\begin{enumerate}

\item\label{some-solvable}
$\TCR{\tilde{e},I,\delta}$ is solvable for some
$\delta : \distinctpairs{I} \rightarrow \Delta \setminus \{-\infty\}$ s.t.\
$G_{\delta}$ is strongly connected and contains no negative cycles.

\item\label{all-solvable}
$\TCR{\tilde{e},I,\delta}$ is solvable for every
$\delta : \distinctpairs{I} \rightarrow \Delta \setminus \{-\infty\}$ s.t.\
$G_{\delta}$ is strongly connected and contains no negative cycles.

\item\label{sr-solvable}
$\SR{\tilde{e},I}$ is solvable.

\end{enumerate}

\end{cor}

\begin{proof}

\ref{sr-solvable} $\Rightarrow$ \ref{all-solvable}:
Let $P \in \SR{\tilde{e},I}$.
By \corref{implementable-iff-no-negative-cycles}, $\delta$ is implementable.
Therefore, by \claimref{solvable-iff-implementable}, $\TCR{\tilde{e},I,\delta}$ is
solvable.

\ref{all-solvable} $\Rightarrow$ \ref{some-solvable}: Immediate.

\ref{some-solvable} $\Rightarrow$ \ref{sr-solvable}: 
Let $P \in \TCR{\tilde{e},I,\delta}$ for some such $\delta$. W.l.o.g., assume
that $P$ is a full-information protocol.
By combining the ``stand-alone external inputs'' and ``no foresight'' properties
of the continuous-time model presented in \appref{continuous}, we
may construct a run $r \in \TRP$ in which
no message is delivered early less than 1 time unit after it is sent.
(For a discrete-time model, this holds for any run $r \in \TRP$.)
By \corref{path-traversing-centipede-implies-broom},
$r$ contains an $\tilde{e}$-broom for $I$.
Thus, $P$
may be modified to solve $\SR{\tilde{e},I}$ by modifying its response
logic for all agents to: ``respond at the earliest horizon of an
$\tilde{e}$-broom for all $I$''.
(The fact that $P$ is a full-information protocol,
together with the fact that the clock is shared,
guarantees that if a broom for all $I$ exists by any $t \in \timeset$,
then each agent in $I$ can deduce this at $t$.)
\end{proof}

\corref{tcr-iff-sr}, together with \thmref{broom}, imply that
$\TCR{\tilde{e},I,\delta}$
is solvable under
condition \ref{some-solvable} of this corollary iff
there exists an agent $i \in \agents$ to which
there exists a path $\apath \in \paths{\contextgraph}$
from $i_{\tilde{e}}$ and s.t.\
$\max\{\hat{\delta}_{\contextgraph}(i, j)\}_{j \in I} < \infty$.
Furthermore, if $\delta \ge 0$, then
given a full-information protocol endowed with the optimal
response logic, the latest of the responses of $I$ occurs,
in each triggered run of this
protocol, no later than
\begin{equation}\label{response-time-bound}
\min\{\glength(\apath)+\max\{\hat{\delta}_{\contextgraph}(p_n, j)\}_{j \in I} \mid \apathfull \in \paths{\contextgraph} \And p_1=i_{\tilde{e}} \}\footnotemark
\end{equation}\footnotetext{
A generalization of this expression, for cases in which $\delta \ngeq 0$,
may be readily obtained by applying
\remarkref{broom-implies-path-traversing-centipede}.
}time units after the occurrence of $\tilde{e}$, and this bound is tight.
(If this value is infinite,
then there exist triggered runs in which the time of the latest of the
responses is arbitrarily large.)

Using \corref{path-traversing-centipede-implies-centibroom},
we may similarly deduce a generalization of \corref{tcr-iff-sr}, for the case
in which $G_{\delta}$ is not necessarily strongly connected, yielding
a characterisation for solvability of any instance of the timely-coordinated
response problem of finitely many agents in a shared-clock model.

\begin{cor}\label{cor:solvability-iff}
In a shared-clock model, let $\TCRspec$ be a TCR-spec s.t.\ $|I|<\infty$ and
s.t.\ $G_{\delta}$ contains no negative cycles.
$\TCR{\tilde{e},I,\delta}$ is solvable iff
for every tuple $(I_1,\ldots,I_n)$ of strongly-connected components of $G_{\delta}$
s.t.\ there exists a path in $G_{\delta}$ from $I_m$ to $I_{m+1}$
for every $m \in [n-1]$, there exists $(i_m)_{m=1}^n \in \agents$ satisfying:
\begin{itemize}
\item
There exists a path in $\contextgraph$ from $i_{\tilde{e}}$, through $i_n$,
through $i_{n-1}$, \ldots, to $i_1$.
\item
$\max\{\hat{\delta}_{\contextgraph}(i_m, j)\}_{j \in I_m} < \infty$, for every $m \in [n]$.
\end{itemize}
(A natural, yet more cumbersome, analogue of \eqref{response-time-bound} may be
phrased under the conditions of this corollary as well.)
\end{cor}

\corref{solvability-iff}
implies, in particular, that solvability of $\TCR{\tilde{e},I,\delta}$
depends only \linebreak on the strongly-connected components of
$G_{\delta}$, on the partial order it induces on them,\linebreak and on $\contextgraph$.
In most practical situations, $\contextgraph$ is strongly connected, as absence
of this property means that there exist two agents $i,j \in \agents$
s.t.\ $i$ may never hope to send any data to $j$, either directly or indirectly.
If indeed $\contextgraph$ is strongly connected, then
\corref{solvability-iff} reduces to the following, surprisingly simple,
condition for solvability of the timely-coordinated response problem, with
which we conclude this section.

\begin{cor}\label{cor:solvability-depends-on-strongly-connected-components}
In a shared-clock model, let $\TCRspec$ be a TCR-spec
s.t.\ $\contextgraph$ is strongly connected,
s.t.\ $|I|<\infty$
and s.t.\ $G_{\delta}$ contains no negative cycles.
The following conditions are equivalent:
\begin{enumerate}
\item
$\TCR{\tilde{e},I,\delta}$ is solvable.
\item
$\SR{\tilde{e},J}$ is solvable, for each
strongly-connected component $J$ of $G_{\delta}$.
\item
For every strongly-connected component $J$
of $G_{\delta}$, there exists $i \in \agents$ s.t.\ 
$\max\{\hat{\delta}_{\contextgraph}(i, j)\}_{j \in J} < \infty$.
\end{enumerate}
\end{cor}

\section{Finite-Influence Protocols}

In this section, we explain, as promised, why the results of the previous
section hold also for models in which the memory or processing power
of each agent is limited. As a nice bonus, many of those results, which in
the previous section held only for finite sets of agent, will turn out to
hold in such models for infinite sets of agents as well. To give some intuition
for the definition that we use to make this explanation precise, we first
prove a result regarding simultaneous response among infinitely many agents.

Recall that \thmref{broom} implies that in a discrete-time model,
simultaneous response of finitely many agents based on an ND event
requires the existence of a broom.
In \corref{path-traversing-centipede-implies-broom}, we have relaxed the
requirement for discrete-time in this result to a requirement
for a bounded-syncausal-path context.
A natural question to ask
is whether the requirement for finiteness of the set of responding agents
may somehow be relaxed as well. A quick check shows that none of the proof
strategies we have seen so far for (any variant of) \thmref{broom} scale to
the case of coordinating a simultaneous response among infinitely many agents.
Indeed, it turns out that while a broom for any finitely-sized subset of agents
is guaranteed under such conditions (it is not hard to show that such a
collection of finite brooms is sufficient to optimally coordinate a simultaneous
response of all agents),\footnote{
This observation is closely tied to the fact that by
\defnref{common-knowledge}, in any context
we have:
\[ C_I(\psi) \:= \smashoperator[r]{\bigcap_{\substack{J \subseteq I \\ |J|<\infty}}}\, C_J(\psi), \]
for any set of agents $I$, regardless of its cardinality. Readers who find
it intuitively difficult to accept this ``finiteness'' of common knowledge
(which yields the above mentioned possibility of a simultaneous response of $I$
in a bounded-syncausal-path context in the absence of a broom for $I$)
may wish to read \cite{barwise}, in which Barwise suggests an
interpretation of knowledge in which the fixed-point definition of common
knowledge (\thmref{common-knowledge-fixed-point})
is strictly stronger than its definition as a conjunction of
finitely-nested knowledge events (\defnref{common-knowledge}).
} a broom for all agents is not guaranteed even in
bounded-syncausal-path contexts. (As may be expected, such a broom is sufficient
for coordinating a simultaneous response, although it is possible to construct
an example in which the response logic based on
the existence of such a broom is non-optimal.)
The following theorem characterises the conditions required for simultaneous
response of countable many
agents in a bounded-syncausal-path context and shows that there are,
in a sense, some unintuitive consequences to the absence of a broom for all
agents in this situation. We will shortly use the insights this theorem
provides us in order to phrase a restriction on protocols, which disallows such
``consequences''.

\begin{thm}[Infinite Broom or Infinitely Many Brooms]\label{thm:infinite-broom-or-infinite-brooms}
Let $\gamma$ be a bounded-syncausal-path context,
let $I \subseteq \agents$ s.t.\ $|I|=\aleph_0$,
let $\tilde{e} \in \externalinputs$ and let
$P \in \SR{\tilde{e},I}$. For every $r \in \TRP$, one of the following holds:
\begin{enumerate}
\item
There exists an $\tilde{e}$-broom for $I$ by $t_r$ in $r$; or
\item
For every finite $J \subseteq I$,
there exist infinitely many distinct $\tilde{e}$-brooms for $J$
by $t_r$ in $r$.
\end{enumerate}
\end{thm}

\begin{proof}
Set $t \eqdef t_r$. Let $J \subseteq I$ be finite.
We inductively construct a sequence of sets of ND events $(E_k)_{k=1}^{\infty}
\in \ND(r)^{\mathbb{N}}$, satisfying:
\begin{itemize}
\item
$E_1 = J$.
\item
$\forall k \in \mathbb{N}: |E_k| < \infty$.
\item
$\forall k \in \mathbb{N}: E_k \subseteq E_{k+1}$.
\item
$\cup_{k=1}^{\infty} E_k = I$.
\end{itemize}
To construct $(E_k)_{k=1}^{\infty}$, choose any well-ordering of $I$ under which
it is isomorphic to $\omega$, and given $E_k$ for $k \in \mathbb{N}$, set
$E_{k+1} \eqdef E_k \cup \{\min\{I \setminus E_k\}\}$.

For every $k \in \mathbb{N}$, set $B_k \eqdef \{e \in \ND(r) \mid \text{$e$ is an $\tilde{e}$-broom for $E_k$ by $t$ in $r$}\}$.
By \remarkref{tcr-properties} and by
\corref{path-traversing-centipede-implies-broom},
$\forall k \in \mathbb{N}: B_k \ne \emptyset$.
Furthermore, by the definition of a broom and since $(E_k)_{k=1}^{\infty}$ is
increasing, $\forall k \in \mathbb{N}: B_k \supseteq B_{k+1}$.

If $|B_1| < \infty$, then $\forall
k \in \mathbb{N}: |B_k| < \infty$ and all of $\{B_k\}_{k=1}^\infty$ are therefore
closed under the
co-finite topology on $\mathbb{N}$ and have the finite intersection property.
From compactness of $\mathbb{N}$ under this topology, we obtain
$B \eqdef \cap_{k=1}^{\infty} B_k \ne \emptyset$. Let $e \in B$. By definition
of $(B_k)_{k=1}^{\infty}$, $e$ is an $\tilde{e}$-broom for
$\cup_{k=1}^{\infty} E_k =I$ by $t$, and the first condition of
\thmref{infinite-broom-or-infinite-brooms} is satisfied.

Otherwise, $B_1$ is of infinite cardinality, and thus $J$ satisfies the second
condition of \thmref{infinite-broom-or-infinite-brooms}.
\end{proof}

\thmref{infinite-broom-or-infinite-brooms} shows that in the absence of a
single broom for all agents in $I$, there are, for each $i \in I$,
infinitely many
brooms that are ``important'' to $i$, in the sense that the existence of any
finite subset thereof is not sufficient to trigger the response of $i$ at $t$.
It seems unrealistic for $i$ to check (directly, or
indirectly by receiving this information from some other agent) that
infinitely many brooms exist.\footnote{
It should be noted that while one may argue that sending infinitely many
messages is equally unrealistic (an argument that suggests that
even $\ER{\tilde{e},I}$
is unsolvable for infinite $I$), we would like to argue that broadcasting a
single message from one agent to infinitely many agents may not be inconceivable.
For example, one may post such a message in
a publicly-visible place such that each agent is guaranteed to notice this
message within a given time period after its posting.
}
We now formalize this intuition.

\begin{defn}[Finite-Influence Protocol]
Let $\gamma$ be a context.
We say that a protocol $P \in \protocols$ is a ``finite-influence'' protocol if
for any run
$r \in \RP$ and for any agent-time pair $(i,t) \in \agents \times \timeset$
s.t.\ $i$ responds in $r$ at $t$, there exists a
finite $t$-retainable§ set $E \in \RND(r,t)$ s.t.\ $i$ still responds at
$t$ in $r \tcap E$.
\end{defn}

Obviously, all protocols in any context in which only finitely many messages
may be sent (or rather, may be delivered early) in any bounded time frame,
are finite-influence protocols. When only finitely many agents exist,
then all protocols in any discrete-time context in which a universal positive
lower bound on all delivery times holds, are also finite-influence protocols.
As noted in the previous section, both of these properties are traditionally
taken as axioms, and this is equally true for finiteness of the set of agents.
While, as noted above, not all finite-memory and finite-processing power models
guarantee bounded-syncausal-path contexts, we now explain why all protocols in
such models are finite-influence ones.

If a protocol $P$ is not a finite-influence protocol,
then there exists a run $r \in \RP$ and an agent-time pair $(i,t)$ s.t.\
by time $t$ either $i$, or some other agent who sends information to $i$,
has to either take infinitely many ND events into
account when performing some state change (i.e.\
no finite subset of these events taking place would have yielded the
same state change), or perform infinitely many state changes.
This implies that at least one agent utilizes infinite processing power in
finite time. If, furthermore, we assume that no two ND events reach an agent
at exactly the same time (it is enough to assume that the logic of $i$ can not
atomically access information regarding more than one ND event observed by it),
then that agent not only utilizes infinite processing power,
but also infinite memory, as the infinite amount of state changes described
above must involve infinitely many unique states. (Alternatively, the
computation of such a single state change that takes into account infinitely
many ND events, requires infinite memory.)

We may conclude that for finite-influence protocols in
bounded-syncausal-path contexts, \thmref{broom} holds even when time is
continuous and when $I$ may be countably infinite.
It took us quite a chain of reductions and conclusions to show this.
One may argue that this might suggest that the concept of
finite influence is somewhat artificial, and ``not from the book''.
To try and refute this
argument, we now present a novel, direct and concise proof of \thmref{broom}
which holds for any finite-influence protocol (even in a continuous-time
model), regardless of the context
and of the cardinality of $I$.\footnote{
This proof extends immediately to proving the second part of
\corref{path-traversing-centipede-implies-broom} under the conditions
of \claimref{finite-influence-implies} below (implying an infinite-agent
version of \thmref{uneven-broom} for finite-influence protocols),
and even extends (using an inductive argument
very similar to the one used in \thmref{path-traversing-centipede}) to prove
the second part of \corref{path-traversing-centipede-implies-centibroom}
under the conditions of \claimref{finite-influence-implies} (implying
an infinite-agent version of \thmref{centibroom} for finite-influence
protocols).
For the sake of conciseness, though, we phrase and
prove it here only for simultaneous response.
}

\begin{thm}[Broom]\label{thm:concise-broom}
Let $\gamma$ be a context, let
$I \subseteq \agents$, let $\tilde{e} \in \externalinputs$ and
let $P \in \SR{\tilde{e},I}$ be a finite-influence protocol.
Each $r \in \TRP$ contains an $\tilde{e}$-broom for $I$ by $t \eqdef t_r$.
\end{thm}

\begin{claim}\label{claim:all-seen-by-all}
Under the conditions of \thmref{concise-broom},
there exists $r' \tsubseteq r$ such that
$t_{r'} = t$, such that $|\ND(r',t)|<\infty$ and such that
$\PND{r'}(i,t) = \ND(r',t)$ for every $i \in I$.
\end{claim}

\begin{proof}
Let $E \in \RND(r,t)$ be minimal such that $t_{\smash{r \tcap E}}=t$.
(There always exists such a finite set, due to $P$ being a
finite-influence protocol, since by correctness of $P$, if even one agent
$i \in I$ responds at $t$, then all agents in $I$ do.)
By definition, $E = \ND(r \tcap E,t)$.

Let $i \in I$.
Set $E' \eqdef \PND{r \tcap E}(i,t)$. By \corref{adjusted-lemma-3},
$E' \in \RND(r \tcap E,t)$,
and the state of $i$ at $t$ is the same in $r \tcap E$ and in
$(r \tcap E) \tcap E' = r \tcap E'$.
Therefore, $i$ still responds at $t$ in
$r \tcap E'$. Thus,
by correctness of $P$, we obtain $t_{\smash{r \tcap E'}} = t$.
By minimality of $E$, therefore, since $E' \subseteq E$, we obtain
$E' = E$, and thus any $r' \in r \tcap E$ fulfills the above requirements.
\end{proof}

\begin{proof}[Proof of \thmref{concise-broom}]
Let $r'$ be as in \claimref{all-seen-by-all} and choose
$e \in \ND(r',t)$ with maximal $t_e$ among all those
satisfying $\tilde{e} \syncausal{r'} e$. (There always exists such an
event, by finiteness of $\ND(r',t)$ and as $\tilde{e}$ is always a viable
candidate, since
$\tilde{e} \in \ND(r',t_{r'})=\ND(r',t)$, by correctness of $P$ and as $t_{r'} = t < \infty$.)
Since $r' \tsubseteq r$, any syncausal path in $r'$ ending
no later than at $t$ is also
a syncausal path in $r$, and therefore also $\tilde{e} \syncausal{r} e$.
Let $i \in I$. As $e \in \ND(r',t) = \PND{r'}(i,t)$, we obtain
$e \syncausal{r'} (i,t)$. Note that any delivery $d \ne e$ along any
syncausal path $e \syncausal{r'} (i,t)$ satisfies both
$\tilde{e} \syncausal{r'} d$ and $t_e < t_d$. Hence, by
maximality of $t_e$, we have $d \notin \ND(r',t)$.
Moreover, since $t_d \le t$, we obtain
$d \notin \ND(r')$ as well. Therefore, we obtain
$e \boundguarantee (i,t)$.
\end{proof}

As noted in the previous section, the fact that a protocol is a
finite-influence one need not dictate, in general, that the context is a
bounded-syncausal-path one.
Nonetheless, we now demonstrate a general technique
that can be employed to show that many results regarding the existence of
syncausal structures in bounded-syncausal-path contexts hold for
finite-influence protocols in arbitrary contexts as well. Furthermore,
many such results, which hold only for finite sets of agents in
bounded-syncausal-path contexts, generalize, for finite-influence protocols,
to hold for infinitely many agents as well.

\begin{claim}\label{claim:finite-influence-implies}
\corref{path-traversing-centipede-implies-broom}
(resp.\ \corref{path-traversing-centipede-implies-centibroom})
also holds when $P$ is a
finite-influence protocol, even when dropping the requirements
for a bounded syncausal path from $\tilde{e}$ to $(i,\timpl_r(i))$
and for
finiteness of $J$ (resp.\ of each $J_m$).
Under these conditions, the first part of that corollary guarantees,
for arbitrary $i \in J$ (resp.\ $i \in J_1$) of our choosing,
a finite\footnote{
Once we allow $J$ (resp,\ each $J_m$) to be infinite, though,
the first part of that corollary
requires the additional assumption that $\hat{\delta}$ is bounded
from above on $J$ (resp.\ on each $J_m$), for the guaranteed bound to be
finite. (This assumption was redundant as long as $J$ (resp.\ each $J_m)$ was
finite.)
} bound of the form
$\tilde{b}(\timpl_r(i),\sup(\hat{\delta}|_{J^2}),|E|)$
(resp.\
$\tilde{b}(\timpl_r(i),n,(\sup(\hat{\delta}|_{{J_m}^2}))_{m=1}^n,(\inf(\hat{\delta}|_{J_m \times J_{m+1}}))_{m=1}^{n-1},|E|)$),
where $E$ is a
minimal $\timpl_r(i)$-retainable set that still guarantees $i$'s response at
$\timpl_r(i)$.
\end{claim}

\begin{proof}
We present a proof of \claimref{finite-influence-implies}
with regard to the first part of
\corref{path-traversing-centipede-implies-broom}.
The proofs with regard to the second part thereof, and to both parts of
\corref{path-traversing-centipede-implies-centibroom}, are analogous.

Let $\TCRspec$ be a TCR-spec,
Let $P \in \TCR{\tilde{e},I,\delta}$ be a finite-influence protocol,
let $r \in \TRP$ and let
$J \subseteq I$ s.t.\ $\sup(\hat{\delta}|_{J^2})<\infty$.
Let $i \in J$ and let $E \in \RND(r,t)$ be minimal s.t.\
$\timpl_{\smash{r \tcap E}}(i) = \timpl_r(i)$. (Since $P$
is a finite-influence protocol, there exists such a minimal $E$,
and $|E|<\infty$.)

We first give a proof for the special case
in which $|J|<\infty$, proving a weaker statement as we allow the bound to
depend on $|J|$ for the time being.

Define $t\eqdef \timpl_r(i)$ and $r' \eqdef r \tcap E$.
As $L_{r'}(\tilde{e},(i,\timpl_{r'}(i))) \le |E| < \infty$,
\corref{path-traversing-centipede-implies-broom} may be applied
to show the existence of an $\tilde{e}$-broom in $r'$ by (no later than)
$\tilde{b}(\timpl_{r'}(i),\sup(\hat{\delta}|_{J^2}),|J|,L_{r'}(\tilde{e},(i,\timpl_{r'}(i))))\le
\tilde{b}(\timpl_r(i),\sup(\hat{\delta}|_{J^2}),|J|,|E|)$.
As the broom event occurs no later than $\timpl_{r'}(i) = t$
(and thus is in $\ND(r',t)=E$), this broom exists in $r$ as well.

We now explain why the requirement for finiteness of $|J|$ may be dropped.
Set $\tilde{b} \eqdef \tilde{b}(\timpl_r(i),\sup(\hat{\delta}|_{J^2}),|E|,|E|)<\infty$
--- the same bound as in the first part, substituting $|J|$ with $|E|$.
If some $e \in E$ constitutes an $\tilde{e}$-broom for $J$ by $\tilde{b}$,
then we are done.
Assume, by way of contradiction, that this is not the case.
Thus, for each $e \in E$, there exists $j_e \in J$ s.t.\
$e \notboundguarantee (j_e,\tilde{b})$.
The set $\{j_e\}_{e \in E}$ is of size no greater than $|E|<\infty$,
so the first part of this proof may be applied to it, yielding that there
exists $e \in E$ that constitutes
an $\tilde{e}$-broom for 
$\{j_e\}_{e \in E}$ by $\tilde{b}$ --- a contradiction.
\end{proof}

As may be expected, \claimref{finite-influence-implies} yields
generalized versions of the results of the corollaries that it
generalizes.

\begin{cor}\label{cor:finite-influence-in-many-corollaries}
In Corollaries \ref{cor:tcr-iff-sr}, \ref{cor:solvability-iff} and
\ref{cor:solvability-depends-on-strongly-connected-components},
the requirement for finiteness of agents may be relaxed to 
$\hat{\delta}$ being bounded from above on each strongly-connected
component of $G_{\delta}$, if $\mathit{SR}_{\gamma}$
and $\mathit{TCR}_{\gamma}$ are restricted to finite-influence protocols,
and if the requirement of $G_{\delta}$ containing no negative cycles is
generalized to $\delta$ being implementable.
\end{cor}

\section{Unbounded-Message-Delivery Contexts}

We conclude this chapter by revisiting the scenario with
which we opened it --- that \linebreak of the two generals who attempt to coordinate a
simultaneous
attack on a village. By now, they would probably be content with even an
approximately simultaneous attack, so we consider this
more generalized case. We use the insight this scenario has given us,
which led us to define bounded-syncausal-path contexts,
to prove the following impossibility result. This result both
generalizes \cite[Corollary~6.1.4]{book}, which shows that in a
context that exhibits unbounded message delivery in a discrete-time model,
common knowledge is unattainable, and strictly strengthens
\cite[Corollary~11.6.4]{book},
which shows that in a context that exhibits arbitrary message \linebreak loss
in a discrete-time model, $\upvarepsilon$\mbox{-}coordination based on an ND
event is unattainable.

\begin{cor}\label{cor:bounded-syncausal-path-or-zeno}
Let $\gamma$ be a bounded-syncausal-path context satisfying
$\contextbounds(i,j)=\infty$ for every $(i,j) \in \contextneighbours$\footnote{
Such a context is said, in \cite{book}, to exhibit ``unbounded message delivery''.}, let $\implspec$ be an implementation-spec s.t.\ $I \subseteq \agents$,
and let $\tilde{e}$ be any ND event.
If $G_{\delta}$ has any nontrivial (i.e.\ non-singleton) strongly-connected
component, then $\dck(\nolaterthan{0}(\tilde{e})=(\emptyset)_{i \in I}$ in
$\protocolpoints$, for every $P \in \protocols$.
\end{cor}

\begin{proof}
By revisiting the proof of \thmref{path-traversing-centipede}, we may notice
that it does not use all the properties of the timely-coordinated response
problem.
Let us define $\TCRTAG{\tilde{e},I,\delta}$, a strictly weaker\footnote{
In the sense that $\TCR{\tilde{e},I,\delta} \subsetneq
\TCRTAG{\tilde{e},I,\delta}$.
}
variant of the timely-coordinated response problem,\footnote{
While we have only defined the timely-coordinated response problem as based on
external input events, it may be readily verified that all our results regarding
it still hold if we allow it to be based on any ND event.
} as the set of all protocols
$P \in \protocols$ satisfying:
\begin{itemize}
\item In each $r \in \TRP$, either no $i \in I$ responds, or they all do, each
exactly once. In each $r \in \RP \setminus \TRP$, no $i \in I$ responds.
\item In every run $r \in \TRP$ in which all $I$ respond, it holds that
$\timpl_r \in T(\delta)$.
\end{itemize}
It may be readily verified that the result of
\thmref{path-traversing-centipede}, using the
exact same proof, still holds for $\mathit{TCR'}_{\gamma}$,
for every triggered run during which all $I$ respond.
Furthermore,
by \thmref{delta-coordination} and by
\corref{delta-common-knowledge-optimality},
the response logic defined in the latter is optimal in the sense that
a full-information protocol endowed with it solves
$\TCRTAG{\tilde{e},I,\delta}$\footnote{
It should be noted that if $G_{\delta}$ is not strongly connected,
then this is not generally true for the response logic defined in
\corref{path-traversing-centipede-optimality}, nor is it generally true
if we replace $\dck$ with $\dckg$.
}
and moreover, for every full-information
protocol $P \in \TCRTAG{\tilde{e},I,\delta}$, in each run thereof during
which all $I$ respond,
replacing the response logic of $P$ with this optimal response logic would still
yield responses of all $I$ in this run, and no response time would grow.

Assume, by way of contradiction, that there exists $P \in \protocols$ and
$i \in I$
s.t.\ $\dck(\nolaterthan{0}(\tilde{e}))_i \ne \emptyset$ in $\protocolpoints$.
Assume w.l.o.g.\ that $P$ is a full-information protocol.\footnote{
It may be readily verified that adding auxiliary variables to the state of
each agent in order to turn $P$ into a full-information protocol may only
enlarge $\dck(\nolaterthan{0}(\tilde{e})_i$.
} Furthermore, as we have not given any restrictions regarding the response
logic of $P$, assume w.l.o.g.\ that $P$ is endowed with the response logic
defined in \corref{delta-common-knowledge-optimality}.
By the above discussion, $P \in \TCRTAG{\tilde{e},I,\delta}$.
Let $(r,t) \in \dck(\nolaterthan{0}(\tilde{e})_i$.
As in the proof of \thmref{delta-coordination},
$(r,t) \subseteq \sometime(K_j(\dck(\nolaterthan{0}(\tilde{e})_j)))$ for every
$j \in I$, and thus all $I$ respond in $r$ according to $P$.
Let $(j,k) \in \distinctpairs{I}$ be a pair of distinct agents from the same
strongly-connected component of $G_{\delta}$.
As $\contextbounds \equiv \infty$,
the bound-guarantee relation is local-only.

\thmref{path-traversing-centipede}, when applied
to $\hat{\delta}$ and to paths alternating between $j$ and $k$, implies, for
every $n \in \mathbb{N}$, a syncausal path in $r$, from $\tilde{e}$,
alternating,
by $\timpl_r(j)$ (which, by the above discussion, is finite),
$n$ times back and forth between $j$ and $k$ (as the bound-guarantee
relation is local-only),
which is impossible in a bounded-syncausal-path context --- a contradiction.
\end{proof}

\begin{proof}[Alternative proof ending]
By \corref{path-traversing-centipede-implies-broom},\footnote{
The same proof we have given when deducing
\corref{path-traversing-centipede-implies-broom} from
\thmref{path-traversing-centipede} may be used to show that it, too,
holds for $\mathit{TCR'}_{\gamma}$, for every triggered run
during which all $I$ respond.
} a broom for $\{j,k\}$
exists in $r$, contradicting the fact that the bound-guarantee relation is
local-only.
\end{proof}

\begin{remark}
In \corref{bounded-syncausal-path-or-zeno}, utilizing the technique employed in
the proof of \claimref{finite-influence-implies},
the requirement for a bounded-syncausal-path context may be dropped
if $P$ is restricted to be a finite-influence protocol.
\end{remark}

It is only fitting that \corref{bounded-syncausal-path-or-zeno} concludes
the presentation of novel results in this work, as it demonstrates
the added value of our dual approach to solving the timely-coordinated response
problem, as the proof we have given thereto (regardless of the choice of ending)
utilizes elements that are unique to each of the approaches we have
taken.

\chapter{Deriving Previous Results}\label{chapter:previous}

In this chapter, we show how some previously-known results may be derived
from the novel results we have introduced in previous chapters.

\section{General Ordered and Timed Responses}
The problem of ``general ordered'' response was defined and studied by Ben-Zvi
and Moses\cite{bzm2}.
In this coordinated response problem, the relationships between the times of
the responses of finitely-many agents is dictated by a given partial order
relation $\le$ on classes of agents.
The problems of ordered response, simultaneous response and ordered joint
response, which were surveyed in the introduction to
\chapterref{syncausality-approach}, are all special cases of this problem.

It may be readily seen that the general ordered response problem is a special
case of the timely-coordinated response problem, for $\delta$ with the
canonical form
\[ (i,j) \mapsto
\begin{cases}
0 &j \le i \\
\infty &\text{otherwise.}
\end{cases} \]
It should be noted that $G_{\delta}$, sans the weights, is a DAG commonly used
to describe the dual relation $\ge$ in many applications.
Hence, for $\delta$ with the above canonical form,
every path $\apath \in \dpaths$ is a weakly-decreasing tuple.
For the above-mentioned special cases of general ordered response, this
formulation coincides with their formulations as special cases of the
timely-coordinated response problem, which we have given in the introduction
to \chapterref{syncausality-approach}.

Ben-Zvi and Moses show that the syncausal structures
underlying the general ordered response problem (those that are guaranteed to
exist in each triggered run of a solving protocol thereof,
and that may be used to define an optimal
response logic therefor) are centibrooms --- one for each linearly-ordered chain
of classes of agents, as guaranteed by \thmref{centibroom} for such a chain
in the ordered joint response problem.
The second part of \corref{path-traversing-centipede-implies-centibroom}
reduces to this result in the special case of the general ordered response
problem, and reduces to \thmref{centibroom} in the special case of the
joint ordered response problem.

Similarly, in the special case of the simultaneous response problem
(resp.\ the tightly-timed response problem,
with $\delta$ as in \remarkref{tightly-timed-is-timely}),
the second part of \corref{path-traversing-centipede-implies-broom},
together with \remarkref{broom-implies-path-traversing-centipede} and with
\thmref{path-traversing-centipede},
reduce to \thmref{broom} (resp.\ \thmref{uneven-broom}).

Last but not least, we consider the case of ordered response
(resp.\ weakly-timed response with $\delta$ as in
\remarkref{weakly-timed-is-timely}).
In this case, the above-described partial order on $I$ constitutes a
linear ordering thereof. Therefore, all paths in $G_{\delta}$ are
subpaths of the single strongly decreasing Hamiltonian path
$\apathfull \in \dpaths$.
Therefore, for every $i \in I$, all paths $\apath' \in \dpaths$ satisfying
$p'_1=i$ are subpaths
of the suffix of $\apath$ that starts with $i$, which we denote by
$\apath_{\le i}$. Thus, every corresponding $(\apath',\delta)$-traversing
$\tilde{e}$-centipede is a subcentipede of a
$(\apath_{\le i},\delta)$-traversing $\tilde{e}$-centipede by the time of $i$'s
response. This path-traversing centipede is, in turn, simply an
$\tilde{e}$-centipede by that time for $\apath_{\le i}{}^{\mathit{rev}}$.
Thus, \thmref{path-traversing-centipede} and
\corref{path-traversing-centipede-optimality} reduce to \thmref{centipede}
(resp.\ \thmref{uneven-centipede})
in this case.

As previously noted, the proof that we presented to \thmref{concise-broom}
readily generalizes to directly prove, among others, all the results surveyed
in this section.

\section{Common Knowledge and Variants}

For the duration of the section, fix a context $\gamma$, a set of runs
$R \subseteq \runs$, an event $\psi \in \pointsets$ and a set of agents
$I \subseteq \agents$.
As noted above, while all previously-studied variants of common knowledge that
are surveyed in the introduction to \chapterref{fixed-point-approach} are
defined as fixed
points of functions on $\pointsets$, this is not the case
with $\updelta$\mbox{-}common knowledge,
which we define as a fixed point of a function
on ${\pointsets}^I$. Intuitively, as noted in that chapter, this
stems from the asymmetry
of $\updelta$\mbox{-}coordination with regard to the requirements posed on the various
agents.
Given this intuition, one may expect $\delta$\mbox{-}common knowledge
to reduce, for constant $\delta$, to a non-tuple fixed point in some way.
Indeed, if $\delta$ is a constant function, then it is straightforward to
verify that $K_i(\dck(\psi))=K_i(\cap \dck(\psi))$ for every $i \in I$ and
that $\cap \dck(\psi)$ is the greatest fixed point of
$\cap f_{\psi}^{\delta}$.
We now review the previously studied non-tuple variants of common knowledge
and discuss when, and how, the above-described special case of
$\delta$\mbox{-}common knowledge for constant $\delta$ generalizes them.

When $\delta \equiv \infty$, then by definition,
$\delta$\mbox{-}coordination is equivalent to eventual coordination,
$\cap f_{\psi}^{\delta}$ is the function presented in the first part of
\thmref{eventual-coordination}, and thus
$\cap \dck(\psi)=C_I^{\sometime}(\psi)$. In addition, in this case
\thmref{delta-coordination} implies \thmref{eventual-coordination}.

Reducing the results of $\updelta$\mbox{-}common knowledge to $\varepsilon$\mbox{-}common
knowledge, for finite $\varepsilon$, is somewhat more delicate.
Assume, for the remainder of this section, that $\delta \equiv \varepsilon$
for some finite $\varepsilon \ge 0$. (Recall that for $\varepsilon \equiv 0$,
$\varepsilon$\mbox{-}coordination is equivalent to perfect coordination
and \thmref{epsilon-coordination} reduces to \thmref{perfect-coordination}.)

In general, $\varepsilon$\mbox{-}coordination is a stricter condition than
$\delta$\mbox{-}coordination.\footnote{
This stems from two main ``reasons'':
\begin{enumerate}
\item $\delta$\mbox{-}coordination is defined using $\nolaterthan{\delta(i,j)}$
rather than $\atexactly{[-\delta(j,i),\delta(i,j)]}$, which we define to mean
``at some time no earlier than $-\delta(j,i)$ from now and no later than
$\delta(i,j)$ from now''. It may be
readily seen that all the results in this work hold for such a definition as
well, as long as this replacement is performed in the definition of
$f_{\psi}^{\delta}$ as
well. The only difference is that \claimref{delta-common-knowledge-stable},
stating that $\delta$\mbox{-}common knowledge is stable,
requires also stability of $\psi$ and perfect recall in this case,
and is proven by showing that
$\nolaterthan{0}(\dck(\psi)) \le f_{\psi}^{\delta}(\nolaterthan{0}(\dck(\psi))$
and by applying the second part of \lemmaref{delta-common-knowledge}.
\item $\delta$\mbox{-}coordination is based on pairwise constraints. The results
presented in this work may be quite readily generalized to deal with
arbitrary timing constraints of various natures, such as, e.g.\ for some $J
\subseteq I$, ``For every $i \in J$ and for every $(r,t) \in \ensemble_i$,
there exists a time interval $T \subseteq \timeset$ of length at most
$\delta_J$, s.t.\ $t \in T$ and s.t.\
there exist $(t_j)_{j \in J} \in T^J$ satisfying
$(j,t_j) \in T$ for every $j \in J$''. (Whatever the timing constraints are,
the generalized definition of ${f_{\psi}^{\delta}}_i$ simply intersects on
all constraints pertaining to $i$.) Under such a generalization,
$\varepsilon$\mbox{-}coordination is equivalent to $\delta$\mbox{-}coordination,
when setting $\delta_I \equiv \varepsilon$ in the above constraint example,
and when providing no further constraints.
Furthermore, in this case the
generalization of $f_{\psi}^{\delta}$ satisfies that $\cap f_{\psi}^{\delta}$
is the function presented in the first part of \thmref{epsilon-coordination},
and thus the appropriate generalization of \thmref{delta-coordination} reduces
to \thmref{epsilon-coordination}.
\end{enumerate}
\vspace*{-1em} 
}
For a stable ensemble, though, $\delta$\mbox{-}coordination is equivalent to
$\varepsilon$\mbox{-}coordination.
If we restrict ourselves to protocols exhibiting perfect recall, then
by \corref{delta-common-knowledge-ensemble-stable},
the ensemble defined by
$\delta$\mbox{-}common knowledge is stable.
If, in addition, $\psi$ is stable, then it may be
verified that the ensemble defined by $\varepsilon$\mbox{-}common knowledge is stable
as well.\footnote{
The key observation required for showing this is that
$\nolaterthan{0}(C_I^{\varepsilon}(\psi)) \subseteq 
E_I^{\varepsilon}(\psi \cap \nolaterthan{0}(C_I^{\varepsilon}(\psi)))$,
given stability of $\psi$ and perfect recall.}
In this case, by \lemmaref{g-eq-f}, $\dck(\psi)$
is the greatest fixed point of $g_{\psi}^{\delta}$ and thus,
$\cap \dck(\psi)$ is the greatest fixed point of
$\cap g_{\psi}^{\delta}$. Analogously to the proof of \lemmaref{g-eq-f},
but in a less cumbersome way (as $\delta<\infty$),
it may be shown that in this case $C_I^{\varepsilon}(\psi)$ is the greatest
fixed point of $\cap g_{\psi}^{\delta}$ as well,
and thus $\cap \dck(\psi) = C_I^{\varepsilon}(\psi)$.\footnote{
Another way to derive this equality is by using
\cite[Exercise~11.17(d)]{book}, which
shows that, for every $i \in I$, when $\psi$ is stable and given perfect recall,
$K_i(C_I^{\varepsilon}(\psi))=K_i(\cap_{n \in \mathbb{N}} (\atexactly{\varepsilon}E_I)^n(\psi))$, to which \eqref{eq:delta-common-knowledge-nested} reduces when
$\delta \equiv \varepsilon$.
It should be noted, though, that the proof hinted to by
\cite[Exercise~11.17(d)]{book} strongly relies on a discrete modeling of time,
and breaks down in a continuous-time model, unlike the proof that we sketch
above.
}

In the absence of stability of $\psi$, or in the absence of perfect recall
(at least of the
``relevant events''), things stop working so well.
Indeed, as noted above, in such cases $\delta$\mbox{-}coordination does not
necessarily coincide with $\varepsilon$\mbox{-}coordination,
and consequently, examples may be constructed in which the ensembles defined
by $\varepsilon$\mbox{-}common knowledge and by $\delta$\mbox{-}common knowledge differ.

\chapter{Discussion and Open Questions}\label{chapter:open-questions}

\section{A Qualitative Comparison of Approaches}

Throughout this work, reasoning alternated between two approaches, which are
based on different motivations and thus were previously studied only
separately.
In \chapterref{equivalence},
though, we showed that despite the vast conceptual
difference between these two approaches, they in fact yield equivalent results
for the timely-coordinated response problem.
Nonetheless, this conceptual gap makes each approach convenient for
different purposes. Consequently, we have utilized each approach to
attack a different set of problems in \chapterref{practical}.
Moreover, in
\corref{bounded-syncausal-path-or-zeno} we concurrently harnessed both
approaches to obtain a strengthened version of a previously-known result.

The strength of the syncausal approach, similarly to that of Lamport's
asynchronous causality\cite{lamport-causality} that it generalizes,
is its constructiveness and concreteness.
These properties make it ideal for graphical visualization of runs and for
algorithm design. However, their price is the need to adapt and
specifically tailor the general results
for each model flavour, as we have done in \chapterref{practical}.

Conversely, the strength of the fixed-point approach lies in its generality and
in its high level of reasoning.
These properties make proofs that follow this approach far less cumbersome,
and far more general, due to the fact that, as we have seen,
the concept of knowledge
effectively hides the minute
details of the model in question.
The downside of this is the fairly large
gap, both between a fixed-point definition and a constructive definition,
as we have seen in \chapterref{equivalence}, and moreover --- between
a constructive knowledge-based definition and concrete implementation,
as studied by Ben-Zvi and Moses\cite{bzm1,bzm2,bzm3,bzm4} following Chandy and
Misra\cite{chandy-misra}.

Given these observations, it is not surprising that the results we
obtained using fixed-point analysis in
\chapterref{fixed-point-approach} are far more general
(and can even be further generalized, as noted in \chapterref{previous}),
than the results we obtained using syncausal analysis in
\chapterref{syncausality-approach}. However, it is
the latter that were easier for us to turn into a concrete algorithm in
\exref{algorithm}, and into a concrete condition,
in terms of required guarantees on message delivery times, for solvability of
the timely-coordinated response problem in a given context in
Corollaries \ref{cor:solvability-depends-on-strongly-connected-components}
and \ref{cor:finite-influence-in-many-corollaries}.

\section{On Generalizations}

Much of this work is based on two generalizations of known approaches.
Nontrivial generalizations tend to have a sneaky property: on one hand,
a conceptual leap is required in order to achieve them, while on the other
hand, once they are achieved, this leap, in hindsight, seems almost
obvious.

The work of Ben-Zvi and Moses on syncausal analysis\cite{bzm1,bzm2,bzm3,bzm4}
is implicitly intertwined with an insight, which holds for all of the
problems they define and analyze:
Each of these problems has one,
succinctly-describable\footnote{
One may almost claim that the description should
be linear in the number of agents.
}, syncausal structure underlying it.\footnote{
The one exception to this is the general ordered response problem, which is
treated by Ben-Zvi and Moses\cite{bzm2} as a conjunction of {\em independent}
ordered joint response problems, and thus their solution consists of a
conjunction of the syncausal structures (i.e.\ the centibrooms) underlying each
of these ordered joint response problems.
}
Indeed, when we started looking at simple two-agent cases of what would
eventually become the timely-coordinated response problem, we attempted to
find such a simple structure, or possibly only a few simple structures.

Moreover, the asymmetry inherent in the syncausality and bound
guarantee relations expresses itself in the problems defined by Ben-Zvi and
Moses, in that the timing dependency between the response times of
two agents in these problems may only be single-sided, which allows the
response of one of these agents to not depend on the response of the
other.\footnote{
Actually, their analysis, as we have seen, also allows a precise timing
dependency between two response times (i.e.\ a specified fixed time difference),
which allows them to be treated, in a sense, as a single response in the
solution of the problem. Their analysis does not, however, allow a mutual
(i.e.\ double-sided) non-precise dependency.
}
Indeed, as noted in \chapterref{syncausality-approach}, this is the main
difference, both conceptually and technically, between the response problems
defined and studied by Ben-Zvi and Moses\cite{bzm1,bzm2,bzm3,bzm4}
and the timely-coordinated response problem, which we have defined and analyzed
in this work.
As we have seen, it is the presence of mutual non-precise dependencies,
that changes the ``rules of the game'' from revolving around one, fairly simple,
syncausal structure to revolving around infinitely many, or alternatively
finitely many yet very complex\footnote{
In fact, arbitrarily complex even
for two agents.
}, syncausal structures, which may not, in general, be replaced by
simpler or fewer structures.\footnote{
Such a ``replacement'' is performed e.g.\ in the proof of Ben-Zvi
and Moses\cite{bzm1,bzm2} for the first part of \thmref{broom}.
In this work, we obtain such replacements
for some special cases of the timely-coordinated
response problem in Corollaries
\ref{cor:path-traversing-centipede-implies-broom} and
\ref{cor:path-traversing-centipede-implies-centibroom} and in
\claimref{finite-influence-implies}.
}

As we commented earlier, up until now the fixed-point approach
has only been applied to problems whose description exhibits an inherent
symmetry between the agents, in the sense that it is
invariant under permutations on the set of agents. As noted in
\chapterref{fixed-point-approach}, it is the
absence of this symmetry that ``twisted our arms'' and conceptually necessitated the 
nontrivial jump from searching for a fixed point of a scalar function to searching
for a fixed point of a vectorial function. Moreover, even after the realization
that this is the way to go, this vectorial treatment was the main
technical obstacle in our fixed-point analysis.

\section{Open Questions and Further Directions}

Throughout this work, we assume a context in which the eventual delivery of any
message is guaranteed.
In many models that do not present this behaviour, arbitrarily long syncausal
paths present themselves with zero (or very small) probability, effectively
displaying a behaviour similar to that of the class of bounded-syncausal-path
contexts, which we defined in \chapterref{practical}.
It may be interesting, therefore, to develop such probabilistic models
and to check whether the results given in 
\chapterref{practical} for bounded-syncausal-path contexts may be
applied to such models, if only to yield either probabilistic results
or impossibility results.

\enlargethispage{1em} 
\corref{tcr-iff-sr} implies that in a shared-clock, bounded-syncausal-path
context, solving an ``almost-simultaneous'' response problem is not any
more possible
than solving a simultaneous response problem. Moreover, if both are solvable,
then \eqref{response-time-bound} implies that in the worst-case scenario, the
time of the latest of the
responses in the optimal solutions to both problems is the same.\footnote{
By \corref{solvability-iff}, similarly relaxing the tight
constraints of an ordered joint, or tightly-timed,
response problem also does not make it solvable in any additional contexts,
nor does it improve the worst-case time of the latest of the responses.
}\footnote{
Nonetheless, in all cases the relaxed version may be solved
significantly faster than the original one in many runs, as illustrated in
\figref{path-traversing-centipede-within-a-broom}. (Thus, ACME's engineers
were on the right track in \exref{acme}.) It would be interesting
to give this observation a precise meaning in a probabilistic model,
perhaps in terms of average-case response time.
}
As was noted in \corref{finite-influence-in-many-corollaries}, even without the
assumption of a bounded-syncausal-path context, these
results hold as long as we make some reasonable assumptions regarding
finiteness of memory or of processing power of each agent. Furthermore,
\corref{bounded-syncausal-path-or-zeno} shows that
in the lack of any delivery guarantees, none of these problems are solvable
under such reasonable assumptions.
Nonetheless, it has been shown in \cite[Subsection~11.2.1]{book} that in models
without a shared clock, ``up-to\mbox{-}$\upvarepsilon$'' coordination may be possible
even when perfect coordination is not. It would be interesting to see
whether the machinery presented in this work may be applied, perhaps in some
extended or generalized form, to shed
new light on models in which the clock is not shared.

In \chapterref{fixed-point-approach}, we defined and analyzed
$\updelta$\mbox{-}coordination, a
generalization of several forms of coordination defined and analyzed by
Halpern and Moses\cite{halpern-moses-1990} and by Fagin et
al.\cite[Section~11.6]{book}.
In \chapterref{previous}, we noted that some special cases of
$\upvarepsilon$\mbox{-}coordination
(another form of coordination defined in \cite{halpern-moses-1990,book})
are not generalized by $\updelta$\mbox{-}coordination.
While, as we noted there, our definition of $\updelta$\mbox{-}coordination, along
with all our results regarding $\updelta$\mbox{-}common knowledge,
may be quite readily generalized to deal with additional forms of
coordination constraints, including those of $\upvarepsilon$\mbox{-}coordination,
it remains to be seen whether such generalizations are of any real added value.
In this context, it is worth to recall the difficulty we encountered in
\chapterref{fixed-point-approach}, in giving a succinct characterisation to the
ensemble defined by $\updelta$\mbox{-}common knowledge
(or by $\upvarepsilon$\mbox{-}common knowledge, for that matter) of an event along the
lines of ``the greatest $\delta$/$\varepsilon$\mbox{-}coordinated ensemble satisfying\ldots''.
This difficulty, coupled with slight differences in the properties of
$\updelta$/$\upvarepsilon$\mbox{-}common knowledge, raises the following question:
have we truly given the ``right'', ``from the book'' definition for
$\updelta$\mbox{-}common knowledge? (Similarly, have Halpern and
Moses\cite{halpern-moses-1990}, and Fagin et al.\cite[Section~11.6]{book},
given the ``right'' one for $\upvarepsilon$\mbox{-}common knowledge?) or is a similar,
yet succinctly characterisable, fixed-point definition still waiting to be phrased?

We conclude this work with a comment about fixed points. As we have seen,
fixed-point analysis of coordination is useful in a significantly broader range
of cases than previously thought.
Many systems around us, from subatomic physical
systems to astrophysical ones, and from animal societies to some stock markets,
exist in some form of equilibrium fixed point, possibly reached as a result of a
long-forgotten spontaneous symmetry
breaking. This leads us to conjecture that describing distributed algorithms as
fixed points may potentially be of much further advantage and provide us with
additional insights that are yet to be discovered.

\cleardoublepageemptyheadings


\message{\bibname}
\phantomsection 
\addcontentsline{toc}{chapter}{\bibname}

\let\oldbibliography\thebibliography
\renewcommand{\thebibliography}[1]{%
  \oldbibliography{#1}%
  \setlength{\itemsep}{1pt}%
}

\bibliographystyle{abbrv}
\bibliography{timely-coordination}

\cleardoublepageemptyheadings


\message{\appendixtocname}
\noappendicestocpagenum
\appendix
\appendixpage
\addappheadtotoc

\cleardoublepageemptyheadings


\chapter{A Continuous-Time Model}\label{app:continuous}

In this appendix, we describe a novel continuous-time model for which the
results in this work hold verbatim.
In order to avoid repetitions, we only describe
the differences between this model and the model presented in
\chapterref{discrete}.

\section{Context Parameters}

In the continuous-time model, as in the discrete-time one, we denote a context
by a tuple $\gamma=\context$.
In this case, though, $\contextgraph$ is a weighted directed graph with
positive, real or infinite, weights.
Additionally, we define the set of times as $\timeset \eqdef \mathbb{R}_{\ge0}$.
Our reasons for this somewhat unorthodox approach to modeling time
(in a continuous fashion) hopefully become apparent throughout
\chapterref{practical}.

\section{Timers}
For the duration of this section, fix an agent $i \in \agents$.
As time is continuous in our model, we wish to define when $i$ is allowed to
act. We say that $i$ is ``enabled'' (to act) at $t \in \timeset$ if either
$t=0$ or $t$ is a supremum of a
set of times, at each of which $i$ observed an event. (Intuitively, this means
that either $i$ observed an event at exactly $t$, or $i$ observed
infinitely many events, whose respective times converge to $t$ in a
monotonically-increasing fashion.)

In order to allow an agent $i \in I$ to make sure it has a chance to act at
a certain time, we introduce a new type of action, and a new type of event.
At any time $t$ at which $i$ is enabled, $i$ may set a timer for a future
time $t' \in \timeset$ s.t.\ $t'>t$. If $i$ sets such a timer, and if $i$ is not
enabled at any time
between $t$ and $t'$, then a timer ring event is observed by $i$ at $t'$,
thus enabling it at $t'$.
It should be noted that it is also possible, though somewhat less intuitive,
to not introduce timers, but rather to enable every agent at every time.

Due to the introduction of timers, the set of all possible states of the
environment becomes $S_e \eqdef 2^{\externalinputs} \times 2^{\messages \times \timeset \times \contextneighbours} \times 2^{\agents}$ (the last element is a subset
of the agents, for which timers ring at the current time),
and the set of possible actions which may be taken by $i$ at any time at which
it is enabled becomes 
$A_i \eqdef S_i \times 2^{\messages \times \{j \in I \mid (i,j) \in \contextneighbours\}} \times
\{t' \in T \mid t' > t\} \times \{\mathrm{false},\mathrm{true}\}$.

\section{Agent States}

In order to define the state of an agent ``just before''
a time $t \in \timeset$, we assume, for each $i \in \agents$, the existence of
a pseudo-limit function $\lim_i:S_i^{(-1,0)} \rightarrow S_i$ satisfying:
\begin{enumerate}
\item
$\lim_i(f)=\lim_i(g)$,
for every $f,g:(-1,0)\rightarrow S_i$ s.t.\
$f|_{(-\varepsilon,0)}=g|_{(-\varepsilon,0)}$ for some $\varepsilon > 0$.
\item
$\lim_i(f)=s$, for a constant function $f \equiv s \in S_i$.
\end{enumerate}

There are a number of other natural properties which one may expect from
$\lim_i$ (e.g.\ invariance to composition with monotone continuous functions
from $(-1,0)$ to itself, which have 0 as their limit as 0), but we do not
require any such properties for the results of this work.

For a full-information protocol, in which the state of each agent
$i \in \agents$ at any time $t \in \timeset$
uniquely determines the full details of every event observed by $i$ up until,
and including, $t$, a natural pseudo-limit function
is (infinite) union of sets of events.
In general, functions such as union, logical or, max, and min, are useful
building blocks for pseudo-limit functions for many intuitive protocols.

We define full-information protocols in this model in a similar way to
\chapterref{discrete}, although in this model an agent only sends out
messages in a full-information protocol when it is enabled.
Two full-information protocols may thus differ not only
in their response logics, but also in their timer-setting logics.
Thus, there does not necessarily always exist an isomorphism between the sets of
runs of two full-information protocols that preserves the set of ND events.
Nonetheless, it is still true that given a protocol $P$, there exists a
full-information protocol $P'$, s.t.\ there is a natural monomorphism from
$\RP$ into $\RPTAG$, which preserves both the set of ND events, and all
responses.

\section{Runs}

We are now ready to redefine the properties that a function
$r:\timeset \rightarrow S_e \times \bigtimes_{i \in I} S_i$
must satisfy in
order to constitute a legal run of a protocol
$P=((\tilde{S}_i,P_i))_{i \in I} \in \protocols$:

\begin{itemize}
\item Agent state consistency with local protocol:
Let $i \in I$ and $t \in \timeset$.
If $t > 0$,
set $s_i = \lim_i(r_i(t+\cdot))$. (This is well defined even when $t<1$, because
$\lim_i$ only depends on the values of its argument in a left neighbourhood of
0.) Intuitively,
$s_i$ is the state of $i$ ``just before'' $t$.\footnote{
Note that if $r_i|_{(t',t)}$ is constant for some
$t' \in \timeset$ s.t.\ $t'<t$, then $s_i$ equals this constant value.
} If $t = 0$, then  $s_i$ may be any of the initial states $\tilde{S}_i$.

If $i$ is not enabled at $t$ (this condition depends only on the environment
states before and at $t$), then $r_i(t)=s_i$ must hold.

If $i$ is enabled at $t$, then $r_i(t)$ must equal the first part of the output
of $P_i$, when evaluated on $s_i$ and on the events observed by $i$ at $t$.
(Once again, the other parts thereof determine the actions of $i$ at $t$.)

\item
Environment state properties:
\begin{enumerate}

\item The requirements regarding external inputs and message deliveries are
unchanged.

\item Timer events: A timer event for $i \in I$ occurs at time $t \in \timeset$
iff there exists $t' \in \timeset$ such that $t' < t$ and such that $i$ set,
at $t'$, a timer to ring at $t$, and $i$ was not enabled during $(t,t')$.
\end{enumerate}
\end{itemize}

\section{Excluding Degeneracies}

While, in the discrete-time model, a given ``partial run''
$r:\{t \in \timeset \mid t \le 1\} \rightarrow
S_e \times \bigtimes_{i \in I} S_i$ of a protocol $P \in \protocols$
may be inductively
``rolled forward'' to create a full (infinite) run
(see, e.g.\ \claimref{no-nd-between-tprime-t}), this may no longer be the
case in a continuous-time model with infinitely many agents.
Intuitively,
consider such a partial run, in which infinitely many messages
$\{m_k\}_{k=1}^\infty$
are sent to some agent $i \in \agents$ before time $1$, but are not yet
delivered by that time.
Assume, furthermore, that for every $k \in \mathbb{N}$, the message $m_k$ is
guaranteed
to be delivered no later than at $1+\frac{1}{k}$. It is not clear how to
``roll the run forward'', even for a fraction of a time unit.
Similarly, if the delivery guarantee for each $m_k$ is at $2+\frac{2}{k}$,
then it is not clear that it is possible to 
roll the run forward while avoiding
early deliveries, as required in \claimref{no-nd-between-tprime-t}.
Indeed, if it is not possible to do so,
then in any run that is indistinguishable
from this partial run up to time $1$, it is {\em deterministic} that some early
delivery \linebreak takes place between times $1$ and $2$, effectively voiding the
non-determinism of some early deliveries, possibly allowing them to be predicted
before they occur.

In order to avoid degeneracies such as those described above,
and thus sufficiently maintain the
non-deterministic nature of the events that we call ``ND events'',
we axiomatically make the following assumptions regarding the richness of the
set of runs of any protocol $P \in \protocols$:\footnote{
As is discussed in \sectionref{adapting-lemma-3}, both of these assumptions
may be shown to hold for the discrete-time model presented in
\chapterref{discrete}.
Furthermore, if $\inf(\contextbounds)>0$,
then an inductive argument,
rolling a run forward $\inf(\contextbounds)$ time units at a time, may
be used to show that these assumptions also hold for protocols implemented 
using certain ``nice'' pseudo-limit functions such as the union pseudo-limit
function described above, and for arbitrary protocols if
$|\agents|<\infty$.}
\begin{itemize}
\item No foresight: 
For every $r \in \RP$ and for every $t,d \in \timeset$,
there exists a run $r' \in \RP$, satisfying:
\begin{enumerate}
\item $r'|_{[0,t]} = r|_{[0,t]}$.
\item No external inputs are triggered in $r'$ after $t$.
\item Any message delivered {\em early} in $r'$ after $t$ is delivered no less
than $d$ time units after it is sent.
\end{enumerate}
\item No extrasensory perception:
For every $r \in \RP$, for every $t \in \timeset$ and
for every $i \in \agents$, it holds that
$\PND{r}(i,t) \in \RND(r,t)$.\footnote{
See \sectionref{adapting-lemma-3} for the definitions of $\mathit{RND}$
and of $\mathit{PND}$.
(We allow ourselves to state this assumption in terms of syncausality,
as we only utilize it in our syncausal analysis.)
}
\end{itemize}

The two above assumptions, regarding non-determinism of future events and
independence of past events, respectively, imply that agents may not predict
the occurrence of certain ND events.
These assumptions suffice for most of our analysis.
For some arguments, though, we require some guarantee that agents
may not predict the {\em absence} of ND events.
The following assumption, regarding non-determinism and independence of
present events, provides such a guarantee and, if $\RP \ne \emptyset$,
complements the above assumptions in the strongest way possible in some sense.
\begin{itemize}
\item 
Eternal vigilance:
For every run $r \in \RP$, for every $t \in \timeset$, for every set
$E \subseteq \externalinputs$
of external inputs that are not triggered in $r$ before $t$ and for every
set of $M$ of potential early deliveries at $t$ in $r$ (i.e.\
messages sent
before $t$, not delivered before $t$, and with a delivery guarantee
greater than $t$),
there exists a run $r' \in \RP$, satisfying:
\begin{enumerate}
\item $r'|_{[0,t)} = r|_{[0,t)}$.
\item $r'_e(t) = (E,M)$.
\end{enumerate}
\end{itemize}
Although it may be readily verified that the above assumption holds in the
discrete-time model presented in \chapterref{discrete},
this assumption is restrictive for a continuous-time model, as
e.g.\ it does not hold for some naturally-occurring models, such as models with
minimum bounds on delivery times. Moreover, this assumption also hinders the
possibility of capturing a discrete-time model using our continuous-time model.
(See the next section for more details.)
For these reasons, we replace this assumption with the following, weaker
assumption, which stems from combining the ``eternal vigilance'' and
``no foresight'' assumptions when $\RP \ne \emptyset$:
\begin{itemize}
\item
Stand-alone external inputs:
For every external input $\tilde{e} \in \externalinputs$, there exists
a run $r \in \TRP$, in which no ND events other than $\tilde{e}$ occur
before or at $t_{\tilde{e}}$.
\end{itemize}

\section{Modeling Discrete Time}
Now that we finished describing this model, it should be noted
that discrete-time models, such as the one presented in \chapterref{discrete},
may be captured by this model. As an example,
an integral-time model may be modeled by setting $\contextbounds(i,j)$ to integral (or
infinite) values for every $(i,j) \in \contextneighbours$, by forcing the
environment to perform ND events only at integral times (i.e.\ removing from
$\RP$ any runs in which any ND events occur at non-integral times),
and by allowing timers to be set for integral times only. (Or, alternatively, by
dropping timers altogether, and enabling every agent at every integral time.)
We consider environment constraints, such as ``all ND events occur at integral
times'', or ``any message may be delivered, at the earliest, $\varepsilon$
after it was sent'', as integral parts of the model (just as the delivery
bounds are). Care should be taken to make sure that such constraints
do not interfere with the assumptions of the previous section.

\cleardoublepageemptyheadings

\end{document}